\documentclass[11pt]{article}%

\usepackage{amsmath, amsthm}

\usepackage{bm}
\usepackage{natbib}

\usepackage[width=1.0\textwidth]{caption}
\usepackage{amssymb}
\usepackage{psfrag,graphicx}
\usepackage{epstopdf}

\usepackage{subcaption}

\usepackage{caption}
\usepackage{setspace}

\allowdisplaybreaks

\newcommand{\rd}{\mathrm{d}}

\newcommand{\tP}{\tilde{P}}

\usepackage[right=2.5cm,left=2.5cm,top=1.5cm,bottom=2.5cm]{geometry}
\newtheorem{theorem}{Theorem}

\newtheorem{corollary}{Corollary}

\newtheorem{lemma}{Lemma}

\newtheorem{proposition}{Proposition}
\newtheorem{remark}{Remark}

\newtheorem{assumption}{Assumption}

\begin{document}

\title{Efficient implementation of Markov chain Monte Carlo when using an unbiased
likelihood estimator}

%
%
%
\title{Efficient implementation of Markov chain Monte Carlo when using an unbiased
likelihood estimator}
\author{Arnaud Doucet\\Department of Statistics,\\University of Oxford,\\\texttt{doucet@stats.ox.ac.uk}
\and Michael Pitt\\Department of Economics, \\University of Warwick\\\texttt{m.pitt@warwick.ac.uk}
\and George Deligiannidis\\Department of Statistics,\\University of Oxford\\\texttt{deligian@stats.ox.ac.uk}
\and Robert Kohn\\Department of Economics,\\University of New South Wales\\\texttt{r.kohn@unsw.edu.au}}

\maketitle
\begin{abstract}
When an unbiased estimator of the likelihood is used within a
Metropolis--Hastings chain, it is necessary to trade off the number of Monte
Carlo samples used to construct this estimator against the asymptotic variances of averages computed under this chain.
Many Monte Carlo samples will typically result in Metropolis--Hastings averages with lower asymptotic variances
than the corresponding Metropolis--Hastings averages using fewer samples. However, the computing time required to construct the likelihood
estimator increases with the number of Monte Carlo samples. Under the assumption that the
distribution of the additive noise introduced by the log-likelihood estimator
is Gaussian with variance inversely proportional to the number of Monte Carlo
samples and independent of the parameter value at which it is evaluated, we
provide guidelines on the number of samples to select. We demonstrate our results by considering a
stochastic volatility model applied to stock index returns.
\end{abstract}
\noindent\textit{Keywords:}
Intractable likelihood, Metropolis-Hastings algorithm, Particle filter, Sequential Monte Carlo, State-space model.




\section{Introduction\label{section:introduction}}
The use of unbiased estimators within the Metropolis--Hastings algorithm was
initiated by \cite{lin:2000}, with a surge of interest in these ideas since
their introduction in Bayesian statistics by \cite{beaumont2003estimation}. In a Bayesian context, an unbiased likelihood estimator is commonly constructed using
importance sampling as in \cite{beaumont2003estimation} or particle filters as
in \cite{andrieu:doucet:holenstein:2010}. \cite{andrieu2009pseudo} call this method the pseudo-marginal algorithm, and establish some of its theoretical properties.

Apart from the choice of proposals inherent to any Metropolis--Hastings algorithm, the main practical issue with the pseudo-marginal algorithm is the choice of the number, $N$, of Monte Carlo samples or particles used to estimate the likelihood. For any fixed $N$, the transition kernel of the pseudo-marginal algorithm leaves the posterior distribution of interest invariant. Using many Monte Carlo samples usually results in pseudo-marginal averages with asymptotic variances lower than the corresponding averages using fewer samples, as established 
by \citet{AV14}
for likelihood estimators based on importance sampling. Empirical evidence suggests this result also holds when the likelihood is estimated by particle filters. However, the computing cost of constructing the likelihood estimator increases with $N$. We aim to select $N$ so as to minimize the computational resources necessary to achieve a specified asymptotic variance for a particular pseudo-marginal average. This quantity, which is referred to as the computing time, is typically proportional to $N$ times the asymptotic variance of this average, which is itself a function of $N$. Assuming that the distribution of the additive noise introduced by the log-likelihood estimator
is Gaussian, with a variance inversely proportional to $N$ and independent of
the parameter value at which it is evaluated, this minimization was carried
out in \cite{PittSilvaGiordaniKohn(12)} and in 
\citet{Sherlock2013efficiency}.
%
However, \cite{PittSilvaGiordaniKohn(12)}
assume that the Metropolis--Hastings proposal is the posterior density,
whereas \citet{Sherlock2013efficiency} relax the Gaussian noise assumption, but restrict
themselves to an isotropic normal random walk proposal and assume that the
posterior density factorizes into $d$ independent and identically distributed
components and $d\rightarrow\infty$.

Our article addresses a similar problem but considers general proposal and
target densities and relaxes the Gaussian noise assumption. In this more
general setting, we cannot minimize the computing time, and instead
minimize explicit upper bounds on it. Quantitative results are presented
under a Gaussian assumption. In this scenario, our guidelines are that $N$ should be chosen such that the standard deviation of the
log-likelihood estimator should be around $1.0$ when the Metropolis--Hastings
algorithm using the exact likelihood is efficient and around $1.7$ when it is
inefficient. In most practical scenarios, the
efficiency of the Metropolis--Hastings algorithm using the exact likelihood is unknown as
it cannot be implemented. In these cases, our results suggest selecting a standard deviation around $1.2$.

\section{ Metropolis--Hastings method using an estimated
likelihood\label{SS: sim likelihood}}

We briefly review how an unbiased likelihood estimator may be used within a
Metropolis--Hastings scheme in a Bayesian context. Let $y\in\mathsf{Y}$ be the
observations and $\theta\in\Theta\subseteq\mathbb{R}^{d}$ the parameters of
interest. The likelihood of the observations is denoted by $p(y\mid\theta)$
and the prior for $\theta$ admits a density $p(\theta)$ with respect to
Lebesgue measure so the posterior density of interest is $\pi
(\theta)\propto p(y\mid\theta)p(\theta)$. We slightly abuse notation by using the same symbols for distributions and densities.

The Metropolis--Hastings scheme to sample from $\pi$ simulates a Markov chain
according to the transition kernel%
\begin{equation}
Q_{\textsc{ex}}\left(  \theta,\mathrm{d}\vartheta\right)  =q\left(
\theta,\vartheta\right)  \alpha_{\textsc{ex}}(\theta,\vartheta
)\mathrm{d}\vartheta+\left\{  1-\varrho_{\textsc{ex}}\left(
\theta\right)  \right\}  \delta_{\theta}\left( \mathrm{d}\vartheta\right),
\label{eq:Q_EX}%
\end{equation}
where
\begin{equation}
\alpha_{\textsc{ex}}(\theta,\vartheta)=\min\{1,r_{\textsc{ex}%
}(\theta,\vartheta)\}\text{, \ \ }\varrho_{\textsc{ex}}\left(
\theta\right)  =\int q\left(  \theta,\vartheta\right)  \alpha
_{\textsc{ex}}(\theta,\vartheta)\mathrm{d}\vartheta,\label{eq: accept EX}%
\end{equation}
with $r_{\textsc{ex}}(\theta,\vartheta)=\pi(\vartheta)q\left(
\vartheta,\theta\right)  /\left\{  \pi(\theta)q\left(  \theta,\vartheta
\right)  \right\}$. This Markov chain cannot be simulated if $p(y\mid\theta)$ is intractable.

Assume $p(y\mid\theta)$ is intractable, but we have access to a non-negative
unbiased estimator \linebreak$\widehat{p}(y\mid\theta,U)$ of $p(y\mid\theta)$, where
$U\sim m\left(  \cdot\right)  $ represents all the auxiliary random variables
used to obtain this estimator. In this case, we introduce the joint density
$\overline{\pi}(\theta,u)$ on $\Theta\times\mathcal{U}$, where%
\begin{equation}
\overline{\pi}(\theta,u)=\pi(\theta)m(u)\widehat{p}(y\mid\theta,u)/p(y\mid
\theta). \label{eq:norm_jd}%
\end{equation}
This joint density admits the correct marginal density $\pi(\theta)$, because
$\widehat{p}(y\mid\theta,U)$ is unbiased. The pseudo-marginal algorithm is a
Metropolis--Hastings scheme targeting (\ref{eq:norm_jd}) with proposal density
$q\left(  \theta,\cdot\right)  m\left(  \cdot\right)  $, yielding the acceptance
probability
\begin{equation}
\min\left\{  1,\frac{\widehat{p}(y\mid\vartheta,v)p\left(  \vartheta\right)
q\left(  \vartheta,\theta\right)  }{\widehat{p}(y\mid\theta,u)p\left(
\theta\right)  q\left(  \theta,\vartheta\right)  }\right\}  =\min\left\{
1,\frac{\widehat{p}(y\mid\vartheta,v)/p(y\mid\vartheta)}{\widehat{p}%
(y\mid\theta,u)/p(y\mid\theta)}r_{\textsc{ex}}(\theta,\vartheta
)\right\},  \label{eq:jointutheta}%
\end{equation}
for a proposal $\left(  \vartheta,v\right)  $. In practice, we only record
$\left\{  \theta,\log\widehat{p}(y\mid\theta,u)\right\}  $ instead of
$\left\{  \theta,u\right\}  $. We follow \cite{andrieu2009pseudo} and
\cite{PittSilvaGiordaniKohn(12)} and analyze this scheme using additive noise,
$Z=\log\widehat{p}(y\mid\theta,U)-\log p(y\mid\theta)=\psi(\theta,U)$, in the
log-likelihood estimator, rather than $U$. In this parameterization, the
target density on $\Theta\times\mathbb{R}$ becomes%
\begin{equation}
\overline{\pi}(\theta,z)=\pi(\theta)\exp\left(  z\right)  g(z\mid\theta),
\label{eq:jointztheta}%
\end{equation}
where $g(z\mid\theta)$ is the density of $Z$ when $U\sim m(\cdot)$ and the
transformation $Z=\psi(\theta,U)$ is applied.

To sample from $\overline{\pi}(\theta,z)$, we could use the scheme previously
described to sample from $\overline{\pi}(\theta,u)$ and then set
$z=\psi(\theta,u)$. We can equivalently use the transition kernel%
\begin{align}
Q\left\{  \left(  \theta,z\right)  ,\left(  \mathrm{d}\vartheta,\mathrm{d}%
w\right)  \right\}   &  =q\left(  \theta,\vartheta\right)  g(w\mid
\vartheta)\alpha_{Q}\left\{  \left(  \theta,z\right)  ,\left(  \vartheta
,w\right)  \right\}  \mathrm{d}\vartheta\mathrm{d}%
w\label{eq:transitionkernelQ}\\
&  \text{ \ \ \ \ \ \ \ \ \ \ \ \ \ \ \ \ \ \ \ \ \ }+\left\{  1-\varrho
_{Q}\left(  \theta,z\right)  \right\}  \delta_{\left(  \theta,z\right)
}\left(  \mathrm{d}\vartheta,\mathrm{d}w\right)  ,\nonumber
\end{align}
where
\begin{equation}
\alpha_{Q}\left\{  \left(  \theta,z\right)  ,\left(  \vartheta,w\right)
\right\}  =\min\{1,\exp(w-z)\text{ }r_{\textsc{ex}}(\theta,\vartheta)\}
\label{eq:acceptanceprobaQ}%
\end{equation}
is (\ref{eq:jointutheta}) expressed in the new parameterization. Henceforth,
we make the following assumption.

\begin{assumption}
\label{assumption:noiseorthogonal}The noise density is
independent of $\theta$ and is denoted by $g\left(  z\right)  $.
\end{assumption}

Under this assumption, the target density (\ref{eq:jointztheta}) factorizes as
$\pi(\theta)\pi_{\textsc{z}}\left(  z\right)  $, where%
\begin{equation}
\pi_{\textsc{z}}\left(  z\right)  =\exp\left(  z\right)  g(z). \label{eq:targetinz}%
\end{equation}
Assumption \ref{assumption:noiseorthogonal} allows us to analyze in detail
the performance of the pseudo-marginal algorithm. This simplifying
assumption is not satisfied in practical scenarios. However, in the stationary
regime, we are concerned with the noise density at values of the parameter
which arise from the target density $\pi\left(\theta\right)$ and the marginal density
of the proposals at stationarity $\int \pi\left(\mathrm{d}\vartheta\right)  q\left(\vartheta, \theta \right) $. If the noise density does
not vary significantly in regions of high probability mass of these densities,
then we expect this assumption to be a reasonable approximation. In Section
\ref{sect:applications}, we examine experimentally how the noise density
varies against draws from  $\pi\left(\theta\right)$ and $\int \pi\left(\mathrm{d}\vartheta\right)  q\left(\vartheta, \theta \right) $.

\section{Main results}

\subsection{Outline\label{sec:Qstar}}

This section presents the main contributions of the paper. All the proofs are
in Appendix~1 and in the Supplementary Material. We minimize upper bounds
on the computing time of the pseudo-marginal algorithm, as discussed in Section \ref{section:introduction}. This requires establishing upper bounds on the asymptotic variance of an ergodic average under the kernel $Q$ given in
(\ref{eq:transitionkernelQ}). To
obtain these bounds, we introduce a new Markov kernel $Q^{\ast}$, where
\begin{align}
Q^{\ast}\left\{  \left(  \theta,z\right)  ,\left(  \mathrm{d}\vartheta
,\mathrm{d}w\right)  \right\}   &  =q\left(  \theta,\vartheta\right)
g(w)\alpha_{Q^{\ast}}\left\{  \left(  \theta,z\right)  ,\left(  \vartheta
,w\right)  \right\}  \mathrm{d}\vartheta\mathrm{d}w
\label{eq:MarkovchainQstar}\\
&  \text{ \ \ \ \ \ \ \ \ \ \ \ \ }+\left\{  1-\varrho_{\textsc{ex}%
}\left(  \theta\right)  \varrho_{\textsc{z}}\left(  z\right)  \right\}
\delta_{\left(  \theta,z\right)  }\left(  \mathrm{d}\vartheta,\mathrm{d}%
w\right)  ,\nonumber
\end{align}
and
\begin{align}
\alpha_{Q^{\ast}}\left\{  \left(  \theta,z\right)  ,\left(  \vartheta
,w\right)  \right\}   &  =\alpha_{\textsc{ex}}(\theta,\vartheta
)\alpha_{\textsc{z}}(z,w)\text{, \ \ }\alpha_{\textsc{z}}\left(
z,w\right)  =\min\{1,\exp(w-z)\},\label{eq:acceptanceprobabilityQ*}\\
\varrho_{\textsc{z}}\left(  z\right)   &  =\int g\left(  w\right)
\alpha_{\textsc{z}}\left(  z,w\right)  \mathrm{d}w.
\label{eq:acceptanceprobaz}%
\end{align}

As $Q$ and $Q^{\ast}$ are reversible with respect to $\overline{\pi}$ and the
acceptance probability (\ref{eq:acceptanceprobabilityQ*}) is always smaller
than (\ref{eq:acceptanceprobaQ}), an application of the theorem in
\cite{peskun1973optimum} ensures that the variance of an ergodic average under
$Q^{\ast}$ is greater than or equal to the variance under $Q$. We obtain an exact expression
for the variance under the bounding kernel $Q^{\ast}$ and simpler upper bounds by
exploiting a non-standard representation of this variance, the factor form of the
acceptance probability (\ref{eq:acceptanceprobabilityQ*}) and the\ spectral
properties of an auxiliary Markov kernel.

\subsection{Inefficiency of Metropolis--Hastings type chains
\label{sec:knownlik}}

This section recalls and establishes various results on the integrated
autocorrelation time of Markov chains, henceforth referred to as the
inefficiency. In particular, we present a novel representation of the
inefficiency of Metropolis--Hastings type chains, which is the basic component
of the proof of our main result.

Consider a Markov kernel $\Pi$ on the measurable space $\left(  \mathsf{X}%
,\mathcal{X}\right)  =\left\{  \mathbb{R}^{n},\mathcal{B}\left(
\mathbb{R}^{n}\right)  \right\}  $, where $\mathcal{B}\left(  \mathbb{R}%
^{n}\right)  $ is the Borel $\sigma$-algebra on $\mathbb{R}^{n}$. For any measurable real-valued
function $f$, measurable set $A$ and probability measure $\mu$, we use the
standard notation: $\mu\left(  f\right)  =\int_{\mathsf{X}}\mu\left(
\mathrm{d}x\right)  f\left(  x\right)  $, $\mu\left(  A\right)  =\mu\left\{
\mathbb{I}_{A}\left(  \cdot\right)  \right\}  ,$ $\Pi f\left(  x\right)
=\int_{\mathsf{X}}\Pi\left(  x,\mathrm{d}y\right)  f\left(  y\right)  $ and
for $n\geq2$, $\Pi^{n}\left(  x,\mathrm{d}y\right)  =\int_{\mathsf{X}}%
\Pi^{n-1}\left(  x,\mathrm{d}z\right)  \Pi\left(  z,\mathrm{d}y\right)  $,
with $\Pi^{1}=\Pi$. We introduce the Hilbert spaces
\[
L^{2}\left(  \mathsf{X},\mu\right)  =\left\{  f:\mathsf{X\rightarrow
}\mathbb{R}\text{\ : }\mu\left(  f^{2}\right)  <\infty\right\}  ,\text{ }%
L_{0}^{2}\left(  \mathsf{X},\mu\right)  =\left\{  f:\mathsf{X\rightarrow
}\mathbb{R}\text{\ : }\mu\left(  f\right)  =0,\text{ }\mu\left(  f^{2}\right)
<\infty\right\}
\]
equipped with the inner product $\left\langle f,g\right\rangle _{\mu}=\int
f\left(  x\right)  g\left(  x\right)  \mu\left(  \mathrm{d}x\right)  $. A
$\mu$-invariant and $\psi$-irreducible Markov chain is said to be
ergodic; see \cite{Tierney94} for the definition of $\psi$-irreducibility.
The next result follows directly from \cite{kipnis1986central} and Theorem 4
and Corollary 6 in \cite{rosenthalCLT2007}.

\begin{proposition}
\label{prop:kipnisvaradhan}Suppose $\Pi$ is a $\mu$-reversible and ergodic
Markov kernel. Let $(X_{i})_{i\geqslant 1}$ be a stationary Markov chain evolving according to
$\Pi$ and let $h\in L^{2}\left(  \mathsf{X},\mu \right)  $ be such that $\mu \left(  \bar{h}^{2}\right)>0$ where $\bar{h}=h-\mu \left(  h\right)$. Write
$\phi _{n}\left( h,\Pi \right) =\left \langle \bar{h},\Pi^{n}\bar{h}\right \rangle _{\mu}%
/\mu \left(  \bar{h}^{2}\right)  $ for the autocorrelation at lag $n\geq0$ of $\left \{  h\left(  X_{i}\right)  \right \}  _{i\geq1}$ and $\textsc{\protect\small IF}(h,\Pi)=1+2%
{\textstyle \sum \nolimits_{n=1}^{\infty}}
\phi _{n}\left( h,\Pi \right) $ for the associated inefficiency. Then,
\begin{enumerate}
\item[(i)] there exists a probability measure $e\left( h,\Pi \right)$ on $[-1,1)$ such
that the autocorrelation and inefficiency satisfy the spectral
representations
\begin{equation}
\phi_{n}\left( h,\Pi \right) =%
{\textstyle \int \nolimits_{-1}^{1}}
\lambda^{n}e\left( h,\Pi \right) \left( \mathrm{d}\lambda \right) ,\quad \textsc{\protect\small IF}(h,\Pi)=%
{\textstyle \int \nolimits_{-1}^{1}}
(1+\lambda)(1-\lambda)^{-1}e\left( h,\Pi \right) \left( \mathrm{d}\lambda \right); \label{eq:spectral}%
\end{equation}
\item[(ii)] if $\textsc{\protect\small IF}({h}, {\Pi})<\infty$, then as $n\rightarrow \infty$%
\begin{equation}
n^{-1/2}\sum_{i=1}^{n}\left \{  h(X_{i})-\mu \left(  h\right)  \right \}
{\longrightarrow } \mathcal{N}\left \{  0;\mu \left(  \bar{h}^{2}\right)
\textsc{\protect\small IF}(h,\Pi)\right \},  \label{eq:CLTKipnisVaradhan}%
\end{equation}
in distribution, where $\mathcal{N}\left(  a;b^2\right)  $ denotes the normal distribution with
mean $a$ and variance $b^2$.
\end{enumerate}
\end{proposition}

When estimating $\mu\left(  h\right)  $, equation (\ref{eq:CLTKipnisVaradhan})
implies that we need approximately $n$ $\textsc{\protect\small IF}(h,\Pi)$ samples from the
Markov chain $\left(  X_{i}\right)  _{i\geq1}$ to obtain an estimator of the
same precision as an average of $n$ independent draws from $\mu$.

We consider henceforth a $\mu$-reversible kernel given by%
\[
P\left(  x,\mathrm{d}y\right)  =q\left(  x,\mathrm{d}y\right)  \alpha\left(
x,y\right)  +\left\{  1-\varrho\left(  x\right)  \right\}  \delta_{x}\left(
\mathrm{d}y\right)  \text{, \ \ }\varrho\left(  x\right)  =\int q\left(
x,\mathrm{d}y\right)  \alpha\left(  x,y\right)  ,
\]
where the proposal kernel is selected such that $q(x,\{  x\})=0$,  
$\alpha\left(  x,y\right)  $ is the acceptance probability and we assume there
does not exist an $x$ such that $\mu\left(  \left\{  x\right\}  \right)  =1$.
We refer to $P$ as a Metropolis--Hastings type kernel since it is structurally
similar to the Metropolis--Hastings kernel, but we do not require
$\alpha\left(  x,y\right)  $ to be the
Metropolis--Hastings acceptance probability. This generalization is required
when studying the kernel $Q^{\ast}$ as the acceptance probability
$\alpha_{Q^{\ast}}\left\{  \left(  \theta,z\right)  ,\left(  \vartheta
,w\right)  \right\}  $ in (\ref{eq:acceptanceprobabilityQ*}) is not the
Metropolis--Hastings acceptance probability.

Let $(  X_{i})  _{i\geq1}$ be a Markov chain evolving according to
$P$. We now establish a non-standard expression for $\textsc{\protect\small IF}(h,P)$ derived
from the associated jump\ chain representation $(\widetilde{X}_{i},\tau
_{i})_{i\geq1}$ of $\left(  X_{i}\right)  _{i\geq1}$. In this representation,
$(\widetilde{X}_{i})_{i\geq1}$ corresponds to the sequence of accepted
proposals and $(\tau_{i})_{i\geq1}$ the associated sojourn times, that is
$\widetilde{X}_{1}=X_{1}=\cdots=X_{\tau_{1}},$ $\widetilde{X}_{2}=X_{\tau
_{1}+1}=\cdots=X_{\tau_{1}+\tau_{2}}$ etc., with $\widetilde{X}_{i+1}%
\neq\widetilde{X}_{i}$. Some properties of this jump chain are
now stated; see Lemma 1 in \cite{douc2011vanilla}.

\begin{lemma}
\label{lemm:jumpchain}Let $P$ be $\psi$-irreducible. Then $\varrho \left(
x\right)  >0$ for any $x\in \mathsf{X}$ and $(\widetilde{X}%
_{i},\tau_{i})_{i\geq1}$ is a Markov chain with a $\overline{\mu}$-reversible
transition kernel $\overline{P}$, where%
\begin{equation}
\overline{P}\left \{  \left(  x,\tau \right)  ,\left(  \mathrm{d}%
y,\zeta \right)  \right \}  =\widetilde{P}\left(  x,\mathrm{d}y\right)  G\left \{
\zeta;\varrho \left(  y\right)  \right \},  \quad \overline{\mu
}\left(  \mathrm{d}x,\tau \right)  =\widetilde{\mu}\left(  \mathrm{d}x\right)
G\left \{  \tau;\varrho \left(  x\right)  \right \}  ,\text{ }
\label{eq:transitionjumpjoint}%
\end{equation}
with
\begin{equation}
\widetilde{P}\left(  x,\mathrm{d}y\right)  =\frac{q(x,\mathrm{d}y)\alpha(x,y)}%
{\varrho \left(  x\right)  },\quad \widetilde{\mu}\left(
\mathrm{d}x\right)  =\frac{\mu \left(  \mathrm{d}x\right)  \varrho \left(
x\right)  }{\mu \left(  \varrho \right)  }, \label{eq:transitionjump}%
\end{equation}
and $G\left(  \cdot;\upsilon \right)  $ denotes the geometric distribution with
parameter $\upsilon$.
\end{lemma}

The next proposition gives the relationship between $\textsc{\protect\small IF}(h,P)$ and
$\textsc{\protect\small IF}({h/\varrho},{\widetilde{P}})$.

\begin{proposition}
\label{prop:IACTequality}
Assume that $P$ and $\widetilde{P}$ are ergodic, that $h\in L_{0}^{2}\left(\mathsf{X},\mu \right)$ and that $\textsc{\protect\small IF}(h,P)<\infty$.
Then $h/\varrho \in L_{0}^2(\mathsf{X},\tilde{\mu})$,
\begin{equation}
\mu \left( h^{2}\right) \left \{  1+\textsc{\protect\small IF}(h,P)\right \}  =\mu \left(
\varrho \right)  \widetilde{\mu}\left(  h^{2}/\varrho^{2}\right)  \left \{
1+\textsc{\protect\small IF}({h/\varrho},{\widetilde{P}})\right \},  \label{eq:equalityIACT}%
\end{equation}
and $\textsc{\protect\small IF}({h/\varrho},{\widetilde{P}})\leq \textsc{\protect\small IF}(h,P)$.
\end{proposition}

Lemma \ref{lemm:jumpchain} and Proposition \ref{prop:IACTequality} are used in Section \ref{sec:bound1} to establish a
representation of the inefficiency for the kernel $P=Q^{\ast}$.

We conclude this section by establishing some results on the positivity of the Metropolis--Hastings kernel and
its associated jump kernel. Recall that a $\mu$-invariant Markov kernel $\Pi$ is positive if $\left\langle \Pi
h,h\right\rangle _{\mu}\geq0$ for any $h\in L^{2}\left(  \mathsf{X}%
,\mu\right)  $. If $\Pi$ is reversible, then positivity is equivalent to
$e\left(  h,\Pi\right)  \left(  \left[  0,1\right)  \right)  =1$ for all $h\in
L^{2}\left(  \mathsf{X},\mu\right)  $, where $e\left(  h,\Pi\right)  $ is the
spectral measure, and it implies that $\textsc{\protect\small IF}(h,\Pi)\geq1$; see, for
example, \cite{geyer1992practical}. The positivity of the jump kernel  $\widetilde{P}$ associated
with a Metropolis-Hastings kernel $P$ is useful here as several bounds on the inefficiency established subsequently
require the spectral measure of $\widetilde{P}$ to be supported on $ \left[  0,1\right) $. We now give sufficient conditions ensuring
this property by extending Lemma 3.1 of \cite{baxendale2005}. This complements results of \citet{rudolf2013}.

\begin{proposition}
\label{prop:positivityMHjumpMH}Assume $\alpha \left(  x,y\right)  $ is
the Metropolis--Hastings acceptance probability and $\mu \left(  \mathrm{d}%
x\right)  =\mu \left(  x\right)  \mathrm{d}x$. If $P$ is $\psi$-irreducible,
then  $\widetilde{P}$ and $P$ are both positive if one of the following two conditions
is satisfied:
\begin{enumerate}
\item[(i)] $q(x, \mathrm{d}y)=q(x,y)\mathrm{d}y$
is a $\nu$-reversible kernel with $\nu(\mathrm{d}x)=\nu(x)\mathrm{d}x$, $\mu$ is absolutely continuous with respect to $\nu$, and there
exists $r:\mathsf{X\times Z\rightarrow}\mathbb{R}^{+}$ such that $\nu \left(
x\right)  q(  x,y)  =\int r(  x,z)  r(  y,z)
\chi \left(  \mathrm{d}z\right)  $, where $\chi$ is a measure on $\mathsf{Z};$
\item[(ii)] $q(  x,\mathrm{d}y)  =q(  x,y)  \mathrm{d}y$
and there exists $s:\mathsf{X\times Z\rightarrow}\mathbb{R}^{+}$ such that
$q(  x,y)  =\int s(  x,z)  s(  y,z)\chi(\mathrm{d}z)$,  where $\chi$ is a measure on $\mathsf{Z.}$
\end{enumerate}
\end{proposition}

\begin{remark}
Condition (i) is satisfied for an independent proposal  $q\left(  x,y\right)  =\nu \left(  y\right)$ by 
taking $\mathsf{Z=}\left \{  1\right \}  $, $\chi \left(  \mathrm{d}z\right)
=\delta_{1}\left(  \mathrm{d}z\right)  $ and $r\left(  x,1\right)  =\nu(
x).$  It is also satisfied for autoregressive positively correlated proposals
with normal or Student-t innovations. Condition (ii) holds if $q\left(
x,y\right)  $ is a symmetric random walk proposal whose increments are
multivariate normal or Student-t.
\end{remark}

\subsection{Inefficiency of the bounding chain\label{sec:bound1}}

This section applies the results of Section \ref{sec:knownlik} to establish an
exact expression for $\textsc{\protect\small IF}(h,Q^{\ast})$. The next lemma shows that
$\textsc{\protect\small IF}(h,Q^{\ast})$ is an upper bound on $\textsc{\protect\small IF}(h,Q)$.

\begin{lemma}
\label{Lemma:Peskun}The kernel $Q^{\ast}$ is $\overline{\pi}$-reversible and $\textsc{\protect\small IF}(h,Q)\leq \textsc{\protect\small IF}(h,Q^{\ast})$
for any $h\in L^{2}\left(  \Theta \times \mathbb{R},\overline{\pi}\right)  $.
\end{lemma}


In practice, we are only interested in functions $h\in L^{2}\left(  \Theta
,\pi\right)  $. To simplify notation, we write $\textsc{\protect\small IF}(h,Q)$ in this case,
instead of introducing the function $\widetilde{h}\in L^{2}\left(
\Theta\times\mathbb{R},\overline{\pi}\right)  $ satisfying $\widetilde{h}%
\left(  \theta,z\right)  =h\left(  \theta\right)  $ for all $z\in\mathbb{R}$
and writing $\textsc{\protect\small IF}(\widetilde{h},Q)$. Proposition
\ref{prop:IACTequality} shows that it is possible to express $\textsc{\protect\small IF}%
(h,Q^{\ast})$ as a function of the inefficiency of its jump kernel
$\widetilde{Q}^{\ast}$, which is particularly useful as $\widetilde{Q}^{\ast}$
admits a simple structure.

\begin{lemma}
\label{Lemma:jumpchainQ*}Assume $Q^{\ast}$ is $\overline{\pi}$-irreducible.
The jump kernel $\widetilde{Q}^{\ast}$ associated with $Q^{\ast}$ is%
\begin{equation}
\widetilde{Q}^{\ast}\left \{  \left(  \theta,z\right)  ,\left(  \mathrm{d}%
\vartheta,\mathrm{d}w\right)  \right \}  =\widetilde{Q}_{\textsc{ex}}\left(
\theta,\mathrm{d}\vartheta \right)  \widetilde{Q}_{\textsc{z}}\left(
z,\mathrm{d}w\right), \label{eq:transitionkerneljumpchainQ*1}%
\end{equation}
where
\begin{equation}
\widetilde{Q}_{\textsc{ex}}\left(  \theta,\mathrm{d}\vartheta \right)
=\frac{q\left(  \theta,\vartheta \right)  \alpha_{\textsc{ex}%
}(\theta,\vartheta)\mathrm{d}\vartheta}{\varrho_{\textsc{ex}}\left(  \theta \right)  },\quad \widetilde{Q}_{\textsc{z}}\left(  z,\mathrm{d}w\right)  =\frac{g\left(
w\right)  \alpha_{\textsc{z}}\left(  z,w\right)  \mathrm{d}w}%
{\varrho_{\textsc{z}}\left(  z\right)  }.
\label{eq:transitionkernelsofjumpchainQ*}%
\end{equation}
The kernel $\widetilde{Q}_{\textsc{ex}}\left(  \theta,\mathrm{d}\vartheta \right)  $
is reversible with respect to $\widetilde{\pi}\left(  \mathrm{d}\theta \right)
$ and the kernel $\widetilde{Q}_{\textsc{z}}\left(  z,\mathrm{d}w\right)$ is positive and reversible with respect to
$\widetilde{\pi}_{\textsc{z}}\left(  \mathrm{d}z\right)$, where
\[
\widetilde{\pi}\left(  \mathrm{d}\theta \right)  =\frac{\pi \left(
\mathrm{d}\theta \right)  \varrho_{\textsc{ex}}\left(  \theta \right)}{\pi(\varrho_{\textsc{ex}})},\quad \widetilde{\pi}_{\textsc{z}}\left(\mathrm{d}z\right)  =\frac{\pi_{\textsc{z}}\left(\mathrm{d} z\right)  \varrho_{\textsc{z}%
}\left(  z\right)}{\pi_{\textsc{z}}\left(  \varrho_{\textsc{z}}\right)  }.
\]
If $Q^{\ast}$ is ergodic, $h\in L_{0}^{2}\left(  \Theta,\pi \right)  $,
$\textsc{\protect\small IF}(h,Q^{\ast})<\infty$ and $\widetilde{Q}^{\ast}$ is ergodic, then
$h/\varrho_{\textsc{ex}}\in L_{0}^{2}\left(  \Theta,\widetilde{\pi}\right)$,
$\pi_{\textsc{z}}\left(  1/\varrho_{\textsc{z}}\right)  <\infty$, $\textsc{\protect\small IF}\{  h/\left(  \varrho_{\textsc{ex}}%
\varrho_{\textsc{z}}\right)  ,\widetilde{Q}^{\ast}\}  <\infty$ and%
\begin{equation}
\pi \left(  h^{2}\right)  \left \{  1+\textsc{\protect\small IF}({h}, {Q^{\ast}})\right \}  =\pi \left(
\varrho_{\textsc{ex}}\right)  \pi_{\textsc{z}}\left(  1/\varrho_\textsc{z}
\right)  \widetilde{\pi}\left(  h^{2}/\varrho_{\textsc{ex}}%
^{2}\right)  \left[  1+\textsc{\protect\small IF}\left\{  h/\left(  \varrho_{\textsc{ex}}%
\varrho_{\textsc{z}}\right)  ,\widetilde{Q}^{\ast}\right\} \right].
\label{eq:equalityIACTQ*}%
\end{equation}
Additionally, $\pi_{\textsc{z}}\left(  1/\varrho_{\textsc{z}}\right)  <\infty$ ensures that
$\widetilde{Q}_{\textsc{z}}$ is geometrically ergodic and $\textsc{\protect\small IF}(1/\varrho_{\textsc{z}},\widetilde{Q}_{\textsc{z}}) < \infty$.
\end{lemma}

The following theorem provides an expression for $\textsc{\protect\small IF}(h,Q^{\ast})$
which decouples the contributions of the parameter and the noise components. The proof exploits the relationships between $\textsc{\protect\small IF}%
(h,Q_{\textsc{ex}})$ and $\textsc{\protect\small IF}(h/\varrho_{\textsc{ex}%
},\widetilde{Q}_{\textsc{ex}})$, $\textsc{\protect\small IF}(h,Q^{\ast})$ and
$\textsc{\protect\small IF}\{h/\left(  \varrho_{\textsc{ex}}\varrho_{\textsc{z}%
}\right)  ,\widetilde{Q}^{\ast}\}$ and the spectral representation (\ref{eq:spectral}) of $\textsc{\protect\small IF}\{h/\left(  \varrho_{\textsc{ex}}
\varrho_{\textsc{z}}\right)  ,\widetilde{Q}^{\ast}\}$. This spectral
representation admits a simple structure due to the product form
(\ref{eq:transitionkerneljumpchainQ*1}) of $\widetilde{Q}^{\ast}$.

\begin{theorem}
\label{Th:ineff boundedness theorem} Let $h\in L^{2}\left( \Theta ,\pi
\right) $.\ Assume that $Q_{\textsc{ex}}$, $Q^{\ast },\widetilde{Q}_{\textsc{ex}}$, $\widetilde{Q}^{\ast }$ are
ergodic with $\textsc{\protect\small IF}(h,Q^{\ast })<\infty $.
Then,
$\textsc{\protect\small IF}({h},Q)\leq \textsc{\protect\small IF}(h,Q^{\ast })$ and%
\begin{multline}
\textsc{\protect\small IF}\left( h,Q^{\ast }\right)
=\frac{1+\textsc{\protect\small IF}(h,Q_{\textsc{ex}})}{\pi_{\textsc{z}}\left( \varrho _{%
\textsc{z}}\right) }-1\\
+\frac{2\left\{ 1+\textsc{\protect\small IF}(h,Q_{\textsc{ex}})\right\} }{1+\textsc{\protect\small IF}(h/\varrho _{\textsc{ex}},%
\widetilde{Q}_{\textsc{ex}})}\left\{\pi_{\textsc{z}}\left( 1/\varrho _{%
\textsc{z}}\right) -\frac{1}{\pi_{\textsc{z}}\left( \varrho _{\textsc{z}%
}\right) }\right\}
\sum_{n=0}^{\infty }\phi _{n}(h/\varrho _{\textsc{%
{ex}}},\widetilde{Q}_{\textsc{ex}}) \phi _{n}(
1/\varrho _{\textsc{z}},\widetilde{Q}_{\textsc{z}}).\label{eq:mainequality}
\end{multline}
\end{theorem}

\begin{remark}
\label{remark:exactproposal}
If $q\left( \theta ,\vartheta \right) =\pi \left( \vartheta \right)$, then
$\textsc{\protect\small IF}(h,Q_{\textsc{ex}})=\textsc{\protect\small IF}(h/\varrho _{\textsc{{ex%
}}},\widetilde{Q}_{\textsc{ex}})=1$ and $\phi _{n}( h/\varrho _{%
\textsc{ex}},\widetilde{Q}_{\textsc{ex}}) =0$ for $n\geq 1$. It follows from Theorem \ref{Th:ineff boundedness theorem} that $\textsc{\protect\small IF}(h,Q^{\ast })=2\pi_{\textsc{z}}\left( 1/\varrho _{\textsc{z}}\right) -1$. This result was established in Lemma 4 of \cite{PittSilvaGiordaniKohn(12)}.
\end{remark}

Theorem \ref{Th:ineff boundedness theorem} requires $Q_{\textsc{ex}},Q^{\ast},\widetilde{Q}_{\textsc{ex}}$ and $\widetilde{Q}^{\ast}$ to
be ergodic. The following proposition, generalizing Theorem 2.2 of \cite{robertstweedie1996}, provides sufficient conditions ensuring this.

\begin{proposition}
\label{Proposition:Q*andQ*JHarrisergodic}Suppose $\pi \left(  \theta \right)  $
is bounded away from $0$ and $\infty$ on compact sets, and there exist
$\delta>0$ and $\varepsilon>0$ such that, for every $\theta$,
\begin{equation}
\left \vert \theta-\vartheta \right \vert \leq \delta \Rightarrow q\left(
\theta,\vartheta \right)  \geq \varepsilon. \label{eq:assumptionproposaltheta}%
\end{equation}
Then $Q_{\textsc{ex}},Q^{\ast},\widetilde{Q}_{\textsc{ex}}$ and
$\widetilde{Q}^{\ast}$ are ergodic.
\end{proposition}

\subsection{Bounds on the relative inefficiency of the pseudo-marginal
chain\label{sec:pseudomarginalbounds}}

For any kernel $\Pi$, we define the relative inefficiency $\mathrm{\textsc{\protect\small RIF}}%
(h,\Pi)=\textsc{\protect\small IF}(h,\Pi)/\textsc{\protect\small IF}(h,Q_{\textsc{ex}})$, which measures the
inefficiency of $\Pi$ compared to that of $Q_{\textsc{ex}}$.  This section provides tractable upper bounds for $\mathrm{\textsc{\protect\small RIF}}(h,Q)$.
 From Lemma \ref{Lemma:Peskun}, $\mathrm{\textsc{\protect\small RIF}}(h,Q)\leq \mathrm{\textsc{\protect\small RIF}}%
(h,Q^{\ast})$, but the expression of $\mathrm{\textsc{\protect\small RIF}}(h,Q^{\ast})$ that follows
from Theorem \ref{Th:ineff boundedness theorem} is intricate and depends on the autocorrelation sequence $\{  \phi_{n}(
h/\varrho_{\textsc{ex}},\widetilde{Q}_{\textsc{ex}})
\}  _{n\geq1}$, as well as other terms.
The next corollary provides upper bounds on $\mathrm{\textsc{\protect\small RIF}}(h,Q)$ that depend
only on $\textsc{\protect\small IF}({h},Q_{\textsc{ex}})$. To simplify the notation, we write $\phi_{\textsc{z}}=\phi_{1}(1/\varrho
_{\textsc{z}},\widetilde{Q}_{\textsc{z}})$.

\begin{corollary}
\label{corollary:boundsQex}Under the assumptions of Theorem
\ref{Th:ineff boundedness theorem},

\begin{enumerate}
\item $\mathrm{\textsc{\protect\small RIF}}(h,Q)\leq{\textsc{\protect\small uRIF}}_{1}(h)$, where%
\begin{align}
\textsc{\protect\small uRIF}_{1}\left(  h\right)   &  =\{1+1/\textsc{\protect\small IF}({h},Q_{\textsc{ex}}%
)\}[\pi_{\textsc{z}}\left(  1/\varrho_{\textsc{z}}\right)  +(1-\phi
_{\textsc{z}})\{ \pi_{\textsc{z}}\left(  1/\varrho_{\textsc{z}}\right)
-1/\pi_{\textsc{z}}\left(  \varrho_{\textsc{z}}\right)
\}]\label{eq:maininequalityloser}\\
&  \quad-1/\textsc{\protect\small IF}({h},Q_{\textsc{ex}});\nonumber
\end{align}

\item if, in addition, $\textsc{\protect\small IF}(h/\varrho_{\textsc{ex}},\widetilde
{Q}_{\textsc{ex}})\geq1$, then $\textsc{\protect\small RIF}(h,Q)\leq{\mathrm{\textsc{\protect\small uRIF}}}%
_{2}(h)\leq{\textsc{\protect\small uRIF}}_{1}(h)$, where
\begin{equation}
\textsc{\protect\small uRIF}_{2}\left(  h\right)  =\left \{  1+1/\textsc{\protect\small IF} \left(  {h}%
,Q_{\textsc{ex}}\right)  \right \}  \pi_{\textsc{z}}\left(  1/\varrho
_{\textsc{z}}\right)  -1/\textsc{\protect\small IF}({h},Q_{\textsc{ex}}).
\label{eq:pos_jumpchain}%
\end{equation}

\end{enumerate}
\end{corollary}

Proposition \ref{prop:positivityMHjumpMH} gives sufficient conditions for the
condition $\textsc{\protect\small IF}(h/\varrho_{\textsc{ex}},\widetilde{Q}_{\textsc{ex}})\geq1$ of Part 2 of Corollary \ref{corollary:boundsQex} to hold.

\begin{remark}
The bounds above are tight in two cases. First, if $\pi_{\textsc{z}}\left(  1/\varrho_{\textsc{z}}\right)  \to 1$, then
$\textsc{\protect\small RIF}({h},Q)$, ${\textsc{\protect\small uRIF}}_{1}(h)$, ${\textsc{\protect\small uRIF}}_{2}(h)\rightarrow1$. Second, if $q\left(  \theta,\vartheta \right)  =\pi \left(  \vartheta
\right)  $, then $\textsc{\protect\small RIF}({h},Q)=\textsc{\protect\small uRIF}_{2}(h)$.
\label{remark:IFZto1}
\end{remark}

We now provide upper bounds on $\mathrm{\textsc{\protect\small RIF}}(h,Q)$ and
lower bounds on $\mathrm{\textsc{\protect\small RIF}}(h,Q^{\ast})$  in terms of
$\textsc{\protect\small IF}(h/\varrho_{\textsc{ex}},\widetilde{Q}_{\textsc{ex}})$.

\begin{corollary}
\label{corollary:boundsQexjump}Under the assumptions of Theorem
\ref{Th:ineff boundedness theorem},
\begin{enumerate}
\item $\textsc{\protect\small RIF}(h,Q)\leq{\mathrm{\textsc{\protect\small uRIF}}}_{3}(h)$, where
\begin{align}
\textsc{\protect\small uRIF}_{3}\left(  h\right)
&  =\left \{  1+ \frac{1}{\textsc{\protect\small IF}(h/\varrho
_{\textsc{ex}},\widetilde{Q}_{\textsc{ex}})}\right \}
\left[
\frac{1}{\pi_{\textsc{z}}(  \varrho_{\textsc{z}})}  +\phi_{\textsc{z}}\left \{  \pi_{\textsc{z}}(  1/\varrho_{\textsc{z}})  -\frac{1}{\pi_{\textsc{z}}(  \varrho_{\textsc{z}})}  \right \}  \right]
\label{eq:RIFh3}\\
&  \quad+2\left \{  \pi_{\textsc{z}}\left(  1/\varrho_{\textsc{z}}\right)
-1/\pi_{\textsc{z}}\left(  \varrho_{\textsc{z}}\right)  \right \}  (1-\phi
_{\textsc{z}})/\textsc{\protect\small IF}(h/\varrho_{\textsc{{ex}}},\widetilde
{Q}_{\textsc{{ex}}})-1/\textsc{\protect\small IF}(h/\varrho_{\textsc{{ex}}},\widetilde{Q}_{\textsc{ex}});\nonumber
\end{align}
\item $\textsc{\protect\small RIF}(h,Q)\leq{\mathrm{\textsc{\protect\small uRIF}}}_{4}(h)$, where%
\begin{align}
\textsc{\protect\small uRIF}_{4}\left(  h\right)   &  =\frac{\left \{  1+1/\textsc{\protect\small IF}(h/\varrho_{\textsc{ex}}%
,\widetilde{Q}_{\textsc{ex}})\right \}  }{  1+\textsc{\protect\small IF}(h/\varrho
_{\textsc{ex}},\widetilde{Q}_{\textsc{ex}})  }\left \{
\pi_{\textsc{z}}\left(  1/\varrho_{\textsc{z}}\right)  -1/\pi_{\textsc{z}}\left(
\varrho_{\textsc{z}}\right)  \right \}  \{1+\textsc{\protect\small IF}(1/\varrho
_{\textsc{z}},\widetilde{Q}_{\textsc{z}})\} \label{eq:RIFh4}\\
& + 1/\pi_{\textsc{z}}\left(  \varrho_{\textsc{z}%
}\right)  +\frac{1}{\textsc{\protect\small IF}(h/\varrho_{\textsc{{ex}}},\widetilde{Q}%
_{\textsc{ex}})}\left \{  \frac{1}{\pi_{\textsc{z}}\left(  \varrho
_{\textsc{z}}\right)  }-1\right \};\quad \nonumber
\end{align}
\item if $\widetilde{Q}_{\textsc{ex}}$ is positive, then $\mathrm{\textsc{\protect\small RIF}}(h,Q^{\ast})\geq{\textsc{\protect\small lRIF}_{1}(h)}$, where
\begin{equation}
\textsc{\protect\small lRIF}_{1}(h)
=\frac{1}{\pi_{\textsc{z}}\left(  \varrho_{\textsc{z}}\right)  }
+\frac{2}{1+\textsc{{\small IF}}
(h/\varrho_{\textsc{ex}},\widetilde{Q}_{\textsc{ex}})}
\big\{  \pi_{\textsc{z}}\left(1/  \varrho_{\textsc{z}}\right)
-1/\pi_{\textsc{z}}\left(  \varrho_{\textsc{z}}\right) \big\};
\label{eq:newlowerboundonRIFQ*}%
\end{equation}
\item$\mathrm{\textsc{\protect\small RIF}}(h,Q^{\ast})\geq{\textsc{\protect\small lRIF}_{2}}$, where
\begin{equation}
\textsc{\protect\small lRIF}_{2}=1/\pi_{\textsc{z}}\left(  \varrho_{\textsc{z}}\right)
,\label{eq:lowerboundonRIFQ*}%
\end{equation}
and $\mathrm{\textsc{\protect\small RIF}}(h,Q^{\ast}), {\textsc{\protect\small uRIF}}_4(h)\rightarrow{\textsc{\protect\small lRIF}_{2}}$ as $\textsc{\protect\small IF}(h/\varrho
_{\textsc{ex}},\widetilde{Q}_{\textsc{ex}})\rightarrow \infty$.
\end{enumerate}
\end{corollary}

Proposition \ref{prop:positivityMHjumpMH} gives sufficient conditions for $\widetilde{Q}_{\textsc{ex}}$ to be positive.
Section \ref{sec:optim} discusses these bounds in more detail.

\bigskip

\subsection{Optimizing the computing time under a Gaussian
assumption\label{sec:optim}}
This section provides quantitative guidelines on how to select the standard deviation $\sigma$ of the noise density, under the following assumption.

\begin{assumption}
\label{assumption:Gaussiannoise}The noise density is $g^{\sigma
}\left(  z\right)  ={\varphi}\left(  z;-\sigma^{2}/2,\sigma^{2}\right)  $, where
${\varphi}(z;a,b^{2})$ is a univariate normal density with mean $a$ and
variance $b^{2}$.
\end{assumption}

Assumption \ref{assumption:Gaussiannoise} ensures that $\int\exp\left(
z\right)  g^{\sigma}\left(  z\right)  dz=1$ as required by the unbiasedness of
the likelihood estimator. Consider a time series $y_{1:T}=\left(  y_{1}%
,\ldots,y_{T}\right)  $, where the likelihood estimator $\widehat{p}%
(y_{1:T}\mid\theta)$ of $p(y_{1:T}\mid\theta)$ is computed through a particle
filter with $N$ particles. Theorem 1 of an unpublished technical report
(arXiv:1307.0181) by B\'{e}rard et al. shows that, under regularity
assumptions, the log-likelihood error is distributed according to a normal
density with mean $-\delta\gamma^{2}/2$ and variance $\delta\gamma^{2}$ as
$T\rightarrow\infty$, for $N=\delta^{-1}T$. Hence, in this important scenario,
the noise distribution satisfies approximately the form specified in
Assumption \ref{assumption:Gaussiannoise} for large $T$ and the variance is
asymptotically inversely proportional to the number of samples. This
assumption is also made in \cite{PittSilvaGiordaniKohn(12)}, where it is
justified experimentally. Section \ref{sect:applications} below provides
additional experimental results.

The next result is Lemma 4 in \cite{PittSilvaGiordaniKohn(12)} and follows
from Assumption \ref{assumption:Gaussiannoise}, equation (\ref{eq:targetinz})
and Remark \ref{remark:exactproposal}. We now make the dependence on $\sigma$
explicit in our notation.

\begin{corollary}
\label{corr:IF_Z_gauss}Under Assumption \ref{assumption:Gaussiannoise}, $\pi_{\textsc{z}}^{\sigma }(z)={\varphi }\left( z;\sigma ^{2}/2,\sigma ^{2}\right) $,
\begin{equation*}
\varrho _\textsc{z}^{\sigma }\left( z\right) =1-\Phi (z/\sigma
+\sigma /2)+\exp (-z)\Phi (z/\sigma -\sigma /2),\text{ \ }\pi_{\textsc{z}}^{\sigma
}\left( 1/\varrho_\textsc{z}^{\sigma }\right) =\int \frac{\varphi ( w;0,1) }{1-\overline{\varrho }_\textsc{z}^{\sigma }\left( w\right) }%
\mathrm{d}w,
\end{equation*}%
where $\overline{\varrho }_\textsc{z}^{\sigma }\left( w\right) =\Phi (w+\sigma )-\exp
(-w\sigma -\sigma ^{2}/2)\Phi (w)$ and $\Phi (\cdot )$ is the standard
Gaussian cumulative distribution function. Additionally, $\pi_{\textsc{z}}^{\sigma
}\left( \varrho _\textsc{{z}}^{\sigma }\right) =2\Phi (-\sigma
/\surd 2)$.
\end{corollary}

The terms $\pi_{\textsc{z}}^{\sigma}\left(1/ \varrho _\textsc{{z}}^{\sigma }\right) $, $\phi_{\textsc{z}}^{\sigma}$ and $\textsc{\protect\small IF}(1/\varrho _\textsc{{z}}^{\sigma },{\widetilde{Q}_{\textsc{{z}}}})$, appearing in the bounds of Corollaries \ref{corollary:boundsQex} and
\ref{corollary:boundsQexjump}, do not admit analytic expressions, but can be computed numerically.
We note that $\pi_{\textsc{z}}^{\sigma}\left(1/ \varrho _\textsc{{z}}^{\sigma }\right)$ is finite, and thus by Lemma~\ref{Lemma:jumpchainQ*} $\textsc{\protect\small IF}(1/\varrho _\textsc{{z}}^{\sigma },{\widetilde{Q}_{\textsc{{z}}}})$ is also finite.
Consequently, for specific
values of $\sigma,$ $\textsc{\protect\small IF}(h,Q_{\textsc{ex}})$ and $\textsc{\protect\small IF}(h/\varrho_{\textsc{ex}},\widetilde{Q}_{\textsc{ex}})$, these bounds can be calculated.


We now use these bounds to guide the choice of $\sigma$.
The quantity we aim to minimize is the relative
computing time for $Q$ defined as $\textsc{\protect\small RCT}({h},Q;\sigma)=\textsc{\protect\small RIF}%
\left(  h,Q;\sigma\right)  /\sigma^{2}$ because $1/\sigma^{2}$ is usually
approximately proportional to the number of samples $N$ used to estimate the
likelihood and the computational cost at each iteration is typically
proportional to $N$, at least in the particle filter scenario described previously. We define $\textsc{\protect\small RCT}({h},Q^\ast;\sigma)$ similarly. As $\textsc{\protect\small RIF}\left(  h,Q;\sigma\right)$ is intractable, we instead minimize the upper bounds ${\textsc{\protect\small uRCT}}_{i}(h;\sigma)=\textsc{\protect\small uRIF}%
_{i}\left(  h;\sigma\right)  /\sigma^{2}$, for $i=1,\dots,4$. We similarly define the quantities ${\textsc{\protect\small lRCT}_{1}}(h; \sigma)=\textsc{\protect\small lRIF}_{1}(h;\sigma)/\sigma^{2}$ and ${\textsc{\protect\small lRCT}_{2}}(\sigma)=\textsc{\protect\small lRIF}_{2}(\sigma)/\sigma^{2}$, which bound $\textsc{\protect\small RCT}({h},Q^{\ast
};\sigma)$ from below. Figure \ref{fig:theory} plots these bounds against $\sigma$ for
different values of $\textsc{\protect\small IF}(h,Q_{\textsc{ex}})$ and $\textsc{\protect\small IF}%
(h/\varrho_{\textsc{ex}},\widetilde{Q}_{\textsc{ex}})$.

Prior to discussing how these results guide the selection of $\sigma$, we outline some properties of the bounds. First, as the corresponding inefficiency increases, the upper bounds ${\textsc{\protect\small uRCT}}_{i}(h;\sigma)$  displayed in Fig.~\ref{fig:theory} become flatter as functions of $\sigma$, and the corresponding minimizing argument $\sigma_\text{opt}$ increases. This flattening effect suggests less sensitivity to the choice of $\sigma$ for the pseudo-marginal algorithm.
Second, for given $\sigma$, all the upper bounds are decreasing functions of the corresponding inefficiency, which suggests that the penalty from using the pseudo-marginal algorithm drops as the exact algorithm becomes more inefficient.
Third, in the case discussed in Remark 2, where $q(\theta, \vartheta)= \pi( \vartheta)$, so that $\textsc{\protect\small IF}(h,Q_{\textsc{ex}})=\textsc{\protect\small IF}
(h/\varrho_{\textsc{ex}},\widetilde{Q}_{\textsc{ex}})=1$, we obtain ${\textsc{\protect\small uRCT}}_2(h;\sigma)={\textsc{\protect\small uRCT}}_3(h;\sigma)=\textsc{\protect\small RCT}({h},Q^\ast;\sigma)=\textsc{\protect\small RCT}({h},Q;\sigma)$.
Fourth, ${\textsc{\protect\small uRCT}}_4(h;\sigma)$ agrees with the lower bound ${\textsc{\protect\small lRCT}_2}(\sigma)$ as $\textsc{\protect\small IF}%
(h/\varrho_{\textsc{ex}},\widetilde{Q}_{\textsc{ex}})\to \infty$ as indicated by Part 2 of Corollary~\ref{corollary:boundsQexjump}. In this case, these two bounds, as well as $\textsc{\protect \small uRCT}_1(h;\sigma)$, are sharp for $\textsc{\protect\small RCT}({h},Q^\ast;\sigma)$.
Fifth, ${\textsc{\protect\small uRCT}}_2(h;\sigma)$ is sharper than ${\textsc{\protect\small uRCT}}_1(h;\sigma)$ for $\textsc{\protect\small RCT}({h},Q^\ast;\sigma)$, but requires a mild additional assumption.

\begin{figure}[!h]
\captionsetup{font=small}
\begin{center}
{\small Relative computing time against $\sigma$ for different inefficiencies of the exact chain.}
\begin{subfigure}{.4\textwidth}
\caption*{$\scriptstyle\textsc{uRCT}_1$}
\vskip-10pt
\includegraphics[width=1\linewidth]{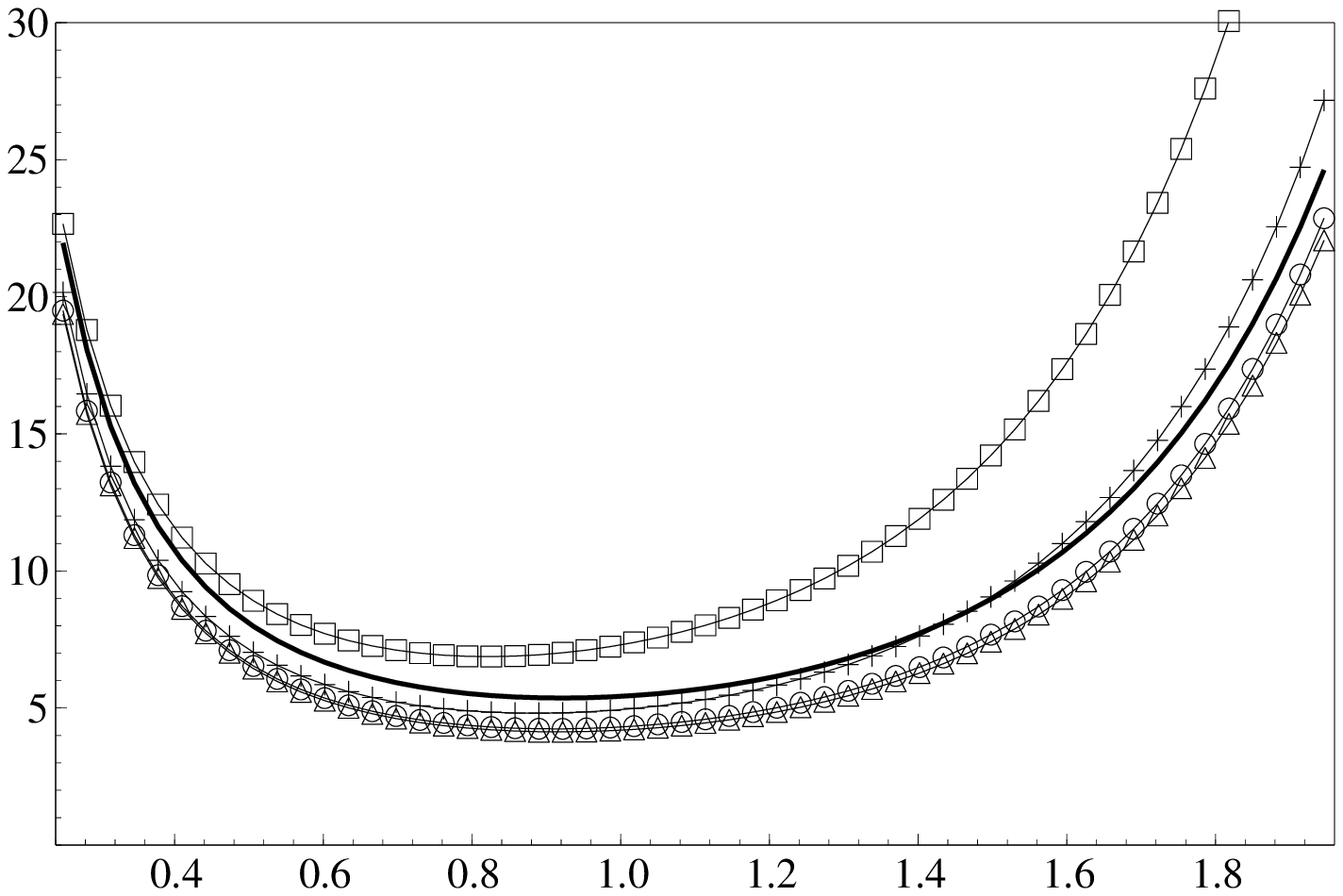}
\end{subfigure}
\begin{subfigure}{.4\textwidth}
\caption*{$\scriptstyle\textsc{uRCT}_2$}
\vskip-10pt
\includegraphics[width=1\linewidth]{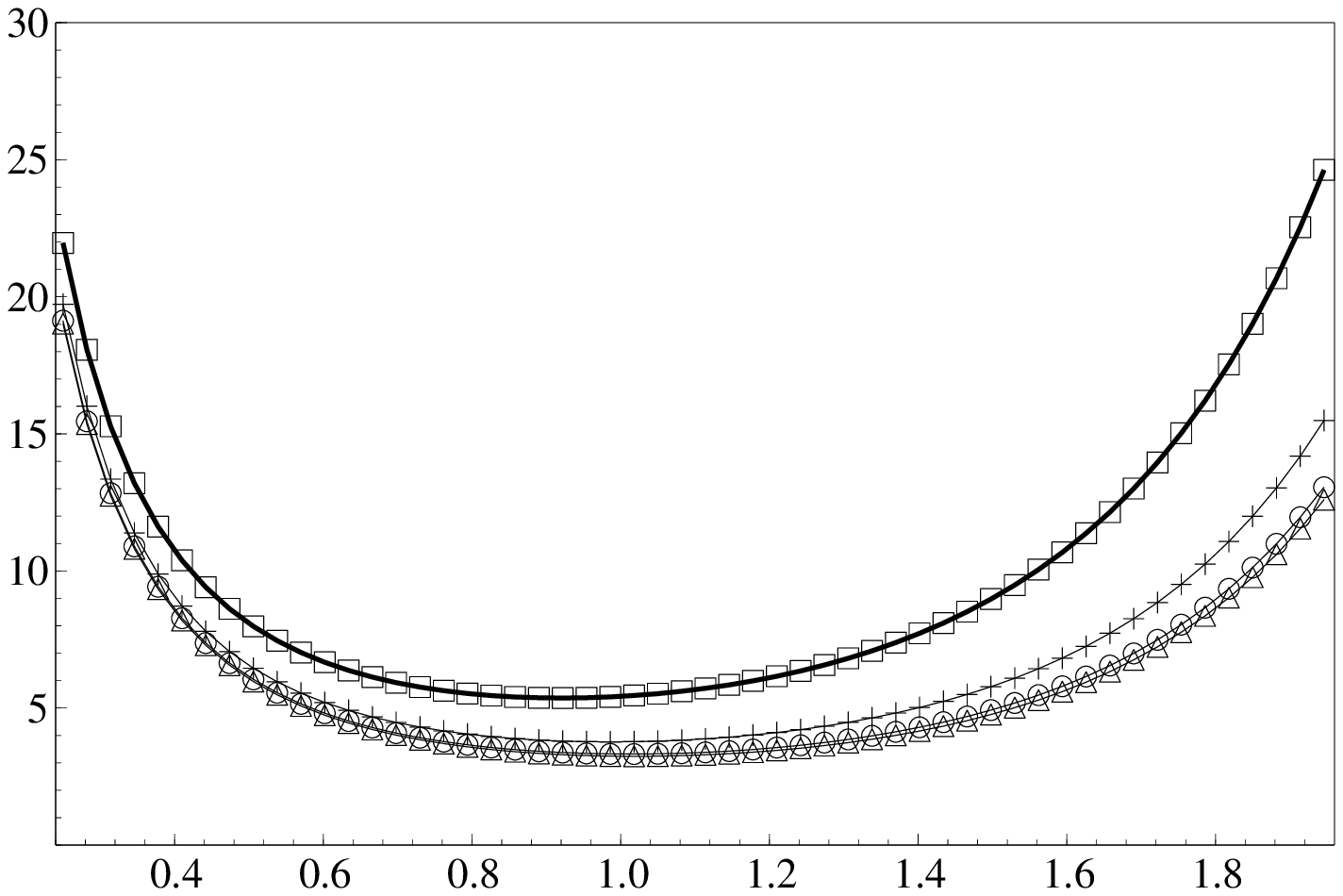}
\end{subfigure}
\end{center}
\vskip 5pt
\begin{center}
{\small Relative computing time against $\sigma$ for different inefficiencies of the exact jump chain.}
\begin{subfigure}{.4\textwidth}
\caption*{$\scriptstyle\textsc{uRCT}_3$}
\vskip-10pt
\includegraphics[width=1\linewidth]{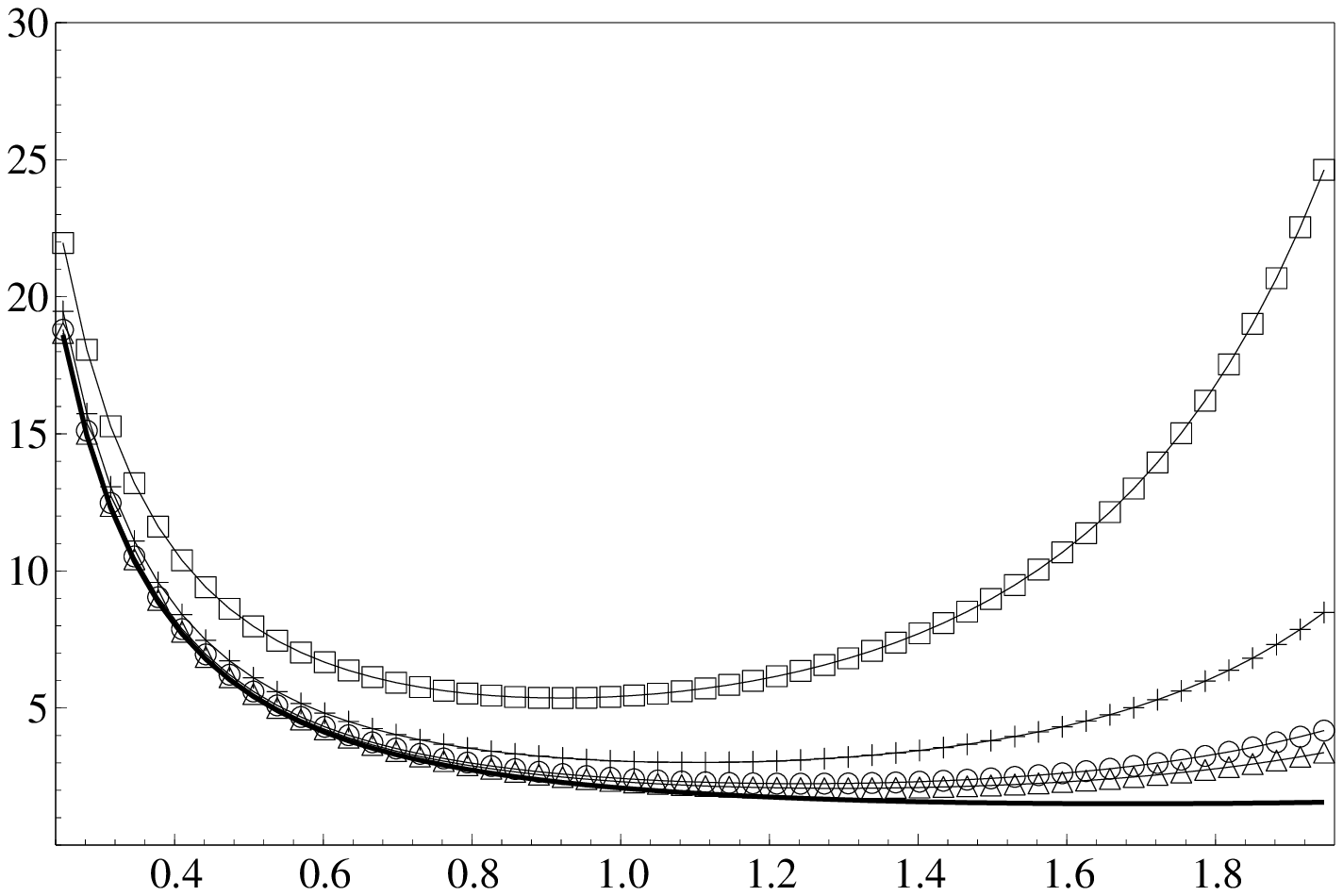}
\end{subfigure}
\begin{subfigure}{.4\textwidth}
\caption*{$\scriptstyle\textsc{uRCT}_4$}
\vskip-10pt
\includegraphics[width=1\linewidth]{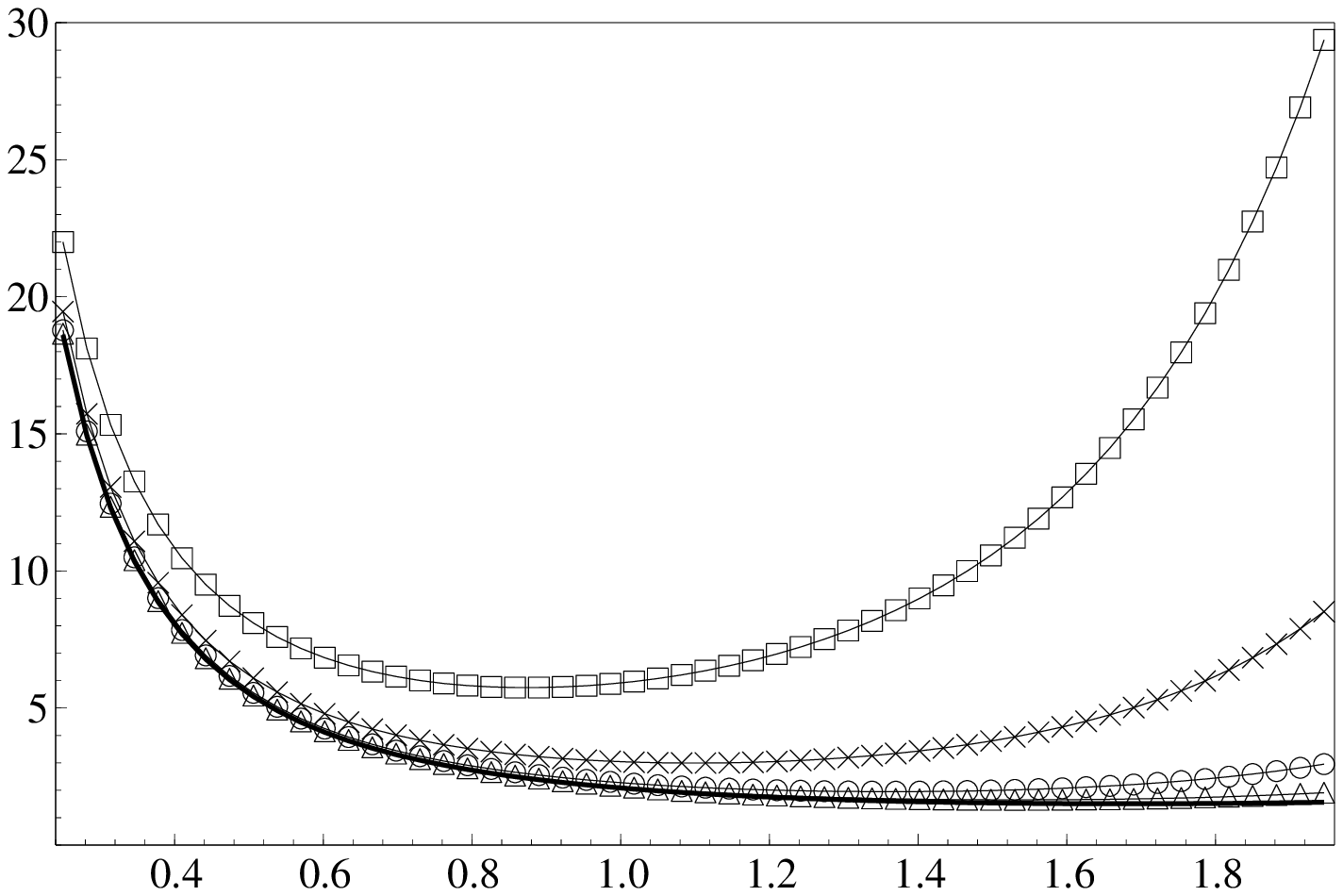}
\end{subfigure}
\caption{{Theoretical results for relative computing time against $\sigma$.
Top: The bounds ${\textsc{{\protect \small uRCT}}}_{1}(h;\sigma)$
(left) and ${\textsc{{\protect \small uRCT}}}_{2}(h;\sigma)$ (right)
are displayed. Different
values of $\textsc{{\protect \small IF}}\left(  h,Q_{\textsc{ex}}\right)  $ are
taken as $1$ (squares), $4$ (crosses), $20$ (circles) and $80$ (triangles).
The solid line corresponds to the perfect proposal, as discussed in Remark 2.
Bottom: The lower bound ${\textsc{{\protect \small lRCT}}}_2(\sigma)$ (solid
line) is shown together with ${\textsc{{\protect \small uRCT}}}_{3}(h;\sigma)$
(left) and ${\textsc{{\protect \small uRCT}}}_{4}(h;\sigma)$ (right). Different
values of $\textsc{{\protect \small IF}}(h/\varrho_{\textsc{ex}},\widetilde
{Q}_{\textsc{ex}})$ are taken as $1$ (squares), $4$ (crosses), $20$ (circles)
and $80$ (triangles). }}%
\label{fig:theory}
\end{center}
\vskip-20pt
\end{figure}
As the likelihood is intractable, it is necessary to make a judgment on how to
choose $\sigma$, because $\textsc{\protect\small IF}\left(
h,Q_{\textsc{ex}}\right)  $ and $\textsc{\protect\small IF}(h/\varrho_{\textsc{{ex}%
}},\widetilde{Q}_{\textsc{ex}})$ are unknown and cannot be easily
estimated. Consider two extreme scenarios. The first is the perfect
proposal $q\left(  \theta,\vartheta\right)  =\pi\left(  \vartheta\right)  $, so
that by Corollary~\ref{corr:IF_Z_gauss} and Remark~\ref{remark:exactproposal},
$\textsc{\protect\small RCT}\left(  h,Q;\sigma\right)  =\{2\pi_\textsc{z}^\sigma(1/\varrho_\textsc{z}^\sigma)-1\}/\sigma^2$,  which we denote by $\textsc{\protect\small RCT}\left(  h,Q_\pi;\sigma\right)$, is minimized at $\sigma_\text{opt}=0.92$. The second scenario
considers a very inefficient proposal corresponding to Part 4 of Corollary
\ref{corollary:boundsQexjump} so that $\textsc{\protect\small RCT}%
\left(  h,Q^{\ast};\sigma\right) ={\textsc{\protect\small lRCT}}_2(\sigma)$, which is minimized at $\sigma_\text{opt}=1.68$.
If we choose $\sigma_\text{opt}=1.68$ over $\sigma_\text{opt}=0.92$ in scenario $1$, then
$\textsc{\protect\small RCT}\left(  h,Q_\pi;\sigma\right)$ rises from $5.36$ to $12.73$. Conversely, if we choose
$\sigma_\text{opt}=$ $0.92$ over $\sigma_\text{opt}=1.68$ in scenario $2$, the relative
computing time $\textsc{\protect\small RCT}\left(  h,Q^{\ast};\sigma\right)  $ rises from
$1.51$ to $2.29$. This suggests that the penalty in choosing the wrong value is much more
severe if we incorrectly assume we are in scenario $2$ than if we incorrectly
assume we are in scenario $1$. This is because as $\textsc{\protect\small IF}(  h/\varrho
_{\textsc{ex}},\widetilde{Q}_{\textsc{ex}})$ increases,  ${\textsc{\protect\small lRCT}}_2(\sigma)$ is very flat relative to $\textsc{\protect\small RCT}\left(  h,Q_\pi;\sigma\right)$, as a function of $\sigma$.
In practice, choosing $\sigma_\text{opt}$ slightly greater than $1.0$ appears
sensible. For example, a value of $\sigma=1.2$ leads to an increase in
${\textsc{\protect\small RCT}}({h},Q_{\pi};\sigma)$ from the minimum value of $5.36$ to
$6.10$ and an increase in ${\textsc{\protect\small lRCT}}_2(\sigma)$ from the minimum value of
$1.51$ to $1.75$. In Appendix~2, we compute lower and upper bounds for the minimizing argument of ${\textsc{\protect\small RCT}}(h, Q^\ast; \sigma)$
for various values of $\textsc{\protect\small IF}(h/\varrho_{\textsc{{ex}%
}},\widetilde{Q}_{\textsc{ex}})$.

Some caution should be exercised in interpreting these results as the lower bounds apply to ${\textsc{\protect\small RCT}}(h, Q^{\ast};\sigma)$, but not in general to ${\textsc{\protect\small RCT}}(h, Q;\sigma)$. Similarly, whilst $\textsc{{\protect \small uRCT}}_4(h;\sigma)$ and the lower bounds become exact for ${\textsc{\protect\small RCT}}(h, Q^{\ast};\sigma)$ as $\textsc{\protect\small IF}(h/\varrho_{\textsc{{ex}%
}},\widetilde{Q}_{\textsc{ex}})\to\infty$, they only provide upper bounds for ${\textsc{\protect\small RCT}}(h, Q;\sigma)$.

However, in an important class of problems  $\textsc{\protect\small IF}(h/\varrho_{\textsc{{ex}%
}},\widetilde{Q}_{\textsc{ex}})$ is large, for instance when $q(\theta, \vartheta)$ is a random walk proposal with small step size. In this case, we expect that as the step size gets smaller the acceptance probability $\alpha_{\textsc{ex}}$ of $Q_{\textsc{ex}}$ will tend towards unity and hence asymptotically $\alpha_{Q^\ast}= \alpha_Q$. This suggests that, for small enough step size, ${\textsc{\protect\small RCT}}(h, Q^{\ast};\sigma)\approx{\textsc{\protect\small RCT}}(h, Q;\sigma)$.
The numerical results in this section are based on Assumption 2. However, the bounds on the relative inefficiences of $Q$  and $Q^{\ast}$ presented in Corollaries \ref{corollary:boundsQex} and \ref{corollary:boundsQexjump} can be calculated for any
other noise distribution $g\left(  z\right)  $, subject to $\int\exp\left(z\right)  g\left(  z\right)  dz=1$. These bounds can in turn be used
to construct corresponding bounds on the relative computing times of $Q$  and $Q^{\ast}$, provided that an appropriate penalization term is employed to account
for the computational effort of obtaining the likelihood estimator.

\subsection{Discussion}

We now compare informally the bound $\textsc{\protect\small lRIF}_{2}(\sigma)=1/\{2\Phi
(-\sigma/\surd2)\}$ of Part 4 of Corollary~\ref{corollary:boundsQexjump} to the
results in \citet{Sherlock2013efficiency}. These authors make Assumption \ref{assumption:noiseorthogonal}, assume
that the target factorises into $d$ independent and identically distributed
components and that the proposal is an isotropic Gaussian random walk of jump
size $d^{-1/2}l$. In the Gaussian noise case, for $h\left(  \theta\right)
=\theta_{1}$ where $\theta=\left(  \theta_{1},...,\theta_{d}\right)  $, their
results and a standard calculation with their diffusion limit, suggest that as $d\rightarrow\infty$ the
relative inefficiency satisfies
\begin{equation}
\frac{\textsc{\protect\small IF}\left(  h,Q;\sigma,l\right)  }{\textsc{\protect\small IF}\left(
h,Q_{\textsc{ex}};l\right)  }=\textsc{\protect\small RIF}(h,Q;\sigma,l)\rightarrow
\textsc{\protect\small aRIF}(\sigma,l)=\frac{J_{\sigma^{2}=0}(l)}{J_{\sigma^{2}}(l)}%
=\frac{\Phi(-l/2)}{\Phi\left\{  -\left(  2\sigma^{2}+l^{2}\right)
^{1/2}/2\right\}  }, \label{eq:RIF_gr}%
\end{equation}
where the expression for $J_{\sigma^{2}}(l)$ is given by equations (3.3) and
(3.4) of \citet{Sherlock2013efficiency}. We observe that $\textsc{\protect\small aRIF}(\sigma,l)$ converges to $\textsc{\protect\small lRIF}_{2}(\sigma)$ as
$l\rightarrow0$. This is unsurprising. As $d\rightarrow\infty$, we conjecture
that in this scenario the conditions of
Part 4 of Corollary~\ref{corollary:boundsQexjump} apply, in particular that $\textsc{\protect\small IF}(
h/\varrho_{\textsc{ex}},\widetilde{Q}_{\textsc{ex}})
\rightarrow\infty$ for any $l>0$. Therefore, in this case, $\textsc{\protect\small RIF}%
\left(  h,Q^{\ast};\sigma,l\right)  \rightarrow\textsc{\protect\small lRIF}_{2}(\sigma)$. As
$l\rightarrow0$, we have informally that $\varrho_{\textsc{ex}}\left(
\theta\right)  \rightarrow1$, so that it is reasonable to conjecture that
$\textsc{\protect\small RIF}\left(  h,Q;\sigma,l\right)  /\textsc{\protect\small RIF}\left(  h,Q^{\ast
};\sigma,l\right)  \rightarrow1$. If one of these limits holds uniformly, then
$\textsc{\protect\small aRIF}(\sigma,l)\rightarrow\textsc{\protect\small lRIF}_{2}(\sigma)$.


\section{Application\label{sect:applications}}
\subsection{Stochastic volatility model and pseudo-marginal algorithm}

This section examines a multivariate partially observed diffusion model, which
was introduced by \cite{ChernovGallantGhyselsTauchen(03)}, and discussed
in \cite{HuangTauchen(05)}. The regularly observed log price $P(t)$ evolves according to,%
\begin{align*}
&\mathrm{d\log}P(t)    =\mu_{y}\mathrm{d}t+\text{s-}\exp\left[  \left\{
v_{1}(t)+\beta_{2}v_{2}(t)\right\}  /2\right]  \mathrm{d}B(t),\\
&\mathrm{d}v_{1}(t)    =-k_{1}\left\{  v_{1}(t)-\mu_{1}\right\}
\mathrm{d}t+\sigma_{1}\mathrm{d}W_{1}(t), \text{  }\mathrm{d}v_{2}(t)=-k_{2}%
v_{2}(t)\mathrm{d}t+\left\{  1+\beta_{12}v_{2}(t)\right\}  \mathrm{d}W_{2}(t),
\end{align*}
%
and the leverage parameters corresponding to the correlations between the
driving Brownian motions are $\phi_{1}=$corr$\left\{  B(t),W_{1}%
(t)\right\}  $ and $\phi_{2}=$corr$\left\{  B(t),W_{2}(t)\right\}  $. The
function s-$\exp\left(  \cdot\right)  $ is a spliced exponential function to
ensure non-explosive growth, see \cite{HuangTauchen(05)}. The two components
for volatility allow for quite sudden changes in log price whilst retaining
long memory in volatility. We note that the Brownian motion of the price
process may be expressed as $\mathrm{d}B(t)=a_{1}\mathrm{d}W_{1}(t)%
+a_{2}\mathrm{d}W_{2}(t)+\surd b\mathrm{d}\overline{B}(t)$, where $a_{1}=\phi
_{1}(1-\phi_{2}^{2})/(1-\phi_{1}^{2}\phi_{2}^{2})$, $a_{2}=\phi_{2}(1-\phi
_{1}^{2})/(1-\phi_{1}^{2}\phi_{2}^{2})$ and $b=(1-\phi_{1}^{2})(1-\phi_{2}%
^{2})/(1-\phi_{1}^{2}\phi_{2}^{2})$. Here $\overline{B}(t)$ is an independent
Brownian motion. Suppose the log prices are observed at equally spaced times
$\tau_{1}<\tau_{2}<$ $\tau_{2}<\ldots<\tau_{T}<\tau_{T+1}$ and $\Delta
=\tau_{s+1}-\tau_{s}$ for any $s$  which gives returns $Y_{s}=\log P(\tau_{s+1})-\log P(\tau_{s})$, for
$s=1,\ldots,T$. The distribution of these returns conditional upon the
volatility paths and the driving processes $W_{1}(t)$ and $W_{2}(t)$ is available in closed form as
$Y_{s}\sim\mathcal{N}\left(  \mu_{y}\Delta+a_{1}Z_{1,s}+a_{2}Z_{2,s}%
;b\sigma_{s}^{2\ast}\right),$ where%
\begin{equation}
Z_{1,s}=\int_{\tau_{s}}^{\tau_{s+1}}\sigma(u)\mathrm{d}W_{1}(u)\text{,
}Z_{2,s}=\int_{\tau_{s}}^{\tau_{s+1}}\sigma(u)\mathrm{d}W_{2}(u),\text{
}\sigma_{s}^{2\ast}=\int_{\tau_{s}}^{\tau_{s+1}}\sigma^{2}(u)du,
\label{exact_intvol}%
\end{equation}
and $\sigma(t)=$s-$\exp\left[  \left\{  v_{1}(t)+\beta_{2}v_{2}(t)\right\}
/2\right]  $. An Euler scheme is used to approximate the evolution of the
volatilities $v_{1}(t)$ and $v_{2}(t)$ by placing a number, $M-1$, of latent points between $\tau_{s}$ and $\tau_{s+1}$. The volatility components are
denoted by $v_{1,1}^{s},...,v_{1,M-1}^{s}$ and $v_{2,1}^{s},...,v_{2,M-1}^{s}%
$. For notational convenience, the start and end points are set to $v_{1,0}%
^{s}=v_{1}(\tau_{s})$ and $v_{1,M}^{s}=v_{1}(\tau_{s+1}),$ and similarly for
$v_{2}(t)$. These latent points are evenly spaced in time by $\delta=\Delta
/M$. The equation for the Euler evolution, starting at $v_{1,0}^{s}%
=v_{1,M}^{s-1}$ and $v_{2,0}^{s}=v_{2,M}^{s-1}$, is
\begin{align*}
v_{1,m+1}^{s}  &  =v_{1,m}^{s}-k_{1}(v_{1,m}^{s}-\mu_{1})\delta+\sigma
_{1}\surd\delta u_{1,m},\\
v_{2,m+1}^{s}  &  =v_{2,m}^{s}-k_{2}v_{2,m}^{s}\delta+\left(  1+\beta
_{12}v_{2,m}^{s}\right)  \surd{\delta}u_{2,m},\text{ }m=0,\ldots,M-1,
\end{align*}
where $u_{1,m}\sim\mathcal{N}\left(  0,1\right)  $ and $u_{2,m}\sim
\mathcal{N}\left(  0,1\right)  $. Conditional upon these trajectories and the
innovations, the distribution of the returns has a closed form so that
$Y_{s}\sim\mathcal{N}\left(  \mu_{y}\Delta+a_{1}\widehat{Z}_{1,s}%
+a_{2}\widehat{Z}_{2,s};b\widehat{\sigma}_{s}^{2\ast}\right)  ,$ where
$\widehat{Z}_{1,s}$, $\widehat{Z}_{2,s}$ and $\widehat{\sigma}_{s}^{2\ast}$
are the Euler approximations to the corresponding expression in
(\ref{exact_intvol}).

We consider $T$ daily returns, $y=\left(  y_{1},...,y_{T}\right)  $, from the
S\&P 500 index. Bayesian inference is performed on the $9$-dimensional
parameter vector $\theta=(k_{1},\mu_{1},\sigma_{1},k_{2},\beta_{12},\beta
_{2},\mu_{y},\phi_{1},\phi_{2})$ to which we assign a vague prior. We simulate from the posterior density $\pi(\theta)$ using
 the pseudo-marginal algorithm where the likelihood
is estimated using the bootstrap particle filter with $N$ particles. A
multivariate Student-t random walk proposal on the parameter components transformed to
the real line is used.

\subsection{Empirical results for the error of the log-likelihood estimator\label{sec:estimator_perf}}

This section investigates empirically Assumptions 1 and 2 by examining
the behaviour of $Z=\log\widehat{p}_{N}(y\mid\theta)-\log p(y\mid\theta)$ for $T=40$, 300 and 2700. Corresponding
values of $N$ are selected in each case to ensure
that the variance of $Z$ evaluated at the posterior mean $\overline{\theta}$
is approximately unity. We
use $\delta=0.5$ in the Euler scheme.

The three plots on the left of Fig.~\ref{fig:2fact_hist} display the
histograms corresponding to the density of $Z$
for $\theta=\overline{\theta}$ denoted $g_{N}(z\mid\overline{\theta})$, which is obtained by running $S=6000$ particle
filters at this value. As $p(y\mid\overline{\theta})$ is
unknown, it is estimated by averaging these estimates. The Metropolis--Hastings algorithm is then used to
obtain the histograms corresponding to $\pi_{N}(z\mid\overline{\theta}%
)=\exp\left(  z\right)  g_{N}(z\mid\overline{\theta})$. We overlay on each
histogram a kernel density estimate together with the corresponding assumed density, $g_{\textsc{z}}^{\sigma
}\left(  z\right)  $ or $\pi_{\textsc{z}}^{\sigma}\left(  z\right)  $, where $\sigma^{2}$ is the sample variance of $Z$ over the $S$ particle
filters. For $T=40$, there is a discrepancy between the assumed Gaussian densities and the true
histograms representing $g_{N}(z\mid\overline{\theta})$ and $\pi_{N}%
(z\mid\overline{\theta})$. In particular, whilst $g_{N}(z\mid\overline{\theta
})$ is well approximated over most of its support, it is slightly lighter
tailed than the assumed Gaussian in the right tail and much heavier tailed in
the left tail. This translates into a smaller discrepancy between $g_{N}%
(z\mid\theta)$ and $\pi_{N}(z\mid\theta)$ and a higher acceptance rate for the
pseudo-marginal algorithm than the Gaussian assumption suggests. For $T=300$ and $T=2700$, the assumed Gaussian densities are
very accurate.

We also examine $Z$ when $\theta$ is distributed according to $\pi(\theta)$. We record $200$ samples from $\pi
(\theta)$, for $T=40$, 300 and 2700. For each of these samples, we run the particle filter $300$ times in order to estimate the
true likelihood at these values. The resulting histograms, corresponding to the densities ${\textstyle\int}\pi\left(  \mathrm{d}\theta\right) g_{N}(z\mid\theta)$ and ${\textstyle\int}\pi\left(  \mathrm{d}\theta\right) \pi_{N}(z\mid\theta)$, are displayed in the middle column of
Fig.~\ref{fig:2fact_hist}. We similarly examine the density of $Z$
when $\theta$ is distributed according to the marginal proposal density in the
stationary regime ${\textstyle\int}\pi\left(  \mathrm{d}\vartheta\right)  q\left(  \vartheta,\theta\right)  $.
Here $q\left(  \vartheta,\theta\right)  $ is a multivariate
Student-t random walk proposal, with step size proportional to $T^{-1/2}$. The right hand column of Fig.~\ref{fig:2fact_hist} shows the
resulting histograms. In both scenarios, Assumptions 1 and 2 are problematic
for $T=40$ as $g_{N}(z\mid\overline{\theta})$ is
not close to being Gaussian as $T\ $is too small for the central limit theorem
to provide a good approximation. Moreover, since $T$ is small, $\pi(\theta)$ and ${\textstyle\int}\pi\left(  \mathrm{d}\vartheta\right)  q\left(  \vartheta,\theta\right)  $ are
relatively diffuse. Consequently, $g_{N}(z\mid\overline{\theta})$ is not close
to $g_{N}(z\mid\theta)$ marginalized over $\pi(\theta)$ or ${\textstyle\int}\pi\left(  \mathrm{d}\vartheta\right)  q\left(  \vartheta,\theta\right)  $.
For $T=300$ and $T=2700$, the assumed densities
$g_{\textsc{z}}^{\sigma}\left(  z\right)  $ and $\pi_{\textsc{z}}^{\sigma
}\left(  z\right)  $ are close to the corresponding histograms and
Assumptions 1 and 2 appear to capture reasonably well the salient features of the densities associated with $Z$. In particular,
the approximation suggested by the central limit theorem becomes very good.
Additionally, $\pi(\theta)$ and ${\textstyle\int}\pi\left(  \mathrm{d}\vartheta\right)  q\left(  \vartheta,\theta\right)  $ are
sufficiently concentrated to ensure that the variance of $Z$ as a function of $\theta$ exhibits little variability.

\begin{figure}[!h]
\captionsetup{font=small}
\includegraphics[height=2.4in, width=\textwidth]{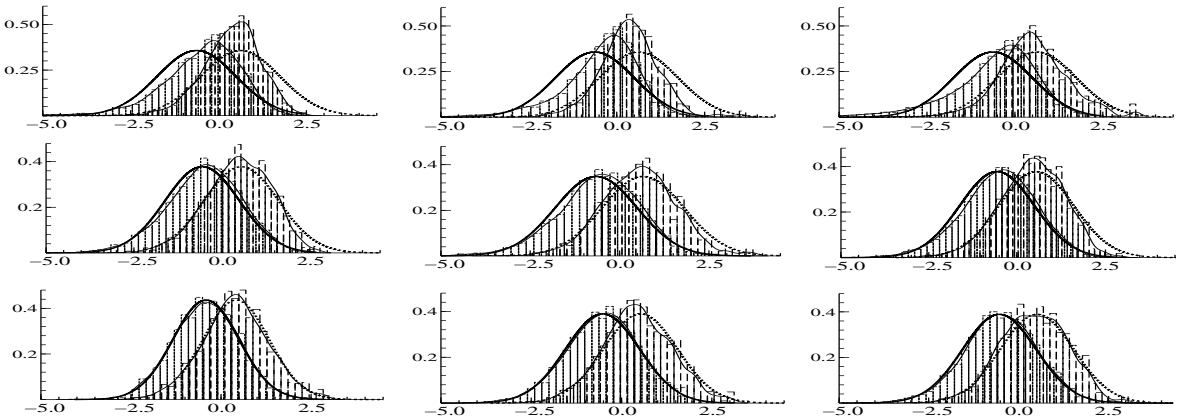}\caption{{Huang and
Tauchen two factor model for S\&P 500 data. Top to bottom: }${T=40,}$ $N=4$
(top), $T=300,$ $N=80$ (middle), $T=2700,$ $N=700$ (bottom). Left to right:
histograms and theoretical densities associated with $g_{N}(z\mid \theta)$ and
$\pi_{N}(z\mid \theta)$ evaluated at the posterior mean $\overline{\theta}$
(left), over values from the posterior $\pi(\theta)$ (middle) and over values
from $\int\pi(\mathrm{d}\vartheta) q(\vartheta,\theta)$ (right). The densities $g_\textsc{z}^\sigma(z)$ and $\pi_\textsc{z}^\sigma(z)$ are overlaid (solid lines).}
\label{fig:2fact_hist}%
\vskip-10pt
\end{figure}

\subsection{Empirical results for the pseudo-marginal algorithm}

We apply the pseudo-marginal algorithm with $\delta=0.05$, $T=300$ and various values of $N$.  The standard deviation $\sigma\left(  \overline{\theta};N\right)$ of $\log\widehat{p}_{N}(y\mid\overline{\theta})$ is evaluated by Monte Carlo simulations, where $\overline{\theta}$ is the posterior mean. For each value of $N$, we compute the inefficiencies, denoted by $\mathrm{\textsc{\protect\small IF}}$, and the corresponding approximate relative computing times, denoted by  $\mathrm{\textsc{\protect\small RCT}}$, of all parameter components. The quantity $\mathrm{\textsc{\protect\small RCT}}$ is computed as $\mathrm{\textsc{\protect\small IF}}/{\sigma^{2}\left(  \overline{\theta};N\right)}$ divided by the inefficiency of $Q$ when $N=2000$, the latter being an approximation of the inefficiency of $Q_{\textsc{ex}}$. The results are very similar for all parameter components and so, for ease of presentation, Fig.~\ref{fig:theory2} shows the average quantities over the $9$ components. For most parameters, the optimal value for $\sigma\left(  \overline{\theta};N\right)$ is between $1.2$ and $1.5$, corresponding to $N=40$ and $60$. The results agree with the bound  ${\textsc{\protect\small uRCT}}_4(h;\sigma)$ in Section 3.5. This can be partly explained because the inefficiencies associated with $\widetilde{Q}$ for $N=2000$ are large, suggesting that the inefficiencies associated with $\widetilde{Q}_{\textsc{ex}}$ are large.
%

As all the bounds in the paper are based on $Q^{\ast}$, it is useful to assess the discrepancy between $Q$ and $Q^{\ast}$.
One approach to explore this discrepancy is to examine the
marginal acceptance probability $\overline{\pi}(\varrho_{\text{{\tiny Q}}})$
under $Q$ against $\sigma=\sigma(\overline{\theta},N)$ as $N$ varies. Using
the acceptance criterion (\ref{eq:acceptanceprobabilityQ*}) of $Q^{\ast}$, we
obtain under Assumptions \ref{assumption:noiseorthogonal} and
\ref{assumption:Gaussiannoise} that $\overline{\pi}(\varrho_{\textsc{{q}}})\geq2\Phi(-\sigma/\surd2)\pi(\varrho_{\textsc{ex}})$.
If $Q$ and $Q^{\ast}$ are close in the
sense of having similar marginal acceptance probabilities, then we expect
$\overline{\pi}(\varrho_{\text{{\tiny Q}}})$ to have a similar shape as its
lower bound where $\pi(\varrho_{\textsc{ex}})$ is approximated using $\overline{\pi}(\varrho_{\text{{\tiny Q}}})$ with $N=2000$. For this model, the two functions on either side of the inequality, displayed in Fig.~\ref{fig:theory2}, are similar.

\begin{figure}[ptb]
\captionsetup{font=small}
\begin{center}
\begin{subfigure}{.32\textwidth}
\caption*{Log of average $\textsc{if}$ against $\sigma$}
\vskip-10pt
\includegraphics[height=3cm,width=1\linewidth]{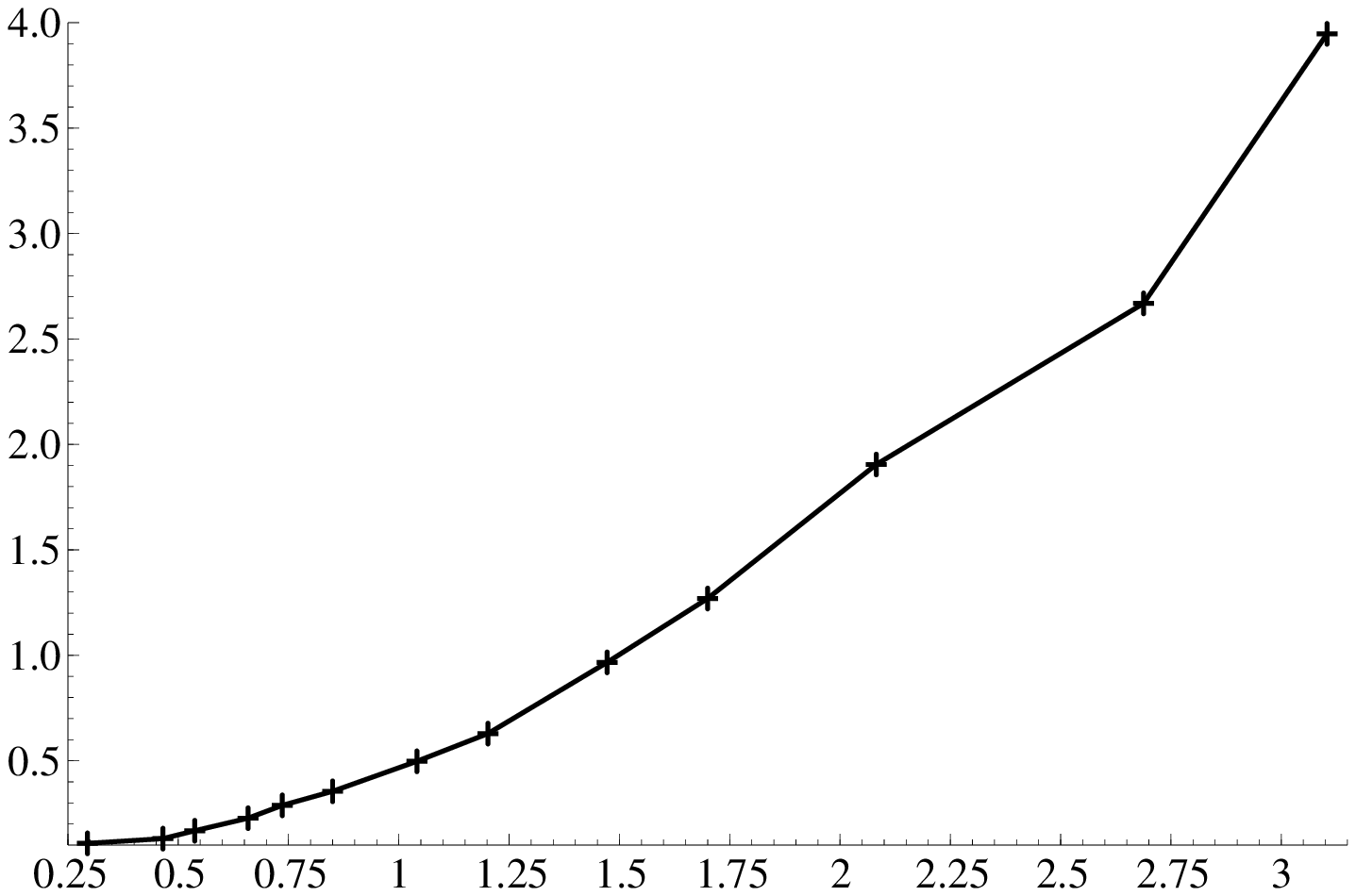}
\end{subfigure}
\begin{subfigure}{.32\textwidth}
\caption*{Average $\textsc{rct}$ against $\sigma$}
\vskip-10pt
\includegraphics[height=3cm,width=1\linewidth]{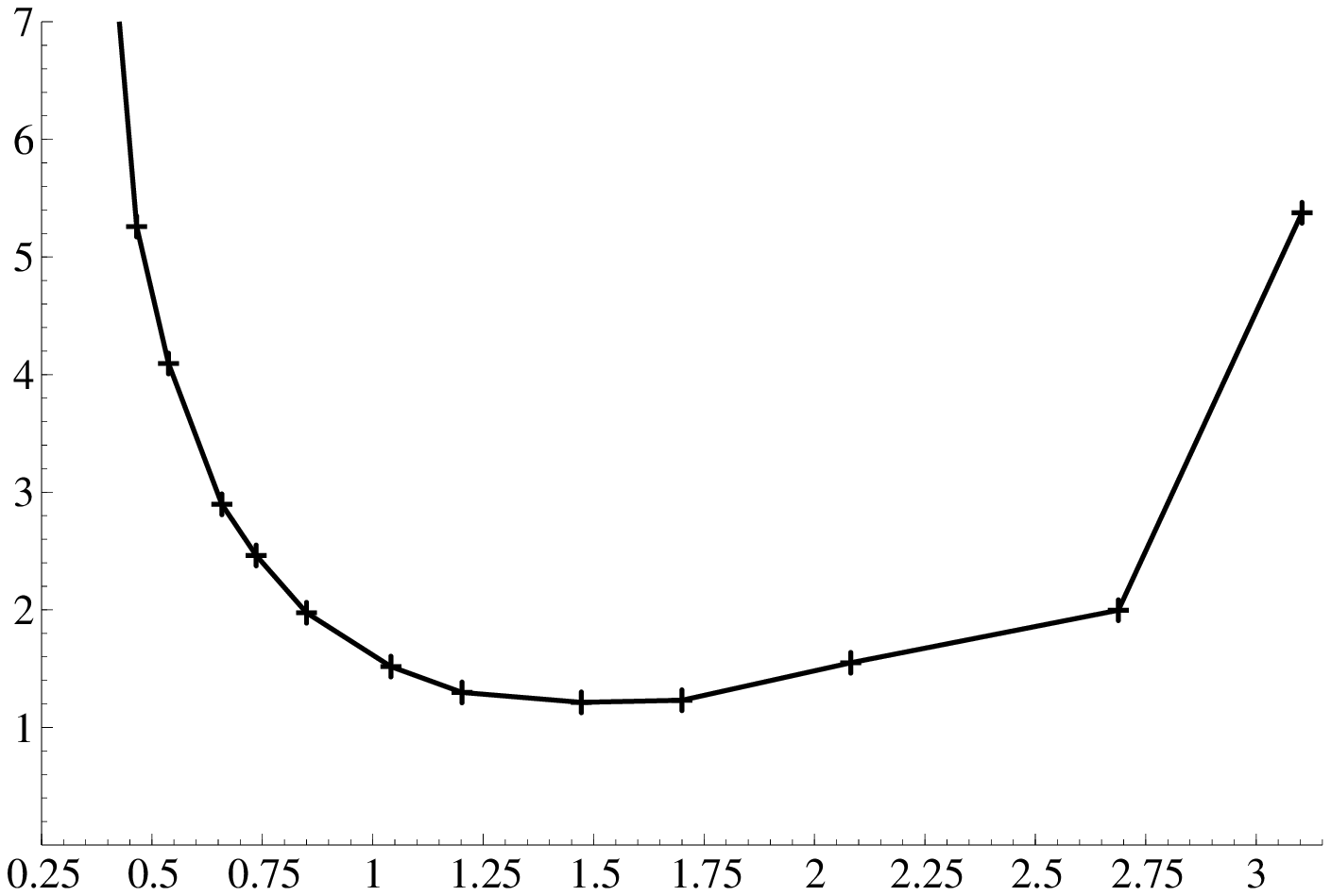}
\end{subfigure}
\begin{subfigure}{.32\textwidth} 
\caption*{Acceptance Probability against $\sigma$}
\vskip-10pt
\includegraphics[height=3cm,width=1\linewidth]{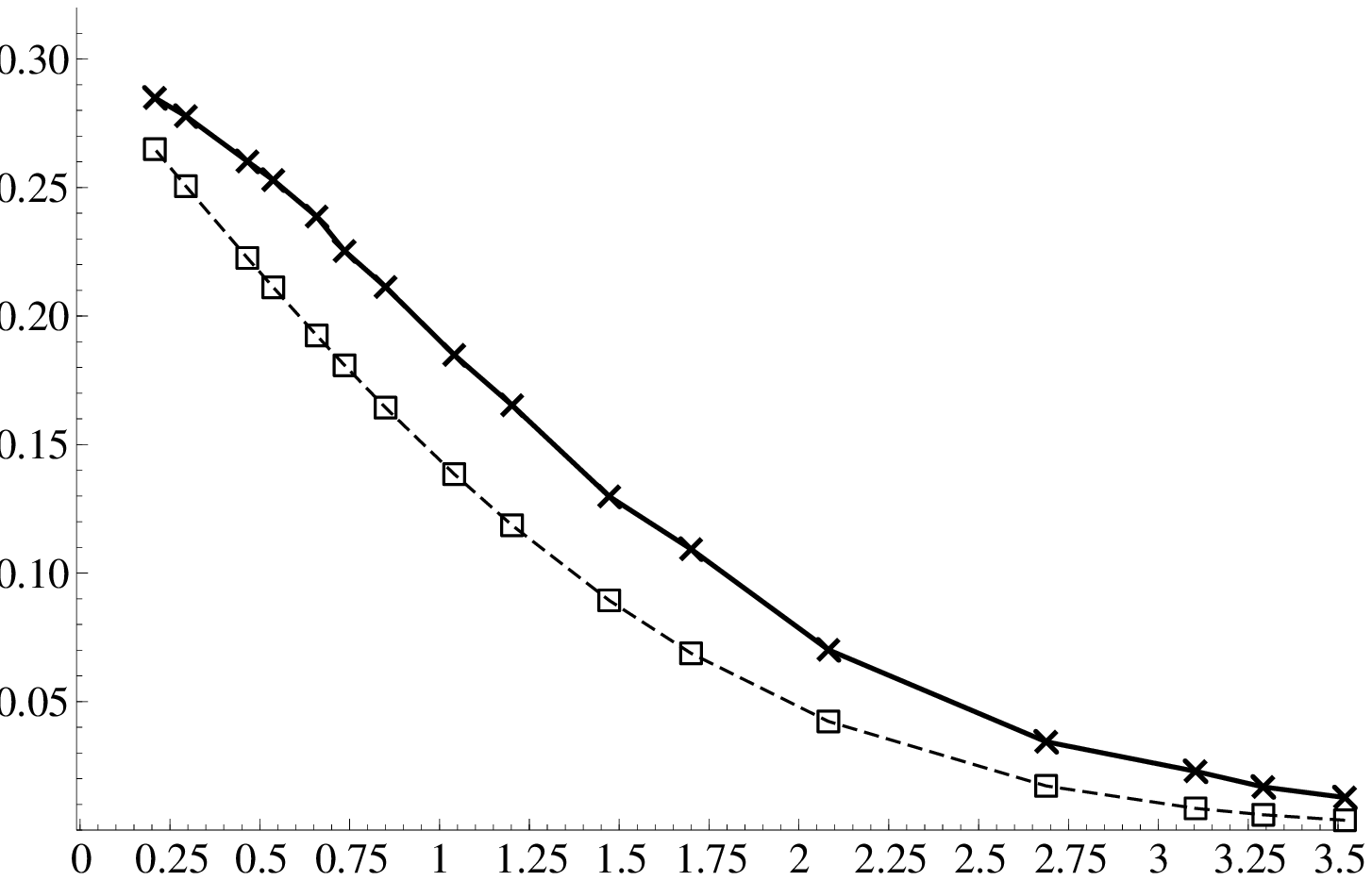}
\end{subfigure}
\end{center}
\caption{{ Huang and Tauchen two factor model for S\&P 500 data, $T=300$.
Inefficiencies (\textsc{if}) and Relative Computing Times (\textsc{rct}) against $\sigma$, where $\textsc{if}$ is computed by averaging over the 9 parameter components. Right panel: The marginal acceptance probability $\overline{\pi
}(\varrho_{\text{{\protect \tiny Q}}})$ (crosses) against $\sigma$ together with
the lower bound (squares) $2\Phi(-\sigma/\surd2)\pi(\varrho_{\textsc{{ex}}})$.}}%
\label{fig:theory2}%
\end{figure}%
\section*{Acknowledgements}

The authors would like to thank the editor, the associate editor and the reviewers for
their comments which helped to improve the paper significantly. Arnaud Doucet
was partially supported by EPSRC\ and Robert Kohn was partially supported by an ARC Discovery grant.
%
%
\appendix
\section*{Appendix 1\label{A: Proofs}}

\begin{proof}
[Proof of Lemma~\ref{Lemma:Peskun}]It is straightforward to establish that $Q^{\ast
}$ is $\overline{\pi}$-reversible. Moreover, for any $a,b\geq0$,
$\min(1,a)\min(1,b)\leq\min(1,ab)$ so $\alpha_{Q^{\ast}}\left\{  \left(
\theta,z\right)  ,\left(  \vartheta,w\right)  \right\}  \leq\alpha_{Q}\left\{
\left(  \theta,z\right)  ,\left(  \vartheta,w\right)  \right\}  $ for any
$\theta,z,\vartheta,w$. Hence, Theorem 4 in \cite{tierney1998note}, which is a
general state-space extension of \cite{peskun1973optimum}, applies and yields
the result.
\end{proof}

\begin{proof}
[Proof of Theorem \ref{Th:ineff boundedness theorem}]Without loss of generality, let
$h\in L_{0}^{2}\left(  \Theta,\pi\right)  $. By Theorem 6 of \citet{AV12}, $\textsc{\protect\small IF}\left(  h,Q_{\textsc{ex}}\right)  \leq\textsc{\protect\small IF}\left(
h,Q\right)  $ and, by Lemma \ref{Lemma:Peskun}, $\textsc{\protect\small IF}\left(
h,Q\right)  \leq\textsc{\protect\small IF}\left(  h,Q^{\ast}\right)  $, where $\textsc{\protect\small IF}%
\left(  h,Q^{\ast}\right)  <\infty$ by assumption. Hence, $\textsc{\protect\small IF}\left(  h,Q_{\textsc{ex}}\right)  <\infty$ and Proposition
\ref{prop:IACTequality} applied to $Q_{\textsc{ex}}$ yields that
$\textsc{\protect\small IF}(  h/\varrho_{\textsc{ex}},\widetilde{Q}%
_{\textsc{ex}}) ,$ $\widetilde{\pi}\left(  h^{2}/\varrho
_{\textsc{ex}}^{2}\right)  <\infty$ and
\begin{equation}
\pi\left(  h^{2}\right)  \left\{  1+\textsc{\protect\small IF}\left(  h,Q_{\textsc{ex}%
}\right)  \right\}  =\pi\left(  \varrho_{\textsc{ex}}\right)
\widetilde{\pi}\left(  h^{2}/\varrho_{\textsc{ex}}^{2}\right)  \left\{
1+\textsc{\protect\small IF}(  h/\varrho_{\textsc{ex}},\widetilde{Q}%
_{\textsc{ex}})  \right\}  . \label{eq:relationshipIACTexact}%
\end{equation}
Since the assumptions of Lemma \ref{Lemma:jumpchainQ*} are satisfied, we can
substitute (\ref{eq:relationshipIACTexact}) into (\ref{eq:equalityIACTQ*}) to
obtain
\begin{equation}
1+\textsc{\protect\small IF}\left(  h,Q^{\ast}\right)  =\pi_{\textsc{z}}\left(  1/\varrho
_{\textsc{z}}\right)  \frac{\left\{  1+\textsc{\protect\small IF}\left(
h,Q_{\textsc{ex}}\right)  \right\}  }{1+\textsc{\protect\small IF}(
h/\varrho_{\textsc{ex}},\widetilde{Q}_{\textsc{ex}})
}\left[  1+\textsc{\protect\small IF}\left\{  h/\left(  \varrho_{\textsc{ex}}%
\varrho_{\textsc{z}}\right)  ,\widetilde{Q}^{\ast}\right\}  \right]  .
\label{eq:identityoninefficiencyQ*}%
\end{equation}
We now provide a spectral representation for $\textsc{\protect\small IF}\{h/\left(
\varrho_{\textsc{ex}}\varrho_{\textsc{z}}\right)  ,\widetilde{Q}%
^{\ast}\}$. With $\widetilde{\pi}\otimes\widetilde{\pi}_{\textsc{z}}\left(
\mathrm{d}\theta,\mathrm{d}z\right)  =\widetilde{\pi}\left(  \mathrm{d}%
\theta\right)  \widetilde{\pi}_{\textsc{z}}\left(  \mathrm{d}z\right)  $,
\begin{align}
\textsc{\protect\small IF}\left\{  h/\left(  \varrho_{\textsc{ex}}\varrho
_{\textsc{z}}\right)  ,\widetilde{Q}^{\ast}\right\}   &  =1+2\sum
_{n=1}^{\infty}\frac{\left\langle \varrho_{\textsc{ex}}^{-1}%
\varrho_{\textsc{z}}^{-1}h,\left(  \widetilde{Q}^{\ast}\right)
^{n}\varrho_{\textsc{ex}}^{-1}\varrho_{\textsc{z}}^{-1}%
h\right\rangle _{\widetilde{\pi}\otimes\widetilde{\pi}_{\textsc{z}}}}{\widetilde{\pi
}\otimes\widetilde{\pi}_{\textsc{z}}\left(  \varrho_{\textsc{z}}^{-2}%
\varrho_{\textsc{ex}}^{-2}h^{2}\right)  } \label{eq:identityIACTmfinite}%
\\
&  =1+2\sum_{n=1}^{\infty}\frac{\left\langle \varrho_{\textsc{z}}%
^{-1},\left(  \widetilde{Q}_{\textsc{z}}\right)  ^{n}\varrho
_{\textsc{z}}^{-1}\right\rangle _{\widetilde{\pi}_{\textsc{z}}}\left\langle
\varrho_{\textsc{ex}}^{-1}h,\left(  \widetilde{Q}_{\textsc{ex}%
}\right)  ^{n}\varrho_{\textsc{ex}}^{-1}h\right\rangle _{\widetilde{\pi
}}}{\widetilde{\pi}_{\textsc{z}}\left(  \varrho_{\textsc{z}}^{-2}\right)
\widetilde{\pi}\left(  \varrho_{\textsc{ex}}^{-2}h^{2}\right)
}\nonumber
\end{align}
and, as $\widetilde{Q}_{\textsc{z}}$ and $\widetilde{Q}%
_{\textsc{ex}}$ are reversible, the following spectral representations,
as in (\ref{eq:spectral}), hold

%

\begin{equation}
\begin{split}
\phi_{n}(  1/\varrho
_{\textsc{z}},\widetilde{Q}_{\textsc{z}})
&=\frac{\left\langle \varrho_{\textsc{z}}^{-1},\left(  \widetilde{Q}%
_{\textsc{z}}\right)  ^{n}\varrho_{\textsc{z}}^{-1}\right\rangle
_{\widetilde{\pi}_{\textsc{z}}}-\left\{  \widetilde{\pi}_{\textsc{z}}\left(  \varrho
_{\textsc{z}}^{-1}\right)  \right\}  ^{2}}{\mathbb{V}_{\widetilde{\pi
}_{\textsc{z}}}\left(  \varrho_{\textsc{z}}^{-1}\right)  }
=
\int_{-1}^1
\lambda^{n}\widetilde{e}_{\textsc{z}}(\mathrm{d}\lambda),\\
\phi_{n}(  h/\varrho_{\textsc{ex}%
},\widetilde{Q}_{\textsc{ex}})&=
\frac{\left\langle \varrho_{\textsc{ex}}^{-1}h,\left(
\widetilde{Q}_{\textsc{ex}}\right)  ^{n}\varrho_{\textsc{ex}}%
^{-1}h\right\rangle _{\widetilde{\pi}}}{\widetilde{\pi}\left(  \varrho
_{\textsc{ex}}^{-2}h^{2}\right)  }
=
\int_{-1}^1
\omega^{n}\widetilde{e}_{\textsc{ex}}(\mathrm{d}\omega),
\end{split}
\label{eq:spectralmeasures}%
\end{equation}
where we define $\mathbb{V}_{\widetilde{\pi}_{\textsc{z}}}\left(  \varrho
_{\textsc{z}}^{-1}\right)  =\widetilde{\pi}_{\textsc{z}}\left[  \left\{
\varrho_{\textsc{z}}^{-1}-\widetilde{\pi}_{\textsc{z}}\left(  \varrho
_{\textsc{z}}^{-1}\right)  \right\}  ^{2}\right]  $, $\widetilde{e}%
_{\textsc{z}}(\mathrm{d}\lambda)=e(  \varrho_{\textsc{z}%
}^{-1},\widetilde{Q}_{\textsc{z}})  (\mathrm{d}\lambda)$ and
$\widetilde{e}_{\textsc{ex}}(\mathrm{d}\omega)=e(  \varrho
_{\textsc{ex}}^{-1}h,\widetilde{Q}_{\textsc{ex}})
(\mathrm{d}\omega)$ to simplify notation. Using $\widetilde{\pi}_{\textsc{z}}\left(
\varrho_{\textsc{z}}^{-1}\right)  =1/\pi_{\textsc{z}}\left(  \varrho
_{\textsc{z}}\right)  $, we can rewrite (\ref{eq:identityIACTmfinite})
as%
\begin{align}
\lefteqn{\textsc{\protect\small IF}\left\{  h/\left(  \varrho_{\textsc{ex}}\varrho
_{\textsc{z}}\right)  ,\widetilde{Q}^{\ast}\right\}    =1+2\sum
_{n=1}^{\infty}\frac{1}{\widetilde{\pi}_{\textsc{z}}\left(  \varrho_{\textsc{z}%
}^{-2}\right)  }\left\{  \mathbb{V}_{\widetilde{\pi}_{\textsc{z}}}\left(
\varrho_{\textsc{z}}^{-1}\right)  \int\lambda^{n}\widetilde{e}%
_{\textsc{z}}(\mathrm{d}\lambda)+\frac{1}{\pi_{\textsc{z}}\left(
\varrho_{\textsc{z}}\right)^{2}  }\right\}  \int\omega^{n}\widetilde{e}%
_{\textsc{ex}}(\mathrm{d}\omega)}\notag\\
& =1+2\left(  1-\gamma\right)  \int\frac{\omega}{1-\omega}\widetilde{e}%
_{\textsc{ex}}(\mathrm{d}\omega)+2\gamma\iint\frac{\lambda\omega
}{1-\lambda\omega}\widetilde{e}_{\textsc{z}}(\mathrm{d}\lambda
)\widetilde{e}_{\textsc{ex}}(\mathrm{d}\omega)\notag\\
 &=-1+2\left(  1-\gamma\right)  \int\left(  1+\frac{\omega}{1-\omega}\right)
\widetilde{e}_{\textsc{ex}}(\mathrm{d}\omega)+2\gamma\iint\left(
1+\frac{\omega\lambda}{1-\lambda\omega}\right)  \widetilde{e}_{\textsc{{z}%
}}(\mathrm{d}\lambda)\widetilde{e}_{\textsc{ex}}(\mathrm{d}\omega),
\label{eq:exactexpressionIFQ*}%
\end{align}
where the second expression is finite since $\int(1+\omega)\left(
1-\omega\right)  ^{-1}\widetilde{e}_{\textsc{ex}}(\mathrm{d}%
\omega)=\textsc{\protect\small IF}(  h/\varrho_{\textsc{ex}%
},\widetilde{Q}_{\textsc{ex}})  <\infty$ and
\begin{equation}
\gamma={\mathbb{V}_{\widetilde{\pi}_{\textsc{z}}}\left(  \varrho_{\textsc{{z}%
}}^{-1}\right)  }/{\widetilde{\pi}_{\textsc{z}}\left(  \varrho_{\textsc{z}}%
^{-2}\right)  }=\left\{
{\pi_{\textsc{z}}\left(  \varrho_{\textsc{z}}^{-1}\right)  -1/\pi_{\textsc{z}}\left(
\varrho_{\textsc{z}}\right)  }\right\}/{\pi_{\textsc{z}}\left(
\varrho_{\textsc{z}}^{-1}\right)  }. \label{eq:identitygamma}
\end{equation}

Rearranging (\ref{eq:exactexpressionIFQ*}), we obtain
\begin{align}
1+\textsc{\protect\small IF}\left\{  h/\left(  \varrho_{\textsc{ex}}\varrho
_{\textsc{z}}\right)  ,\widetilde{Q}^{\ast}\right\}   &  =\left\{
1+\textsc{\protect\small IF}(h/\varrho_{\textsc{ex}},\widetilde{Q}_{\textsc{ex}}) \right\}  \left(  1-\gamma\right)  +\gamma
\beta\nonumber\\
&  =\frac{1+\textsc{\protect\small IF}(h/\varrho_{\textsc{ex}},\widetilde{Q}_{\textsc{ex}})  }{\pi_{\textsc{z}}\left(  \varrho_{\textsc{z}%
}\right)  \pi_{\textsc{z}}\left(  \varrho_{\textsc{z}}^{-1}\right)  }+\left\{
\frac{\pi_{\textsc{z}}\left(  \varrho_{\textsc{z}}^{-1}\right)  -1/\pi_{\textsc{z}}\left(
\varrho_{\textsc{z}}\right)  }{\pi_{\textsc{z}}\left(  \varrho_{\textsc{z}%
}^{-1}\right)  }\right\}  \beta, \label{eq:identityB}%
\end{align}
with
\begin{equation}
\frac{\beta}{2}=\iint\frac{\widetilde{e}_{\textsc{z}}(\mathrm{d}%
\lambda)\widetilde{e}_{\textsc{ex}}(\mathrm{d}\omega)}{1-\omega\lambda
}=\sum_{n=0}^{\infty}\phi_{n}(  h/\varrho_{\textsc{ex}%
},\widetilde{Q}_{\textsc{ex}})  \phi_{n}(  1/\varrho
_{\textsc{z}},\widetilde{Q}_{\textsc{z}}). \label{eq:betadefn}
\end{equation}
By substituting (\ref{eq:identityB}) into (\ref{eq:identityoninefficiencyQ*}),
we obtain the result since
\begin{align}
\textsc{\protect\small IF}\left(  h,Q^{\ast}\right)   &  =\pi_{\textsc{z}}\left(  1/\varrho
_{\textsc{z}}\right)  \frac{\{1+\textsc{\protect\small IF}\left(  h,Q^{\textsc{{ex}%
}}\right)  \}}{1+\textsc{\protect\small IF}(h/\varrho_{\textsc{ex}},\widetilde{Q}%
_{\textsc{ex}})}\left[  1+\textsc{\protect\small IF}\left\{  h/\left(  \varrho
_{\textsc{ex}}\varrho_{\textsc{z}}\right)  ,\widetilde{Q}^{\ast
}\right\}  \right]  -1\nonumber\\
&  =\frac{\left\{  1+\textsc{\protect\small IF}(h,Q_{\textsc{ex}})\right\}
}{1+\textsc{\protect\small IF}(h/\varrho_{\textsc{ex}},\widetilde{Q}%
_{\textsc{ex}})}\left\{  \pi_{\textsc{z}}\left(  \varrho_{\textsc{z}%
}^{-1}\right)  -\frac{1}{\pi_{\textsc{z}}\left(  \varrho_{\textsc{z}}\right)
}\right\}  \beta+\frac{1+\textsc{\protect\small IF}(h,Q_{\textsc{ex}})}{\pi_{\textsc{z}}\left(
\varrho_{\textsc{z}}\right)  }-1. \label{eqn:exact_RIF_Qstar}%
\end{align}

\end{proof}

\begin{proof}
[Proof of Corollary \ref{corollary:boundsQex}]
Dividing  (\ref{eqn:exact_RIF_Qstar}) by $\textsc{\protect\small IF}(h,Q_{\textsc{ex}})$, we obtain
\begin{align}
\textsc{\protect\small RIF}\left(  h,Q^{\ast}\right)   &  =\frac{\pi_{\textsc{z}}\left(
\varrho_{\textsc{z}}^{-1}\right)  \left \{  1+\textsc{\protect\small IF}%
(h,Q_{\textsc{ex}})\right \}  }{\textsc{\protect\small IF}(h,Q_{\textsc{ex}%
})\{1+\textsc{\protect\small IF}(h/\varrho_{\textsc{ex}},\widetilde{Q}%
_{\textsc{ex}})\}}A-\frac{1}{\textsc{\protect\small IF}(h,Q_{\textsc{ex}}%
)},\label{eq:RIF_equal}
\end{align}
where $A$ is the quantity in (\ref{eq:identityB}) and can be expressed in terms of
$\gamma$, defined in (\ref{eq:identitygamma}), as
\begin{align}
A  &  =1+\textsc{\protect\small IF}(h/\varrho_{\textsc{ex}},\widetilde{Q}%
_{\textsc{ex}})-2\gamma \iint \left \{  \frac{1}{(1-\omega)}-\frac
{1}{(1-\lambda \omega)}\right \}  \widetilde{e}_{\textsc{z}}%
(\mathrm{d}\lambda)\widetilde{e}_{\textsc{ex}}(\mathrm{d}\omega
)\nonumber \\
&  =1+\textsc{\protect\small IF}(h/\varrho_{\textsc{ex}},\widetilde{Q}%
_{\textsc{ex}})-2\gamma \iint \frac{\omega \left(  1-\lambda \right)
}{\left(  1-\omega \right)  (1-\omega \lambda)}\widetilde{e}_{\textsc{z}%
}(\mathrm{d}\lambda)\widetilde{e}_{\textsc{ex}}(\mathrm{d}%
\omega).\nonumber
\end{align}
Lemma \ref{Lemma:jumpchainQ*} ensures that the
kernel $\widetilde{Q}_{\textsc{z}}$ is positive, implying that $\widetilde{e}_{\textsc{z}}\left \{  \left[
0,1\right)  \right \}  =1$. Hence,
\[
\iint \left \{  \frac{\omega \left(  1-\lambda \right)  }{\left(  1-\omega \right)
(1-\omega \lambda)}-\frac{\omega \left(  1-\lambda \right)  }{\left(
1-\omega \right)  }\right \}  \widetilde{e}_{\textsc{z}}(\mathrm{d}%
\lambda)\widetilde{e}_{\textsc{ex}}(\mathrm{d}\omega)=\iint \frac
{\omega^{2}(1-\lambda)\lambda}{\left(  1-\omega \right)  (1-\omega \lambda
)}\widetilde{e}_{\textsc{z}}(\mathrm{d}\lambda)\widetilde{e}%
_{\textsc{ex}}(\mathrm{d}\omega)\geq 0.\text{ }%
\]
We can now bound $A$ from above by
\begin{align}
A&\leq1+\textsc{\protect\small IF}(h/\varrho_{\textsc{ex}},\widetilde{Q}_{\textsc{ex}})-2\gamma \iint \frac{\omega \left(  1-\lambda \right)
}{\left(  1-\omega \right)  }\widetilde{e}_{\textsc{z}}(\mathrm{d}%
\lambda)\widetilde{e}_{\textsc{ex}}(\mathrm{d}\omega)\nonumber \\
&  =1+\textsc{\protect\small IF}(h/\varrho_{\textsc{ex}},\widetilde{Q}%
_{\textsc{ex}})-\gamma \left \{  1-\int \lambda \widetilde{e}%
_{\textsc{z}}(\mathrm{d}\lambda)\right \}  \int \frac{2\omega}{\left(
1-\omega \right)  }\widetilde{e}_{\textsc{ex}}(\mathrm{d}\omega
)\nonumber \\
&  =1+\textsc{\protect\small IF}(h/\varrho_{\textsc{ex}},\widetilde{Q}%
_{\textsc{ex}})-\gamma(1-\phi_{\textsc{z}})\left \{  \textsc{\protect\small IF}(h/\varrho_{\textsc{ex}},\widetilde{Q}_{\textsc{ex}})  -1\right \} \nonumber \\
&  =\left \{  1+\textsc{\protect\small IF}(h/\varrho_{\textsc{ex}},\widetilde{Q}_{\textsc{ex}})  \right \}  \left \{  \phi_{\textsc{z}%
}+(1-\phi_{\textsc{z}})(1-\gamma)\right \}  +2(1-\phi_{\textsc{z}%
})\gamma \label{eq:inequal_one}\\
&  \leq \left \{  1+\textsc{\protect\small IF}(h/\varrho_{\textsc{ex}},\widetilde{Q}_{\textsc{ex}})  \right \}  \left \{  \phi
_{\textsc{z}}+(1-\phi_{\textsc{z}})(1-\gamma)+2(1-\phi
_{\textsc{z}})\gamma \right \} \nonumber \\
&  =\left \{  1+\textsc{\protect\small IF}(h/\varrho_{\textsc{ex}},\widetilde{Q}_{\textsc{ex}})  \right \}  \left \{  2(1-\phi_{\textsc{z}%
}/2)-\frac{(1-\phi_{\textsc{z}})}{\pi_{\textsc{z}}\left(  1/\varrho
_{\textsc{z}}\right)  \pi_{\textsc{z}}\left(  \varrho_{\textsc{z}}\right)
}\right \},\nonumber
\end{align}
where we have used the identity $\phi_{\textsc{{z}}}=\int \lambda \widetilde{e}_{\textsc{z}}(\mathrm{d}\lambda)$. The last inequality is established by noting that $\textsc{\protect\small IF}(h/\varrho_{\textsc{ex}},\widetilde{Q}_{\textsc{ex}})$ and $\gamma$ are non-negative. Substituting the expression into (\ref{eq:RIF_equal}) establishes Part 1.
To establish the inequality of Part 2, we note that if $\textsc{\protect\small IF}(
h/\varrho_{\textsc{ex}},\widetilde{Q}_{\textsc{ex}})  \geq
1$, then (\ref{eq:inequal_one}) is bounded from above by
\begin{align*}
&  \left \{  1+\textsc{\protect\small IF}(h/\varrho_{\textsc{ex}},\widetilde{Q}_{\textsc{ex}})  \right \}  \left[  \phi_{\textsc{z}%
}+\frac{(1-\phi_{\textsc{z}})}{\pi_{\textsc{z}}\left(  1/\varrho_{\textsc{{z}%
}}\right)  \pi_{\textsc{z}}\left(  \varrho_{\textsc{z}}\right)  }+(1-\phi
_{\textsc{z}})\left \{  1-\frac{1}{\pi_{\textsc{z}}\left(  1/\varrho
_{\textsc{z}}\right)  \pi_{\textsc{z}}\left(  \varrho_{\textsc{z}}\right)
}\right \}  \right] \\
&  =  1+\textsc{\protect\small IF}(h/\varrho_{\textsc{ex}},\widetilde{Q}_{\textsc{ex}})   .
\end{align*}
\end{proof}
\vskip 1pt
\begin{proof}
[Proof of Corollary \ref{corollary:boundsQexjump}] We establish the upper
bound $\textsc{\protect\small uRIF}_{3}\left(  h\right)$ of Part 1 by first noting that (\ref{eq:inequal_one}) implies
\begin{equation*}
A\leq \left \{  1+\textsc{\protect\small IF}(h/\varrho_{\textsc{ex}},\widetilde{Q}_{\textsc{ex}})  \right \}  \left \{  \phi_{\textsc{z}%
}+(1-\phi_{\textsc{z}})(1-\gamma)\right \}  +2(1-\phi_{\textsc{z}%
})\gamma,
\end{equation*}
with $A$ is the quantity in (\ref{eq:identityB}), $\gamma$ given by (\ref{eq:identitygamma}) and $\phi_{\textsc{{z}%
}}=\int \lambda \widetilde{e}_{\textsc{z}}(\mathrm{d}\lambda)$.
Upon
substituting into (\ref{eq:RIF_equal}), we obtain
\begin{align*}
\textsc{\protect\small RIF}\left(  h,Q^{\ast}\right)  +\frac{1}{\textsc{\protect\small IF}%
(h,Q_{\textsc{ex}})} &  \leq \frac{\left \{  1+\textsc{\protect\small IF}%
(h,Q_{\textsc{ex}})\right \}  }{\textsc{\protect\small IF}(h,Q_{\textsc{ex}}%
)}\left \{  \phi_{\textsc{z}}\pi_{\textsc{z}}\left(  \varrho_{\textsc{z}%
}^{-1}\right)  +\frac{(1-\phi_{\textsc{z}})}{\pi_{\textsc{z}}\left(
\varrho_{\textsc{z}}\right)  }\right \}  \\
&\qquad  +\frac{2(1-\phi_{\textsc{z}})\left \{  1+\textsc{\protect\small IF}%
(h,Q_{\textsc{ex}})\right \}  }{\textsc{\protect\small IF}(h,Q_{\textsc{ex}%
})\left \{  1+\textsc{\protect\small IF}(h/\varrho_{\textsc{ex}},\widetilde{Q}_{\textsc{ex}})  \right \}  }\left \{  \pi_{\textsc{z}}\left(
\varrho_{\textsc{z}}^{-1}\right)  -\frac{1}{\pi_{\textsc{z}}\left(
\varrho_{\textsc{z}}\right)  }\right \},
\end{align*}
and, after further manipulations,
\begin{align*}
\textsc{\protect\small RIF}\left(  h,Q^{\ast}\right)   &  \leq \phi_{\textsc{z}%
}\left \{  \pi_{\textsc{z}}\left(  1/\varrho_{\textsc{z}}\right)  -1/\pi
_{\textsc{z}}\left(  \varrho_{\textsc{z}}\right)  \right \}  +1/\pi_{\textsc{z}}\left(
\varrho_{\textsc{z}}\right)  \\
&\qquad  +\frac{1}{\textsc{\protect\small IF}({h},Q_{\textsc{ex}})}\left[  \phi
_{\textsc{z}}\left \{  \pi_{\textsc{z}}\left(  1/\varrho_{\textsc{z}%
}\right)  -1/\pi_{\textsc{z}}\left(  \varrho_{\textsc{z}}\right)  \right \}
+\frac{1}{\pi_{\textsc{z}}\left(  \varrho_{\textsc{z}}\right)  }-1\right]  \\
&\qquad\qquad  +2\frac{\{1+1/\textsc{\protect\small IF}({h},Q_{\textsc{ex}})\}}{1+\textsc{\protect\small IF}%
(h/\varrho_{\textsc{ex}},\widetilde{Q}_{\textsc{ex}})}\left \{
\pi_{\textsc{z}}\left(  1/\varrho_{\textsc{z}}\right)  -1/\pi_{\textsc{z}}\left(
\varrho_{\textsc{z}}\right)  \right \}  (1-\phi_{\textsc{z}})\\
&  \leq \phi_{\textsc{z}}\left \{  \pi_{\textsc{z}}\left(  1/\varrho
_{\textsc{z}}\right)  -1/\pi_{\textsc{z}}\left(  \varrho_{\textsc{z}}\right)  \right \}  +1/\pi_{\textsc{z}}\left(  \varrho_{\textsc{z}}\right)\\
&\qquad  +\frac{1}{\textsc{\protect\small IF}(h/\varrho_{\textsc{ex}},\widetilde
{Q}_{\textsc{ex}})}\left[  \phi_{\textsc{z}}\left \{  \pi
_{\textsc{z}}\left(  1/\varrho_{\textsc{z}}\right)  -1/\pi_{\textsc{z}}\left(
\varrho_{\textsc{z}}\right)  \right \}  +\frac{1}{\pi_{\textsc{z}}\left(
\varrho_{\textsc{z}}\right)  }-1\right]  \\
&\qquad\qquad  +\frac{2}{\textsc{\protect\small IF}(h/\varrho_{\textsc{ex}},\widetilde
{Q}_{\textsc{ex}})}\left \{  \pi_{\textsc{z}}\left(  1/\varrho_{\textsc{{z}%
}}\right)  -1/\pi_{\textsc{z}}\left(  \varrho_{\textsc{z}}\right)  \right \}
(1-\phi_{\textsc{z}}),
\end{align*}
as $\textsc{\protect\small IF}(h/\varrho_{\textsc{ex}},\widetilde{Q}_{\textsc{ex}})\leq \textsc{\protect\small IF}({h},Q_{\textsc{ex}})$ from Proposition 2.

To establish the upper bound $\textsc{\protect\small uRIF}_{4}\left(  h\right)  $ of Part 2, we use that, in the
right hand side of the equality of (\ref{eq:RIF_equal}), the term $\beta$ defined in (\ref{eq:betadefn}) and appearing in $A$ satisfies the inequality%
\begin{equation}
\beta=\iint \frac{2}{(1-\lambda \omega)}\widetilde{e}_{\textsc{z}%
}(\mathrm{d}\lambda)\widetilde{e}_{\textsc{ex}}(\mathrm{d}\omega
)\leq \int \frac{2}{(1-\lambda)}\widetilde{e}_{\textsc{z}}(\mathrm{d}%
\lambda)=1+\textsc{\protect\small IF}(1/\varrho_{\textsc{z}},\widetilde{Q}%
_{\textsc{z}}),\label{eq:beta_inequal}%
\end{equation}
where $\textsc{\protect\small IF}(1/\varrho_{\textsc{z}},\widetilde{Q}%
_{\textsc{z}})=\int(1+\lambda)/(1-\lambda)\widetilde{e}_{\textsc{z}}(\mathrm{d}\lambda)<\infty$, by assumption. Therefore, upon substituting
into (\ref{eq:RIF_equal}), we obtain
\begin{align*}
\textsc{\protect\small RIF}\left(  h,Q^{\ast}\right)
&\leq \frac{\pi_{\textsc{z}}\left(
\varrho_{\textsc{z}}^{-1}\right)  \left \{  1+\textsc{\protect\small IF}%
(h,Q_{\textsc{ex}})\right \}  }{\textsc{\protect\small IF}(h,Q_{\textsc{ex}%
})\{1+\textsc{\protect\small IF}(h/\varrho_{\textsc{ex}},\widetilde{Q}%
_{\textsc{ex}})\}}\left \{  1-\frac{1}{\pi_{\textsc{z}}\left(  \varrho
_{\textsc{z}}^{-1}\right)  \pi_{\textsc{z}}\left(  \varrho_{\textsc{z}%
}\right)  }\right \}  \{1+\textsc{\protect\small IF}(1/\varrho_{\textsc{z}}%
,\widetilde{Q}_{\textsc{z}})\} \\
&\qquad  +\frac{\left \{
1+\textsc{\protect\small IF}(h,Q_{\textsc{ex}})\right \}  }{\textsc{\protect\small IF}%
(h,Q_{\textsc{ex}})}\frac{1}{\pi_{\textsc{z}}\left(  \varrho_{\textsc{z}}\right)  }-\frac
{1}{\textsc{\protect\small IF}(h,Q_{\textsc{ex}})}\\
&  =\frac{1}{\pi_{\textsc{z}}\left(  \varrho_{\textsc{z}}\right)  }%
+\frac{\left \{  1+1/\textsc{\protect\small IF}(h,Q_{\textsc{ex}})\right \}
}{1+\textsc{\protect\small IF}(h/\varrho_{\textsc{ex}},\widetilde{Q}%
_{\textsc{ex}})}\left \{  \pi_{\textsc{z}}\left(  \varrho_{\textsc{z}%
}^{-1}\right)  -1/\pi_{\textsc{z}}\left(  \varrho_{\textsc{z}}\right)  \right \}
\{1+\textsc{\protect\small IF}(1/\varrho_{\textsc{z}},\widetilde{Q}_{\textsc{z}%
})\} \\
&\qquad  +\frac{1}{\textsc{\protect\small IF}(h,Q_{\textsc{ex}})}\left \{  \frac{1}{\pi
_{\textsc{z}}\left(  \varrho_{\textsc{z}}\right)  }-1\right \}  \\
&  \leq \frac{1}{\pi_{\textsc{z}}\left(  \varrho_{\textsc{z}}\right)  }%
+\frac{\left \{  1+1/\textsc{\protect\small IF}(h/\varrho_{\textsc{ex}},\widetilde
{Q}_{\textsc{ex}})\right \}  }{1+\textsc{\protect\small IF}(h/\varrho
_{\textsc{ex}},\widetilde{Q}_{\textsc{ex}})}\left \{  \pi
_{\textsc{z}}\left(  \varrho_{\textsc{z}}^{-1}\right)  -1/\pi_{\textsc{z}}\left(
\varrho_{\textsc{z}}\right)  \right \}  \{1+\textsc{\protect\small IF}(1/\varrho
_{\textsc{z}},\widetilde{Q}_{\textsc{z}})\} \\
& \qquad +\frac{1}{\textsc{\protect\small IF}(h/\varrho_{\textsc{ex}},\widetilde
{Q}_{\textsc{ex}})}\left \{  \frac{1}{\pi_{\textsc{z}}\left(  \varrho
_{\textsc{z}}\right)  }-1\right \}  ,
\end{align*}
as $\textsc{\protect\small IF}(h/\varrho_{\textsc{ex}},\widetilde{Q}_{\textsc{ex}})
\leq \textsc{\protect\small IF}(h,Q_{\textsc{ex}})$.

To establish the inequality of Part 3, we combine (\ref{eq:identityB}) and (\ref{eq:RIF_equal}) to obtain%
\begin{align}
\textsc{{\small RIF}}\left(  h,Q^{\ast}\right)
&=\frac{\pi_{\textsc{z}}\left(
\varrho_{\textsc{z}}^{-1}\right)  \left\{  1+1/\textsc{{\small IF}%
}(h,Q_{\textsc{ex}})\right\}  }{1+\textsc{{\small IF}}(h/\varrho_{\textsc{ex}%
},\widetilde{Q}_{\textsc{ex}})}\gamma\beta+(1-\gamma)\pi_{\textsc{z}}\left(
\varrho_{\textsc{z}}^{-1}\right)  \left\{  1+1/\textsc{{\small IF}%
}(h,Q_{\textsc{ex}})\right\}  -\frac{1}{\textsc{{\small IF}}(h,Q_{\textsc{ex}%
})}\notag\\
&=\frac{1}{\pi_{\textsc{z}}\left(  \varrho_{\textsc{z}}\right)  }%
+\frac{\left\{  1+1/\textsc{{\small IF}}(h,Q_{\textsc{ex}})\right\}
}{1+\textsc{{\small IF}}(h/\varrho_{\textsc{ex}},\widetilde{Q}_{\textsc{ex}}%
)}\left\{  \pi_{\textsc{z}}\left(  \varrho_{\textsc{z}}^{-1}\right)
-1/\pi_{\textsc{z}}\left(  \varrho_{\textsc{z}}\right)  \right\}  \beta
+\frac{\{1/\pi_{\textsc{z}}\left(  \varrho_{\textsc{z}}\right)  -1\}}%
{\textsc{{\small IF}}(h,Q_{\textsc{ex}})}\notag\\
&\geq\frac{1}{\pi_{\textsc{z}}\left(  \varrho_{\textsc{z}}\right)  }%
+\frac{2\left\{  1+1/\textsc{{\small IF}}(h,Q_{\textsc{ex}})\right\}
}{1+\textsc{{\small IF}}(h/\varrho_{\textsc{ex}},\widetilde{Q}_{\textsc{ex}}%
)}\left\{  \pi_{\textsc{z}}\left(  \varrho_{\textsc{z}}^{-1}\right)
-1/\pi_{\textsc{z}}\left(  \varrho_{\textsc{z}}\right)  \right\}
+\frac{\{1/\pi_{\textsc{z}}\left(  \varrho_{\textsc{z}}\right)  -1\}}%
{\textsc{{\small IF}}(h,Q_{\textsc{ex}})}\label{eq:rifineqthirdline}\\
&\geq\frac{1}{\pi_{\textsc{z}}\left(  \varrho_{\textsc{z}}\right)  }+\frac
{2}{1+\textsc{{\small IF}}(h/\varrho_{\textsc{ex}},\widetilde{Q}_{\textsc{ex}%
})}\left\{  \pi_{\textsc{z}}\left(  \varrho_{\textsc{z}}^{-1}\right)
-1/\pi_{\textsc{z}}\left(  \varrho_{\textsc{z}}\right)  \right\}\notag.
\end{align}
The first inequality follows because the identity for $\beta$ given in (\ref{eq:beta_inequal}) shows that $\beta\geq2$ when $\widetilde{Q}_{\textsc{ex}}$ is positive.
The second inequality follows from $\textsc{{\small IF}}(h,Q_{\textsc{ex}})\geq0$.

From (\ref{eq:RIF_equal}), we have $\textsc{{\small RIF}}\left(  h,Q^{\ast}\right)\geq1/{\pi_{\textsc{z}}\left(  \varrho_{\textsc{z}}\right)  }$
as the second and third terms on the left hand side of the inequality~\eqref{eq:rifineqthirdline} are both
positive. This establishes the inequality of Part 4. We examine the limit of $\textsc{{\small RIF}}\left(  h,Q^{\ast
}\right)  $ as $\textsc{{\small IF}}(h/\varrho_{\textsc{ex}},\widetilde{Q}%
_{\textsc{ex}})\rightarrow\infty$, again noting that $\textsc{{\small IF}%
}(h/\varrho_{\textsc{ex}},\widetilde{Q}_{\textsc{ex}})\leq\textsc{{\small IF}%
}(h,Q_{\textsc{ex}})$. Using the inequality for $\beta$ given by
(\ref{eq:beta_inequal}) and the fact that $\textsc{{\small IF}}%
(1/\varrho_{\textsc{z}},\widetilde{Q}_{\textsc{z}})<\infty$ by Lemma~{3}, we obtain the
limiting form, as $\textsc{{\small IF}}(h/\varrho_{\textsc{ex}},\widetilde{Q}%
_{\textsc{ex}})\rightarrow\infty,$ given by (\ref{eq:lowerboundonRIFQ*}) for $\textsc{{\small RIF}}\left(  h,Q^{\ast}\right)  $.
\end{proof}
\section*{Appendix 2}
We exploit the two upper bounds ${\textsc{{\protect \small uRCT}}}_{3}(h;\sigma)$ and ${\textsc{{\protect \small uRCT}}}_{4}(h;\sigma)$, together with the lower bound ${\textsc{{\protect \small lRCT}}}_{1}(h;\sigma)$, in order to find an interval where the optimal value $\sigma_\text{opt}$ for ${\textsc{\protect\small RCT}}(h, Q^{\ast};\sigma)$ lies. We consider how this interval varies as $\textsc{\protect \small IF}(h/\varrho_\textsc{ex},\widetilde{Q}_{\textsc{ex}})$ increases. To do this, we compute the interval where ${\textsc{{\protect \small lRCT}}}_{1}(h;\sigma)$ lies below the minimum of $\inf_{\sigma}{\textsc{{\protect \small uRCT}}}_{3}(h;\sigma)$, and $\inf_{\sigma}{\textsc{{\protect \small uRCT}}}_{4}(h;\sigma)$. Table~\ref{tab:sandwich} displays this interval together with the minimum of the two upper bounds and the minimum of the lower bound. It is straightforward to see that $\sigma_\text{opt}$ is contained in this interval and ${\textsc{\protect\small RCT}}(h, Q^{\ast};\sigma_\text{opt})$ is contained in the corresponding interval in Table~\ref{tab:sandwich}. It is apparent that the intervals tighten as $\textsc{\protect \small IF}(h/\varrho_\textsc{ex},\widetilde{Q}_{\textsc{ex}})$ increases. Similarly the endpoints of the interval containing $\textsc{\protect\small RCT}(h, Q^\ast;\sigma_\text{opt})$ both decrease whilst the lower endpoint of the interval containing $\sigma_\text{opt}$ increases.
\begin{table}[!h]
\captionsetup{font=small, labelsep=period}
\caption{\textsl{{Sandwiching results based upon different values of
}$\textsc{{\protect \small IF}}(h/\varrho_{\textsc{ex}},\widetilde
{Q}_{\textsc{ex}})$. {These are based upon the upper bounds for
}$\textsc{{\protect \small RCT}}({h},Q^{\ast};\sigma)$ {given by
}${\textsc{{\protect \small uRCT}}}_{3}(h;\sigma)$ {and }%
${\textsc{{\protect \small uRCT}}}_{4}(h;\sigma)$ {and upon the lower
bound }${\textsc{{\protect \small lRCT}}}_{1}(h;\sigma)$.}}
\vskip -15pt
\begin{center}\small
\begin{tabular}{cccccc}
$\textsc{\protect \small IF}(h/\varrho_\textsc{ex},\widetilde{Q}_{\textsc{ex}})$ & 1 & 10 & 25 & 100 & 1000\\
$\textsc{\protect\small RCT}(h, Q^\ast;\sigma_\text{opt})$ & (3.201, 5.327) & (2.020, 2.256) & (1.773, 1.876) &
(1.595, 1.625) & (1.518, 1.522)\\
$\sigma_\text{opt}$  & (0.548, 1.572) & (1.018, 1.598) & (1.205, 1.658) &
(1.421, 1.730) & (1.607, 1.730)
\end{tabular}
\end{center}
\label{tab:sandwich}%
\vskip -12pt
\end{table}
{\center\huge\textbf{Supplementary Material}}
\section{Contents}
This supplement provides some technical proofs and an additional example
for the paper \textquotedblleft Efficient implementation of Markov chain Monte
Carlo when using an unbiased likelihood estimator\textquotedblright. Section
\ref{Sec:ProofofProposition2} presents the proof of
Proposition 2. Section \ref{sec:technicalresultsmainpaper} presents the
proofs of Propositions 3 and 4 and Lemmas 1 and 3. Section \ref{sec:technicalresultssupplementary} presents some auxiliary
technical results. Section \ref{section:asymptoticupperbound}
illustrates the upper bound on the inefficiency of Part 4 of Corollary 2 and
compares it to the results in \cite{Sherlock2013efficiency}. Section
\ref{SS: AR plus noise} applies the pseudo-marginal algorithm to a linear Gaussian state-space model and presents additional simulation results for the stochastic volatility model discussed in the main paper.
Section \ref{Sec:Numericalprocedures} explains how the bounds on the inefficiency introduced in Section 3.5 of the main paper are computed.

All code was implemented in the \texttt{Ox} language with
pre-compiled \texttt{C} code for computationally intensive routines.

%
%
\section{Proof of Proposition~2\label{Sec:ProofofProposition2}}
The proof of Proposition~2 relies on Lemmas 5 to 8, which are given below. Lemmas 5 to 7 establish that $h/\varrho \in L^2(\mathsf{X},\tilde{\mu})$ and $\textsc{\protect\small IF}({h/\varrho},{\tP})<\infty$ whenever $\textsc{\protect\small IF}(h,P)<\infty$. To prove this result, we define the map that sends the functional $h$ to $h/\varrho$ as a linear operator between two Hilbert spaces, $\mathcal{H}$ and $\tilde{\mathcal{H}}$ defined below. The space $\mathcal{H}$, respectively  $\tilde{\mathcal{H}}$, corresponds to the set of functions having finite inefficiencies under $P$, respectively under $\tilde{P}$. We then exploit the structure of the Metropolis--Hastings type kernel $P$ to prove that this linear operator is bounded on a dense subspace $\mathcal{H}_P\subset \mathcal{H}$, which allows us to extend the operator to $\mathcal{H}$. The proof is then completed by checking that the unique extension constructed this way is the one required.
Lemma 8 is a general result on the central limit theorem for reversible and ergodic Markov chains which are not started in their stationary regime. The proof of Proposition 2 uses these preliminary results to establish the identity of interest.

Using the notation
of Proposition~2, we write $\Vert \cdot \Vert_{\mu}$, $\langle \cdot,\cdot
\rangle_{\mu}$ for the norm and inner product of $L^{2}(\mathsf{X,}\mu)$, with a similar notation for $L^{2}(\mathsf{X,}\tilde{\mu})$.
By reversibility of $P$ and $\tilde{P}$ with respect to $\mu$ and $\tilde{\mu}$ respectively, it is easy to check that $(I-P)$ and $(I-\tilde{P})$ are positive, self-adjoint operators on $L^{2}(\mathsf{X,}\mu)$ and
 $L^{2}(\mathsf{X,}\tilde{\mu})$ respectively. By Theorem~13.11 in \cite{Rud91}, the inverses $(I-P)^{-1}$ and $(I-\tilde{P})^{-1}$ are densely defined and self-adjoint. They are also positive, since for any $f\in \mathrm{Domain}\{ (I-P)^{-1}\}$, there exists a function $g$ such that $f= (I-P)g$, and thus
$$\langle (I-P)^{-1}f, f\rangle_{\mu} = \langle (I-P)^{-1} (I-P)g, (I-P)g\rangle_{\mu} =
\langle g, (I-P)g\rangle_{\mu} \geq 0,$$
since $I-P$ is positive.
Therefore, by Theorem~13.31 in \cite{Rud91}, there exists a unique, self-adjoint, positive operator
$(I-P)^{-1/2}$ such that $(I-P)^{-1} = (I-P)^{-1/2} (I-P)^{-1/2}$.
Finally, since $(I-P)^{-1}$ is densely defined, so is $(I-P)^{-1/2}$.
 Similar considerations show the existence and uniqueness of the positive, self-adjoint operator $(I-\tilde{P})^{-1/2}$, which is densely defined on $L^{2}(\mathsf{X,}\tilde{\mu})$.

We now introduce the inner product
spaces $(\mathcal{H},\langle \cdot,\cdot \rangle_{\mathcal{H}})$ and
$(\tilde{\mathcal{H}},\langle \cdot,\cdot \rangle_{\tilde{\mathcal{H}}})$, where
\begin{align*}
\mathcal{H}  &  =\{f\in L_{0}^{2}(\mathsf{X,}\mu):\Vert f\Vert_{\mu}^{2}%
+\Vert(I-P)^{-1/2}f\Vert_{\mu}^{2}<\infty \},\\
\langle f,g\rangle_{\mathcal{H}}  &  =\langle f,g\rangle_{\mu}+\langle
(I-P)^{-1/2}f,(I-P)^{-1/2}g\rangle_{\mu},\\
\tilde{\mathcal{H}}  &  =\{f\in L_{0}^{2}(\mathsf{X,}\tilde{\mu}):\Vert f\Vert
_{\tilde{\mu}}^{2}+\Vert(I-\widetilde{P})^{-1/2}f\Vert_{\tilde{\mu}}%
^{2}<\infty \},\\
\langle f,g\rangle_{\tilde{\mathcal{H}}}  &  =\langle f,g\rangle_{\tilde{\mu}}%
+\langle(I-\widetilde{P})^{-1/2}f,(I-\widetilde{P})^{-1/2}g\rangle_{\tilde
{\mu}}.
\end{align*}
 Clearly the space $\mathcal{H}$, respectively  $\tilde{\mathcal{H}}$, corresponds to the set of functions having finite inefficiencies under $P$, respectively under $\tilde{P}$.

\begin{lemma}\label{lem:hilbert}
Let $P$ and $\widetilde{P}$ be ergodic. Then $(\mathcal{H},\langle \cdot
,\cdot \rangle_{\mathcal{H}})$ and $(\tilde{\mathcal{H}},\langle \cdot,\cdot
\rangle_{\tilde{\mathcal{H}}})$ are Hilbert spaces.
\end{lemma}

\begin{proof}
Since $P$ and $\widetilde{P}$ are ergodic, the only solutions in
$L^{2}(\mathsf{X,}\mu)$ and $L^{2}(\mathsf{X,}\tilde{\mu})$, of
$h=Ph$, respectively $g=\widetilde{P}g$, are almost surely constant with
respect to $\mu$ and $\tilde{\mu}$.
If $f=Pf$ $\mu-$almost surely, then
$$0=\Vert f-Pf\Vert_{\mu}^{2}=\int_{-1}^1 (1-\lambda)^2 e(f,P)(\mathrm{d}\lambda),$$
where $e(f,P)$ is the spectral measure of $P$ with respect to the function $f$,
and therefore $e(f,P)$ must be an atom at 1, which is impossible as $P$ is ergodic; see the proof of Lemma 17 in \citet{rosenthalCLT2007} and Proposition~17.4.1 in \citet{MT09}.
Since $I-P$ and $I-\tilde{P}$ are injective in $L_{0}^{2}(\mu)$
and $L_{0}^{2}(\tilde{\mu})$ respectively,
$(I-P)^{1/2}$ and $(I-\widetilde{P})^{1/2}$ must also be injective on the corresponding spaces, because $(I-P)^{1/2}h=0$ implies $(I-P)h=0$.
In addition, as mentioned above, these operators are
self-adjoint and thus their inverses, $(I-P)^{-1/2}$ and $(I-\widetilde
{P})^{-1/2}$, are densely defined and self-adjoint by Theorem~13.11 in \cite%
{Rud91}.

By Theorem~13.9 in \cite{Rud91}, $(I-P)^{-1/2}$ and $(I-\widetilde
{P})^{-1/2}$ are \textit{closed operators} on
$L_{0}^{2}(\mathsf{X,}\mu)$ and $L_{0}^{2}(\mathsf{X,}\tilde{\mu})$
respectively because they are self-adjoint. By Section~13.1 in \citet{Rud91}, a possibly unbounded operator $T$ on a Hilbert space
$\mathcal{F}$ is said to be closed if and only if
its graph
$$\mathfrak{G}(T)=\{ (x,Tx): x\in \mathcal{F} \},$$
is a closed subset of $\mathcal{F}\times \mathcal{F}$. Equivalently $T$ is closed if
$x_{n}\rightarrow x$ and $Tx_{n}%
\rightarrow y$ implies $Tx=y$. In particular, $x$ is in the domain of $T$.
It follows that $(\mathcal{H},\langle \cdot,\cdot \rangle
_{\mathcal{H}})$ and $(\tilde{\mathcal{H}},\langle \cdot,\cdot \rangle_{\tilde{\mathcal{H}}})$
are Hilbert spaces by Proposition~1.4 in \citet{Schm12}.
\end{proof}

\begin{lemma}
The linear space
\[
\mathcal{H}_{P}=\mathrm{Range}\big \{(I-P)\big \}=\{h\in L_{0}^{2}%
(\mathsf{X},\mu):h=(I-P)g,\, \,g\in L^{2}(\mathsf{X},\mu)\}
\]
is dense in $\mathcal{H}$ in the norm induced by $\langle \cdot, \cdot \rangle_\mathcal{H}$.
\end{lemma}

\begin{proof}
For $h\in \mathcal{H}$, we have
\[
\|(I-P)^{-1/2} h\|_{\mu}=\int_{-1}^{1}\frac{e(h,P)(\mathrm{d}\lambda)}{1-\lambda}<\infty,
\]
where $e(h,P)$ is the spectral measure associated with
$h$ and $P$. For $\epsilon>0$, define
\[
h_{\epsilon}=(I-P)\big \{(1+\epsilon)I-P\big \}^{-1}h\in \mathcal{H}_{P}.
\]
Then,
\begin{align*}
\Vert(I-P)^{-1/2}(h_{\epsilon}-h)\Vert_{\mu}^{2}
&= \Big\|(I-P)^{-1/2}%
\Big[ (I-P)\{(1+\epsilon)I-P\}^{-1}-I\Big]h\Big\|_{\mu}^{2}\\
&  =\int_{-1}^{1}\frac{1}{1-\lambda}\Big(  \frac{1-\lambda}{1+\epsilon-\lambda}-1\Big)  ^{2} e(h,P)(\mathrm{d}\lambda)\\
&  =\int_{-1}^{1}(1-\lambda)\left(  \frac{1}{1+\epsilon-\lambda}-\frac
{1}{1-\lambda}\right)  ^{2}e(h,P)(\mathrm{d}\lambda)\\
&  =\int_{-1}^{1}\frac{\epsilon^{2} e(h,P)(\mathrm{d}\lambda)}{(1+\epsilon-\lambda)^{2}(1-\lambda)}.
\end{align*}
The integrand is bounded above by $1/(1-\lambda)$, since $|\lambda|\leq 1$ implies that $\epsilon^2/(1+\epsilon -\lambda)^2 \leq 1$,  and thus, by dominated
convergence, the integral vanishes as $\epsilon \rightarrow0$. Since $I-P$
is bounded, $\Vert h_{\epsilon}-h\Vert_{\mu}$ also vanishes. Therefore, $h_{\epsilon}\rightarrow h$ in $\mathcal{H}$. In particular,
$\mathcal{H}_{P}$ is dense in $\mathcal{H}$.
\end{proof}

\begin{lemma}\label{lem:finiteIF}
If $\textsc{\protect\small IF}(h,P)<\infty$, then $h/\varrho \in L^2(\mathsf{X},\tilde{\mu})$ and $\textsc{\protect\small IF}({h/\varrho},{\tP})<\infty$.
\end{lemma}

\begin{proof}
For $h\in \mathcal{H}_{P}$, there exists $g\in L^{2}(\mathsf{X,}\mu)$ such
that
\[
h(x)=(I-P)g(x)=\varrho(x)(I-\widetilde{P})g(x).
\]
Therefore, $h(x)/\varrho(x)=(I-\widetilde{P})g(x)\in \tilde{\mathcal{H}}$, since
$\Vert g\Vert_{\tilde{\mu}}^{2}\leq \Vert g\Vert_{\mu}^{2}/\mu(\varrho)$. Thus, we can define the multiplication operator $T:\mathcal{H}_{P}\rightarrow
\tilde{\mathcal{H}}$ by $T:h\rightarrow h/\varrho$.

Let $h\left(  x\right)  =(I-P)g(x)$. Then,
\[
\Vert h\Vert_{\mathcal{H}}^{2}=\Vert h\Vert_{\mu}^{2}+\langle h,(I-P)^{-1}%
(I-P)g\rangle_{\mu}\geq \langle h,g\rangle_{\mu},
\]
because $I-P$ is self-adjoint. Similarly,
\begin{align*}
\Vert Th\Vert_{\tilde{\mathcal{H}}}^{2}  &  =\Vert h/\varrho \Vert_{\tilde{\mu}}%
^{2}+\Vert(I-\widetilde{P})^{-1/2}(h/\varrho)\Vert_{\tilde{\mu}}^{2}\\
&  =\Vert(I-\widetilde{P})g\Vert_{\tilde{\mu}}^{2}+\Vert(I-\widetilde
{P})^{1/2}g\Vert_{\tilde{\mu}}^{2}\leq K\Vert(I-\widetilde{P})^{1/2}%
g\Vert_{\tilde{\mu}}^{2},%
\end{align*}
where $K=1+\Big \Vert (  I-\widetilde{P})  ^{1/2} \Big\Vert $ with $ \Vert (  I-\widetilde{P})  ^{1/2} \Vert $ the
finite norm of the operator $(  I-\widetilde{P})  ^{1/2}$.
Recalling that $h(x)=(I-P)g(x)=\varrho(x)(I-\widetilde{P})g(x)$, we obtain
\[
\Vert(I-\widetilde{P})^{1/2}g\Vert_{\tilde{\mu}}^{2}=\int g(x)(I-\widetilde
{P})g(x)\frac{\varrho(x)\mu(\mathrm{d}x)}{\mu(\varrho)}=\frac{\langle
g,h\rangle_{\mu}}{\mu(\varrho)}.
\]
It follows that $T:\mathcal{H}_{P}\rightarrow \tilde{\mathcal{H}}$ is
bounded as
\[
\sup_{h\in \mathcal{H}_{P}}\frac{\Vert Th\Vert_{\tilde{\mathcal{H}}}^{2}}{\Vert
h\Vert_{\mathcal{H}}^{2}}\leq \frac{K\Vert(I-\widetilde{P})^{1/2}g\Vert
_{\tilde{\mu}}^{2}}{\Vert h\Vert_{\mu}^{2}+\langle h,g\rangle_{\mu}}=\frac
{K}{\mu(\varrho)}\frac{\langle g,h\rangle_{\mu}}{\Vert h\Vert_{\mu}%
^{2}+\langle g,h\rangle_{\mu}}\leq \frac{K}{\mu(\varrho)}.
\]

Since $\mathcal{H}_{P}$ is dense, given $h\in \mathcal{H}$, there is a sequence
$h_{n}\in \mathcal{H}_{P}$ such that $\Vert h_{n}-h\Vert_{\mathcal{H}%
}\rightarrow0$, as $n\rightarrow \infty$. This, in particular, implies that
$h_{n}$ is a Cauchy sequence in $\mathcal{H}$, that is
\[
\Vert h_{n}-h_{m}\Vert_{\mathcal{H}}\rightarrow0,\quad \text{as }n\geq
m\rightarrow \infty.
\]
Since $h_{n}$ and $h_{n}-h_{m}$ are in $\mathcal{H}_{P}$, $Th_{n}%
,T(h_{n}-h_{m})\in \tilde{\mathcal{H}}$ and, from the above calculation,
\[
\Vert Th_{n}-Th_{m}\Vert_{\tilde{\mathcal{H}}}\leq \frac{K}{\mu(\varrho)}\Vert
h_{n}-h_{m}\Vert_{\mathcal{H}}\rightarrow0,
\]
as $m,n\rightarrow \infty$. Therefore, $Th_{n}$ forms a Cauchy sequence in
$\tilde{\mathcal{H}}$; in particular $h_{n}$ and $(I-\widetilde{P})^{-1/2}h_{n}$ are
Cauchy in $L^{2}(\mathsf{X,}\tilde{\mu})$. Since $L^{2}(\mathsf{X,}\tilde{\mu
})$ is complete, we have $h_{n}\rightarrow g\in L^{2}(\mathsf{X,}%
\tilde{\mu})$ and $(I-\widetilde{P})^{-1/2}h_{n}\rightarrow f\in
L^{2}(\mathsf{X,}\tilde{\mu})$. Since $Q=(I-\widetilde{P})^{-1/2}$ is a closed
operator, we can conclude that
\[
g\in \mathrm{Domain}\left \{  Q\right \}  ,\quad Qg=f,
\]
and, in particular, $g\in \tilde{\mathcal{H}}$.

To complete the proof, we need to show that $g=h/\varrho$. Recall that
$h_{n}\rightarrow h$ in $\mathcal{H}$ implies that $\Vert h_{n}-h\Vert_{\mu
}\rightarrow0$. We can then choose a subsequence $n(k)$ such that
$h_{n(k)}\rightarrow h$ $\mu$-almost surely. Since $\tilde{\mu}$ is absolutely
continuous with respect to $\mu$, we also have $h_{n(k)}/\varrho \rightarrow
h/\varrho$ $\tilde{\mu}$-almost surely.

In addition, we know that $Th_{n}=h_{n}/\varrho \rightarrow g$ in $\tilde{\mathcal{H}}$
and thus in $L^{2}(\mathsf{X,}\tilde{\mu})$. Therefore, $h_{n(k)}/\varrho \rightarrow g$ in $L^{2}(\mathsf{X,}\tilde{\mu})$. We can now choose a
further subsequence $n^{\prime}(k)$ such that $h_{n^{\prime}(k)}%
/\varrho \rightarrow g$ $\tilde{\mu}$-almost surely. Since $h_{n(k)}/\varrho$
also converges to $h/\varrho$ $\tilde{\mu}$-almost surely, and $n^{\prime}(k)$
is a subsequence of $n(k)$, we conclude that $g=h/\varrho$ $\tilde{\mu}%
$-almost surely.
\end{proof}

\begin{lemma}
\label{Lemma:absolutecontinuityCLT}Assume $\Pi$ is $\mu
$-reversible and ergodic, $h\in L_{0}^{2}\left(  \mathsf{X},\mu \right)  $ and
$\textsc{\protect\small IF}(h,{\Pi})<\infty$. Let $\left(  X_{i}\right)  _{i\geq1}$ be a Markov chain
evolving according to $\Pi$. If $X_{1}\sim \nu$, where $\nu$ is absolutely
continuous with respect to $\mu$ then, as $n\rightarrow \infty$,
\[
n^{-1/2}\sum_{i=1}^{n}h(X_{i})\longrightarrow \mathcal{N}\left \{  0;\mu \left(
h^{2}\right)  \mathrm{IF}({h}, {\Pi})\right \}  .
\]
\end{lemma}

\begin{proof}
[Proof of Lemma~\ref{Lemma:absolutecontinuityCLT}]Let $e(h,\Pi)(\mathrm{d}%
\lambda)$ be the associated spectral measure and define $S_{n}=\sum_{i=1}%
^{n}h\left(  X_{i}\right)  $. Then,
\begin{align*}
\frac{1}{n}\mathrm{E}_{\mu}\left \{  \mathrm{E}\left(  \left.  S_{n}\right \vert
X_{1}\right)  ^{2}\right \}   &  =\frac{1}{n}\int \mu \left(  dx\right)  \left \{
\sum_{i=0}^{n-1}\Pi^{i}h\left(  x\right)  \right \}  ^{2}=\frac{1}{n}\int
_{-1}^{1}\left(  \sum_{i=0}^{n-1}\lambda^{i}\right)  ^{2}e(h,\Pi)(\mathrm{d}\lambda)\\
&  =\int_{-1}^{1}\left(  \frac{1}{n}\sum_{i=0}^{n-1}\lambda^{i}\right)
\frac{1-\lambda^{n}}{1-\lambda}e(h,\Pi)(\mathrm{d}\lambda)\rightarrow0
\end{align*}
as $n\rightarrow \infty$ by dominated convergence, since $\int \left(
1-\lambda \right)  ^{-1}e(h,\Pi)(\mathrm{d}\lambda)<\infty$ by assumption.
Hence, equation (4) in \cite{wuwoodroofe2004} holds with $\sigma_{n}%
^{2}=\mathrm{E}_{\mu}\left(  S_{n}^{2}\right)  \sim \sigma^{2}n$, where
$\sigma^{2}=\mu \left(  h^{2}\right)  \mathrm{IF}({h},{\Pi})$. It is
straightforward to check, with calculations similar to the above, that the
solution to the approximate Poisson equation given in the proof of Theorem 1.3
in \cite{kipnis1986central},%
\[
h_{n}\left(  x\right)  =\left \{  \left(  1+\frac{1}{n}\right)  I-\Pi \right \}
^{-1}h\left(  x\right)  ,
\]
satisfies equation (5) in \cite{wuwoodroofe2004}, while equation (1.10) in
\cite{kipnis1986central} shows that $H_{n}\left(  x_{0},x_{1}\right)
:=h_{n}\left(  x_{1}\right)  -\Pi h_{n}\left(  x_{0}\right)  $ converges in
$L^{2}\left(  \mathsf{X\times X,}\text{ }\mu \otimes \Pi \right)  $. Therefore,
the conditions of Corollary 2 in \cite{wuwoodroofe2004} are satisfied so the
statement of the lemma follows from their equation (10); see their comments
after this equation.
\end{proof}

\begin{proof}
[Proof of Proposition 2]Let $\left(  X_{i}\right)  _{i\geq1}$ be a Markov chain
evolving according to $P$ and $(\widetilde{X}_{i},\tau_{i})_{i\geq1}$ the
associated jump\ chain representation evolving according to $\overline{P}$,
as defined in Lemma 1. We denote by $\mathrm{P}_{\nu,\Pi}$ the law of a Markov
chain with initial distribution $\nu$ and transition kernel $\Pi$. By Theorem
1.3 in \cite{kipnis1986central}, we have under $\mathrm{P}_{\mu,P}$
\begin{equation}
S_{n}=\sum_{i=1}^{n}h(X_{i})=M_{n}+\xi_{n},
\label{eq:representationsumbymartingaleandremainder}%
\end{equation}
where $M_{n}$ is a square integrable martingale with respect to the natural
filtration of $\left(  X_{i}\right)  _{i\geq1}$, while we have the following
convergence in probability
\begin{equation}
n^{-1/2}\sup_{1\leq i\leq n}\text{ }\left \vert \xi_{i}\right \vert
\overset{\mathrm{P}_{\mu,P}}{\longrightarrow}0.
\label{eq:asymptoticnegligibility}%
\end{equation}
Define $T_{n}=\tau_{1}+\cdots+\tau_{n}$. The kernel $\overline{P}$ is ergodic because $\widetilde{P}$ is ergodic.
Hence, \ $\mathrm{P}_{\widetilde{\mu
},\widetilde{P}}-$almost surely,
\begin{equation}
\frac{T_{n}}{n}\rightarrow \widetilde{\mu}\left(  1/\varrho \right)  =\frac
{1}{\mu \left(  \varrho \right)  }. \label{eq:renewal}%
\end{equation}
The above limit also holds $\mathrm{P}_{\mu,\widetilde{P}}$-almost surely,
since $\mathrm{P}_{\mu,\widetilde{P}}$ is absolutely continuous with respect
to $\mathrm{P}_{\widetilde{\mu},\widetilde{P}}$. We first show that
\begin{equation}
\left \{  n/\mu \left(  \varrho \right)  \right \}  ^{-1/2}(M_{T_{n}%
}-M_{\left \lfloor n/\mu \left(  \varrho \right)  \right \rfloor })\overset
{\mathrm{P}_{\mu,P}}{\longrightarrow}0. \label{eq:convergenceinprobaforMTn}%
\end{equation}
Let $\epsilon>0$ be arbitrary and define the event%
\[
A_{n}=\big \{(1-\epsilon)\frac{n}{\mu \left(  \varrho \right)  }\leq
T_{n}<(1+\epsilon)\frac{n}{\mu \left(  \varrho \right)  }\big \}.
\]
By (\ref{eq:renewal}), we have  $\mathrm{P}_{\mu,P}(A_{n})\rightarrow1$. The
following inequality holds on the event $A_{n}$,
\begin{align*}
|M_{T_{n}}-M_{\left \lfloor n/\mu \left(  \varrho \right)  \right \rfloor }|  &
\leq|M_{T_{n}}-M_{\left \lfloor (1-\epsilon)n/\mu \left(  \varrho \right)
\right \rfloor }|+|M_{\left \lfloor n/\mu \left(  \varrho \right)  \right \rfloor
}-M_{\left \lfloor (1-\epsilon)n/\mu \left(  \varrho \right)  \right \rfloor }|\\
&  \leq2\sup_{1\leq i\leq2\left \lfloor \epsilon n/\mu \left(  \varrho \right)
\right \rfloor +1}|\tilde{M}_{i}|,
\end{align*}
where $\tilde{M}_{i}:=M_{\left \lfloor (1-\epsilon)n/\mu \left(  \varrho \right)
\right \rfloor +i}-M_{\left \lfloor (1-\epsilon)n/\mu \left(  \varrho \right)
\right \rfloor }$ is a square integrable martingale with stationary increments.
Thus, for any $\delta>0$,
\begin{align*}
\mathrm{P}_{\mu,P}\Bigg(\Big|M_{T_{n}}-M_{\left \lfloor n/\mu \left(
\varrho \right)  \right \rfloor }\Big|  &  >\delta \left \{  n/\mu \left(
\varrho \right)  \right \}  ^{1/2}\Bigg)\\
&  \leq \mathrm{P}_{\mu,P}\Bigg(\bigg \{ \big|M_{T_{n}}-M_{\left \lfloor
n/\mu \left(  \varrho \right)  \right \rfloor }|>\delta \left \{  n/\mu \left(
\varrho \right)  \right \}  ^{1/2}\bigg \} \cap A_{n}\Bigg)+\mathrm{P}_{\mu
,P}(A_{n}^{c})\\
&  \leq \mathrm{P}_{\mu,P}\Bigg(2\sup_{1\leq i\leq2\left \lfloor \epsilon
n/\mu \left(  \varrho \right)  \right \rfloor +1}|\tilde{M}_{i}|>\delta \left \{
n/\mu \left(  \varrho \right)  \right \}  ^{1/2}\Bigg)+o(1)\\
&  \leq C_{1}\frac{\mathrm{E}_{\mu,P}\left(  \tilde{M}_{2\left \lfloor \epsilon
n/\mu \left(  \varrho \right)  \right \rfloor +1}^{2}\right)  }{\delta^{2}%
n/\mu \left(  \varrho \right)  }+o(1)\\
&  \leq C_{2}\frac{\epsilon n/\mu \left(  \varrho \right)  }{\delta^{2}%
n/\mu \left(  \varrho \right)  }+o(1)\leq C_{3}\epsilon/\delta^{2}+o(1),
\end{align*}
where $C_{1},C_{2},C_{3}<\infty$. \ The third inequality follows from Doob's
maximal inequality. The last inequality follows because, for any square
integrable martingale $\left(  N_{i}\right)  _{i\geq1}$ with stationary
increments, $\mathrm{E}_{\mu,P}\left(  N{_{n}^{2}}\right)  =\mathrm{E}_{\mu
,P}\left(  N{_{1}^{2}}\right)  $ $n$ holds. This bound is uniform in $n$, and
therefore
\[
\underset{n\rightarrow \infty}{\limsup \text{ }}\mathrm{P}_{\mu,P}%
\Bigg(\Big|M_{T_{n}}-M_{\left \lfloor n/\mu \left(  \varrho \right)
\right \rfloor }\Big|>\delta \left \{  n/\mu \left(  \varrho \right)  \right \}
^{1/2}\Bigg)\leq \epsilon/\delta^{2}.
\]
As $\epsilon>0$ is arbitrary,
\[
\lim_{n\rightarrow \infty}\mathrm{P}_{\mu,P}\Bigg(\Big|M_{T_{n}}%
-M_{\left \lfloor n/\mu \left(  \varrho \right)  \right \rfloor }\Big|>\delta
\left \{  n/\mu \left(  \varrho \right)  \right \}  ^{1/2}\Bigg)=0,
\]
for any $\delta>0$, and therefore (\ref{eq:convergenceinprobaforMTn}) holds.
Now, by Proposition 1, $n^{-1/2}S_{n}\longrightarrow \mathcal{N}\left \{
0;\mu \left(  h^{2}\right)  \textsc{\protect\small IF}({h},{P})\right \}  $. By the asymptotic
negligibility (\ref{eq:asymptoticnegligibility}) of $\xi_{n}$ and
(\ref{eq:convergenceinprobaforMTn}), we have by Slutsky's theorem that
$$\left \{  n/\mu \left(  \varrho \right)  \right \}  ^{-1/2}M_{T_{n}%
}\longrightarrow \mathcal{N}\left \{  0;\mu \left(  h^{2}\right)  \textsc{\protect\small IF}
({h},{P})\right \},$$ 
equivalently $n^{-1/2}M_{T_{n}}\longrightarrow
\mathcal{N}\left \{  0;\mu \left(  h^{2}\right)  \textsc{\protect\small IF}({h},{P})/\mu \left(
\varrho \right)  \right \}  $. Finally, note that for any $\delta>0$,
\begin{align*}
\mathrm{P}_{\mu,P}(|\xi_{T_{n}}|>\delta n^{1/2})  &  \leq \mathrm{P}_{\mu
,P}(\{|\xi_{T_{n}}|>\delta n^{1/2}\} \cap A_{n})+\mathrm{P}_{\mu,P}(A_{n}%
^{c})\\
&  \leq \mathrm{P}_{\mu,P}(\sup_{1\leq i\leq \left \lfloor (1+\epsilon
)n/\mu \left(  \varrho \right)  \right \rfloor }|\xi_{i}|>\delta n^{1/2}%
)+o(1)\rightarrow0\text{ by (\ref{eq:asymptoticnegligibility}).}%
\end{align*}
Therefore, using (\ref{eq:representationsumbymartingaleandremainder}) and
Slutsky's theorem, $n^{-1/2}S_{T_{n}}\rightarrow \mathcal{N}\left \{
0;\mu \left(  h^{2}\right)  \textsc{\protect\small IF}({h},{P})/\mu \left(  \varrho \right)
\right \}  $ when $\widetilde{X}_{1}=X_{1}\sim \mu$. However, this result also
holds when $\widetilde{X}_{1}\sim \widetilde{\mu}$, as established in Lemma
\ref{Lemma:absolutecontinuityCLT}. In particular, the asymptotic variance is
the same. Moreover, $(\tilde{X}_{i},\tau_{i})_{i\geq1}$ is reversible and
ergodic, while Lemma~\ref{lem:finiteIF} guarantees that $h/\varrho \in
L_{0}^{2}\left(  \mathsf{X},\widetilde{\mu}\right)  $ and $\textsc{\protect\small IF}({h/\varrho},{\widetilde{P}})<\infty$. Hence, Proposition 1 applied to $(\tilde{X}_{i},\tau_{i})_{i\geq1}$ ensures that
 the asymptotic variance is also given by the integrated autocovariance time. Equating the
two expressions, we obtain
\begin{align*}
\mu \left(  h^{2}\right)  \textsc{\protect\small IF}({h},{P})/\mu \left(  \varrho \right)   &
=\overline{\mu}\left(  \tau^{2}h^{2}\right)  +2\sum_{n\geq1}\left \langle \tau
h,\overline{P}^{n}\tau h\right \rangle _{\overline{\mu}}\\
&  =\widetilde{\mu}\left(  \frac{2-\varrho}{\varrho^{2}}h^{2}\right)
+2\sum_{n\geq1}\left \langle \frac{h}{\varrho},\widetilde{P}^{n}\frac
{h}{\varrho}\right \rangle _{\widetilde{\mu}}\\
&  =\widetilde{\mu}\left(  h^{2}/\varrho^{2}\right)  +\widetilde{\mu}\left(
h^{2}/\varrho^{2}\right)  \textsc{\protect\small IF}({h/\varrho},{\widetilde{P}})-\mu \left(
h^{2}\right)  /\mu \left(  \varrho \right)  ,
\end{align*}
where the equality in the second line follows from the expression of
$\overline{\mu}$ and $\overline{P}$, given in Lemma 1, and the properties of
the geometric distribution. This yields the equality of Proposition 2, which
can also be written as%
\[
\frac{1+\textsc{\protect\small IF}({h/\varrho},{\widetilde{P}})}{1+\textsc{\protect\small IF}({h},{P}%
)}=\frac{\mu \left(  h^{2}\right)  }{\mu \left(  \varrho \right)  \widetilde{\mu
}\left(  h^{2}/\varrho^{2}\right)  }=\frac{\mu \left(  h^{2}\right)  }%
{\mu \left(  h^{2}/\varrho \right)  }\leq1;
\]
as $0<\varrho \leq1$, implying that $\textsc{\protect\small IF}({h/\varrho},{\widetilde{P}})\leq
\textsc{\protect\small IF}({h},{P})$.
\end{proof}

\section{Proofs of other technical results in the main
paper\label{sec:technicalresultsmainpaper}}

\begin{proof}
[Proof of Lemma~1]As $P$ is $\psi$-irreducible, it is also $\mu$-irreducible as it
is $\mu$-invariant; see, for example, \cite{Tierney94}, p. 1759. Hence, for any
$x\in \mathsf{X}$ and $A\in \mathcal{X}$ with $\mu \left(  A\right)  >0$, there
exists an $n\geq1$ such that $P^{n}\left(  x,A\right)  >0$. As $\mu$ is not
concentrated on a single point by assumption, this implies that $\varrho
\left(  x\right)  >0$ for any $x\in \mathsf{X}$. The rest of the proposition
follows directly from Lemma 1 in \cite{douc2011vanilla}.
\end{proof}

\begin{proof}
[Proof of Lemma~3]Equations
(17) and
(18) and the expressions of their
associated invariant distributions follow from a direct application of Lemma~1.
The positivity of $\widetilde{Q}_{\textsc{z}}$ follows
directly from Proposition~3, see Remark 1. We write
$\widetilde{\pi}\otimes\widetilde{\pi}_{\textsc{z}}\left(  \mathrm{d}\theta
,\mathrm{d}z\right)  =\widetilde{\pi}\left(  \mathrm{d}\theta\right)
\widetilde{\pi}_{\textsc{z}}\left(  \mathrm{d}z\right)  $. By applying
Proposition~2 to $Q^{\ast}$, we obtain for any $h\in
L_{0}^{2}\left(  \Theta,\pi\right)  $ that $h/(\varrho_{\textsc{ex}}%
\varrho_{\textsc{z}})\in L_{0}^{2}(\Theta\times\mathbb{R},\tilde{\pi}%
\otimes\tilde{\pi}_{\textsc{z}})$, $\textsc{\protect\small IF}\left\{  h/(\varrho
_{\textsc{ex}}\varrho_{\textsc{z}}),\widetilde{Q}^{\ast}\right\}  <\infty$ and%
\begin{align*}
\pi\left(  h^{2}\right)  \left\{  1+\textsc{\protect\small IF}\left(  h,Q^{\ast}\right)
\right\}   &  =\overline{\pi}\left(  \varrho_{\textsc{ex}}%
\varrho_{\textsc{z}}\right)  \text{ }\widetilde{\pi}\otimes
\widetilde{\pi}_{\textsc{z}}\left\{  h^{2}/\left(  \varrho_{\textsc{ex}}%
^{2}\varrho_{\textsc{z}}^{2}\right)  \right\}  \left[  1+\textsc{\protect\small IF}%
\left\{  h/(\varrho_{\textsc{ex}}\varrho_{\textsc{z}}),\widetilde{Q}^{\ast
}\right\}  \right]  \\
&  =\pi\left(  \varrho_{\textsc{ex}}\right)  \pi_{\textsc{z}}\left(
\varrho_{\textsc{z}}\right)  \widetilde{\pi}\left(  h^{2}/\varrho
_{\textsc{ex}}^{2}\right)  \widetilde{\pi}_{\textsc{z}}\left(  1/\varrho
_{\textsc{z}}^{2}\right)  \left[  1+\textsc{\protect\small IF}\left\{  h/(\varrho
_{\textsc{ex}}\varrho_{\textsc{z}}),\widetilde{Q}^{\ast}\right\}  \right]  .
\end{align*}
The identity follows easily as $\widetilde{\pi}_{\textsc{z}}\left(  1/\varrho
_{\textsc{z}}^{2}\right)  =\pi_{\textsc{z}}\left(  1/\varrho_{\textsc{z}%
}\right)  /\pi_{\textsc{z}}\left(  \varrho_{\textsc{z}}\right)$ and $\pi_{\text{z}}\left(  1/\varrho_{\text{z}}\right)  <\infty$ .

To prove the geometric ergodicity of $\widetilde{Q}_{\textsc{z}}$, we follow \citet[Chapter~15]{MT09}.
Notice first that
\begin{align*}
\widetilde{Q}_{\textsc{z}}\left(  z,\rd w\right)   &  =\frac{g\left(  \rd w\right)
\alpha\left(  z,w\right)  }{\varrho_{\textsc{z}}\left(  z\right)  }
  \geq g\left(  \rd w\right)  \left\{  e^{w-z}\mathbb{I}\left(  w<z\right)
+\mathbb{I}\left(  w\geq z\right)  \right\} ,
\end{align*}
and consider the set $C=\left(  -\infty,z_{0}\right]  $, where $z_{0}>0$ and $\int_0^{z_0}g(w) \rd w>0$. For any $z\in C$ and $w\geq0$,%
\begin{align*}
\widetilde{Q}_{\textsc{z}}\left(  z, \rd w\right)   &  \geq g\left(  \rd w\right)
\left\{  e^{w-z_{0}}\mathbb{I}\left(  w<z\right)  +\mathbb{I}\left(  w\geq
z\right)  \right\}  \\
&  \geq e^{-z_{0}}g\left(  \rd w\right)  =\varepsilon\text{ }\nu\left(
\rd w\right),
\end{align*}
where
\[
\varepsilon=e^{-z_{0}}\int_{0}^{z_{0}}g\left(  \rd w\right)  \leq1,
\]
and $\nu$ is the probability measure concentrated on $\left[  0,z_{0}\right]
\subset C$, given by
\[
\nu\left(  \rd w\right)  =\frac{g\left(  \rd w\right)  \mathbb{I}\left(  0\leq w\leq
z_{0}\right)  }{\int_{0}^{z_{0}}g\left(  \rd w\right)  }.
\]
Hence, $C$ is a small set.

To complete the proof of geometric ergodicity of $\widetilde{Q}_{\textsc{z}}$, we check
 that $V\left(  z\right)  =1/\varrho_{\text{z}}\left(  z\right)  $
satisfies a geometric drift condition. Note that $V\left(  z\right)  \geq1$
for any $z$, and
\begin{align}
\frac{\int\widetilde{Q}_{\textsc{z}}\left(  z, \rd w\right)  V\left(  w\right)
}{V\left(  z\right)  }  & =e^{-z}\int_{-\infty}^{z}\frac{e^{w}g\left(
\rd w\right)  }{\varrho_{\text{z}}\left(  w\right)  }+\int_{z}^{\infty}%
\frac{g\left( \rd w\right)  }{\varrho_{\text{z}}\left(  w\right)  }\nonumber\\
& =e^{-z}\int_{-\infty}^{z}\frac{\pi_{\text{z}}\left(  \rd w\right)  }%
{\varrho_{\text{z}}\left(  w\right)  }+\int_{z}^{\infty}%
\frac{g\left( \rd w\right)  }{\varrho_{\text{z}}\left(  w\right)  }.\label{eq:Lyapunovfunction}%
\end{align}
We have $\pi_{\text{z}}\left(  1/\varrho_{\text{z}}\right)  <\infty$, as established earlier, because $\textsc{\protect \small IF}\left(  h,Q^{\ast}\right)  <\infty$ by assumption. It follows that the first integral on the right hand side of
(\ref{eq:Lyapunovfunction}) is bounded. To prove that the second integral is bounded, we use the fact that $\varrho_{\text{z}}\left(
z\right)  $ is a non-increasing function. We have
\[
\varrho_{\text{z}}\left(  z\right)  = 1-G\left(  z\right)  +e^{-z}\Pi\left(
z\right),
\]
where $G$, respectively $\Pi$, is the cumulative distribution function of $g$,
respectively $\pi_{\text{z}}$, so its derivative with respect to $z$ is equal
to
\[
\varrho_{\text{z}}^{\prime}\left(  z\right)=-g(z)+ e^{-z}\pi(z)-e^{-z}\Pi\left(  z\right) =-e^{-z}\Pi\left(  z\right)\leq 0  .
\]
It follows that the second term on the right hand side of (\ref{eq:Lyapunovfunction}) is
bounded by
\begin{align*}
\int_{z}^{\infty}\frac{g\left( \rd w\right)  }{\varrho_{\text{z}}\left(  w\right) }  & \leq\int_{-\infty}^{\infty}\frac{g\left(
\rd w\right)  }{\varrho_{\text{z}}\left(  w\right)  }=\int_{-\infty}^{0}%
\frac{g\left(  \rd w\right)  }{\varrho_{\text{z}}\left(  w\right)  }+\int%
_{0}^{\infty}\frac{g\left(  \rd w\right)  }{\varrho_{\text{z}}\left(  w\right)
}\\
& \leq\frac{1}{\varrho_{\text{z}}\left(  0\right)  }\int_{-\infty}^{0}g\left(
\rd w\right)  +\int_{0}^{\infty}\frac{e^{-w}\pi_{\text{z}}\left(  \rd w\right)
}{\varrho_{\text{z}}\left(  w\right)  }<\infty.
\end{align*}
Therefore, for any $0<\lambda<1$, there exists
$z_{0}^{\prime}>0$ such that
\[
\frac{\int\widetilde{Q}_{\textsc{z}}\left(  z, \rd w\right)  V\left(  w\right)
}{V\left(  z\right)  }\leq\lambda,
\]
for all $z\geq$ $z_{0}^{\prime}$.
We now establish that
\[
\underset{z\leq z_{0}^{\prime}}{\sup}\int\widetilde{Q}_{\textsc{z}}\left(
z, \rd w\right)  V\left(  w\right)  <\infty.
\]
As $\varrho_{\text{z}}\left(z\right)  $ is a non-increasing function, it follows that for $z\leq z_{0}^{\prime}$
\begin{align*}
\int\widetilde{Q}_{\textsc{z}}\left(  z, \rd w\right)  V\left(  w\right)    &
=\int\frac{g\left(  \rd w\right)  \alpha\left(  z,w\right)  }{\varrho_{\text{z}%
}\left(  z\right)  \varrho_{\text{z}}\left(  w\right)  }\\
& \leq\frac{1}{\varrho_{\text{z}}\left(  z_{0}^{\prime}\right)  }\int%
\frac{g\left(  \rd w\right)  \alpha\left(  z,w\right)  }{\varrho_{\text{z}%
}\left(  w\right)  }\\
& =\frac{1}{\varrho_{\text{z}}\left(  z_{0}^{\prime}\right)  }\frac
{\int\widetilde{Q}_{\textsc{z}}\left(  z,\rd w\right)  V\left(  w\right)
}{V\left(  z\right)  }.
\end{align*}
We now show that
\[
\sup_z\frac{\int\widetilde{Q}_{\textsc{z}}\left(  z,\rd w\right)
V\left(  w\right)  }{V\left(  z\right)  }<\infty.
\]
The first term on the right hand side of
(\ref{eq:Lyapunovfunction}) is bounded by
\[
e^{-z}\int_{-\infty}^{z}\frac{\pi_{\text{z}}\left(  \rd w\right)  }%
{\varrho_{\text{z}}\left(  w\right)  }\leq\frac{e^{-z}}{\varrho_{\text{z}%
}\left(  z\right)  }\int_{-\infty}^{z}\pi_{\text{z}}\left(  \rd w\right)
\leq\frac{e^{-z}}{e^{-z}\Pi\left(  z\right)  }\Pi\left(  z\right)  =1,
\]
while we have already shown that the second term on  right hand side of
(\ref{eq:Lyapunovfunction}) is bounded.

Hence, we can conclude that, for any $0<\lambda<1$, there exists $z_{0}>0$ and
$b<\infty$ such that
\[
\int\widetilde{Q}_{\textsc{z}}\left(  z,\rd w\right)  V\left(  w\right)
\leq\lambda V\left(  z\right)  +b\mathbb{I}_{C}\left(  z\right),
\]
where $C=\left(  -\infty,z_{0}\right]$.

The inequality $\textsc{\protect \small IF}(1/\varrho_\textsc{z}, \widetilde{Q}_{\textsc{z}})<\infty$ now follows because $\widetilde{Q}_{\textsc{z}}$ is geometrically ergodic with drift function $V\left(  z\right)  =1/\varrho_{\text{z}}\left(  z\right)  $ and $\widetilde{\pi}_\textsc{z}(1/\varrho_\textsc{z}^2)<\infty$.
\end{proof}
\begin{proof}
[Proof of Proposition 3]If $\langle f,\widetilde{P}f\rangle_{\widetilde{\mu}}\geq0$
for any $f\in L^{2}\left(  \mathsf{X},\widetilde{\mu}\right)  $, then
$\widetilde{P}$ is positive by definition, implying the positivity of $P$ as
$L^{2}\left(  \mathsf{X},\mu \right)  \subseteq L^{2}\left(  \mathsf{X}%
,\widetilde{\mu}\right)  $ and
\[
\langle f,Pf\rangle_{\mu}=\mu \left(  \varrho \right)  \langle f,\widetilde
{P}f\rangle_{\widetilde{\mu}}+\mu \left \{  \left(  1-\varrho \right)
f^{2}\right \}  .
\]
For a proposal of the form $q\left(  x,y\right)  =\int s\left(  x,z\right)
s\left(  y,z\right)  \chi \left(  \mathrm{d}z\right)  $, Lemma 3.1 in
\cite{baxendale2005} establishes that $\langle f,\widetilde{P}f\rangle
_{\widetilde{\mu}}\geq0$ for any $f\in L^{2}\left(  \mathsf{X},\widetilde{\mu
}\right)  $. For a $\nu$-reversible proposal such that $\nu \left(  x\right)
q\left(  x,y\right)  =\int r\left(  x,z\right)  r\left(  y,z\right)
\chi \left(  \mathrm{d}z\right)  $, we have for any $f\in L^{2}\left(
\mathsf{X},\widetilde{\mu}\right) $,
\begin{align*}
&  \mu \left(  \varrho \right)  \langle f,\widetilde{P}f\rangle_{\widetilde{\mu
}}\\
&  =\iint f\left(  x\right)  f\left(  y\right)  \nu \left(  x\right)  q\left(
x,y\right)  \min \left \{  \frac{\mu \left(  x\right)  }{\nu \left(  x\right)
},\frac{\mu \left(  y\right)  }{\nu \left(  y\right)  }\right \}  \mathrm{d}%
x\mathrm{d}y\\
&  =\iint f\left(  x\right)  f\left(  y\right)  \nu \left(  x\right)  q\left(
x,y\right)  \left[  \int_{0}^{\infty}\mathbb{I}_{\left \{  0,\mu \left(
x\right)  /\nu \left(  x\right)  \right \}  }\left(  t\right)  \mathbb{I}%
_{\left \{  0,\mu \left(  y\right)  /\nu \left(  y\right)  \right \}  }\left(
t\right)  \mathrm{d}t\right]  \mathrm{d}x\mathrm{d}y\\
&  =\int_{0}^{\infty}\left[  \iiint f\left(  x\right)  r\left(  x,z\right)
f\left(  y\right)  r\left(  y,z\right)  \mathbb{I}_{\left \{  0,\mu \left(
x\right)  /\nu \left(  x\right)  \right \}  }\left(  t\right)  \mathbb{I}%
_{\left \{  0,\mu \left(  y\right)  /\nu \left(  y\right)  \right \}  }\left(
t\right)  \mathrm{d}x\mathrm{d}y\text{ }\chi \left(  \mathrm{d}z\right)
\right]  \mathrm{d}t\\
&  =\int_{0}^{\infty}\left(  \int \left[  \int f\left(  x\right)  r\left(
x,z\right)  \mathbb{I}_{\left \{  0,\mu \left(  x\right)  /\nu \left(  x\right)
\right \}  }\mathrm{d}x\right]  ^{2}\text{ }\chi \left(  \mathrm{d}z\right)
\right)  \mathrm{d}t\geq0,
\end{align*}
by a repeated application of Fubini's theorem.
\end{proof}

\begin{proof}
[Proof of Proposition~4] Theorem 2.2 in \cite{robertstweedie1996} establishes the ergodicity
of $Q_{\textsc{ex}}$. We extend their argument to prove the ergodicity
of $Q^{\ast}$. For the ball $B\left(  \theta,L\right)  $ centred at $\theta$ of radius $L$,
we define
\[
\eta \left(  \theta,L\right)  =\left \{  \underset{\vartheta \in B\left(
\theta,L\right)  }{\sup}\pi \left(  \vartheta \right)  \right \}  ^{-1}%
\underset{\vartheta \in B\left(  \theta,L\right)  }{\inf}\pi \left(
\vartheta \right),
\]
which, by assumption, is such that $0<\eta \left(  \theta,L\right)  <\infty$.
Then, we have for any $\left(  \theta,z\right)  \in \Theta \times \mathbb{R},$
$\vartheta \in B\left(  \theta,\delta \right)  $ and $w\in \mathbb{R}$,
\begin{align}
Q^{\ast}\left \{  \left(  \theta,z\right)  ,\left(  d\vartheta,\rd w\right)
\right \}   &  \geq q\left(  \theta,\vartheta \right)  \alpha_{\text{{\tiny EX}%
}}(\theta,\vartheta)g(w)\alpha_{\textsc{z}}(z,w)\mathrm{d}%
\vartheta \mathrm{d}w\nonumber \\
&  \geq \varepsilon \eta \left(  \theta,\delta \right)  \min \{g(w),e^{-z}\pi
_{\textsc{z}}(w)\} \mathrm{d}\vartheta \mathrm{d}w \label{eq:minorQ*1},%
\end{align}
which is strictly positive on $S:=\left \{  w:g(w)>0\right \}  =\left \{
w:\pi_{\textsc{z}}(w)>0\right \}  $. Hence, the $n$-step density part of $\left(
Q^{\ast}\right)  ^{n}$ is strictly positive for all $\left(  \vartheta
,z\right)  \in B\left(  \theta,n\delta \right)  \times S$. This establishes the
$\mathrm{d}\vartheta \times \pi_{\textsc{z}}\left(  \mathrm{d}z\right)  $ irreducibility of  $Q^{\ast}$, and
hence its ergodicity as it is $\overline{\pi}$-invariant.
For
$\widetilde{Q}^{\ast}$, we have for any $\left(  \theta,z\right)  \in
\Theta \times \mathbb{R},$ $\vartheta \in B\left(  \theta,\delta \right)  $ and
$w\in \mathbb{R}$,
\begin{align*}
\widetilde{Q}^{\ast}\left \{  \left(  \theta,z\right)  ,\left(  d\vartheta
,\rd w\right)  \right \}   &  =\frac{q\left(  \theta,\vartheta \right)
\alpha_{\text{{\tiny EX}}}(\theta,\vartheta)g(w)\alpha_{\textsc{z}}(z,w)}{\varrho_{\text{{\tiny EX}}}\left(  \theta \right)  \varrho
_{\textsc{z}}\left(  z\right)  }\mathrm{d}\vartheta \mathrm{d}w\\
&  \geq \varepsilon \eta \left(  \theta,\delta \right)  \min \{g(w),e^{-z}\pi
_{\textsc{z}}(w)\} \mathrm{d}\vartheta \mathrm{d}w,
\end{align*}
using calculations as in (\ref{eq:minorQ*1}) and the fact that $0<\varrho
_{\text{{\tiny EX}}}\left(  \theta \right)  \varrho_{\textsc{z}}\left(
z\right)  \leq1$ for any $\left(  \theta,z\right)  \in \Theta \times \mathbb{R}$,
as $Q^{\ast}$ is irreducible.
Finally, the
ergodicity of $\widetilde{Q}_{\textsc{ex}}$ follows, using similar
arguments, from the ergodicity of $Q_{\textsc{ex}}$.
\end{proof}

\section{Statements and proofs of auxiliary technical
results\label{sec:technicalresultssupplementary}}

\begin{proposition}
\label{Cor:CT} Define the relative computing time ${\textsc{\small uRCT}}_2(h;\sigma)$
\[
{\textsc{\small uRCT}}_{2}(h;\sigma)
=\frac{{\textsc{\small uRIF}}_{2}(h;\sigma)}%
{\sigma^{2}},
\] where ${\textsc{\small uRIF}}_{2}(h;\sigma)$ is the relative inefficiency. Using  the same assumptions as in Theorem 1,
\begin{enumerate}
\item[(i)]
If $\textsc{\protect\small IF}({h}, Q_{\textsc{ex}})=1$, then ${\textsc{\small uRCT}}_2(h; \sigma)$ is minimized at $\sigma_{\text{opt}}=0.92$ and
$\textsc{\small RIF}(h, Q;\sigma_{opt})=
\textsc{\protect\small IF}(h, Q_\pi; \sigma_{opt})
=4.54$,
$\pi_{\textsc{z}}^{\sigma_{opt}}\left(  \varrho_{\textsc{z}}^{\sigma_{opt}%
}\right)  =0.51$\textsl{.}
\item[(ii)] If $\textsc{\protect\small IF}(h/\varrho_{\textsc{ex}}, \tilde{Q}_{\textsc{ex}})\geq
1$, $\sigma_{opt}$ increases
to $\sigma_{opt}=1.02$ as $\textsc{\small IF}({h}, Q_{\textsc{ex}})\longrightarrow \infty.$
\item[(iii)]
${\textsc{\small uRIF}}_{2}(h;\sigma)$ and ${\textsc{\small uRCT}}_{2}(h;\sigma)$
are decreasing functions of $\textsc{\small IF}({h}, Q_{\textsc{ex}})$.
\end{enumerate}
\end{proposition}

\begin{proof}
[Proof of Proposition~\ref{Cor:CT}]We consider minimizing $\textsc{\small uRCT}%
_{2}(h;\sigma)$ with respect to $\sigma$. Then,
\[
\textsc{\small uRCT}_{2}(h;\sigma)=\frac{\left \{  1+\textsc{\small IF}\left(  h,Q_{\textsc{ex}}\right)  \right \}  }{2\textsc{\small IF}\left(  h,Q_{\textsc{ex}}\right)  }%
\frac{\textsc{\small IF}(h,Q_{\pi};\sigma)}{\sigma^{2}}+\frac{\textsc{\small IF}\left(
h,Q_{\textsc{ex}}\right)  -1}{2\sigma^{2}\textsc{\small IF}\left(  h,Q_{\textsc{ex}}\right)  }.
\]
To obtain Part (i), we note that $\textsc{\small uRCT}_{2}(h;\sigma)=\textsc{\small IF}(h,Q_{\pi};\sigma)/\sigma^{2}$ when $\textsc{\small IF}\left(  h,Q_{\textsc{ex}}\right)  =1$. We define $H(\sigma)=\textsc{\small IF}(h,Q_{\pi}%
;\sigma)/\sigma^{2}$. Using Lemma 5 in \cite{PittSilvaGiordaniKohn(12)}, one can
verify that $H(\sigma)$ is minimized at $\sigma_{\text{opt}}=0.92$ and that
$\partial^{2}\left \{  H(\sigma)\  \right \}  /\left(  \partial \sigma \right)
^{2}>0$. The numerical values of Part (i) at $\sigma_{\text{opt}}=0.92$ can be
found in \cite{PittSilvaGiordaniKohn(12)}. To obtain Part~(ii), we note that
\begin{align}
\partial \textsc{\small uRCT}_{2}(h;\sigma)/\partial \sigma &  =\left \{  1+\textsc{\small IF}%
\left(  h,Q_{\textsc{ex}}\right)  \right \}  /\left \{  2\textsc{\small IF}\left(
h,Q_{\textsc{ex}}\right)  \right \}  \partial H(\sigma)/\partial \sigma
\label{eq:RCT_1stderiv}\\
&  \qquad-\left \{  \textsc{\small IF}\left(  h,Q_{\textsc{ex}}\right)  -1\right \}
/\left \{  \sigma^{3}\textsc{\small IF}\left(  h,Q_{\textsc{ex}}\right)  \right \}
,\nonumber \\
\partial^{2}\textsc{\small uRCT}_{2}(h;\sigma)/\left(  \partial \sigma \right)  ^{2}
&  =(1+\textsc{\small IF}\left(  h,Q_{\textsc{ex}}\right)  )/(2\textsc{\small IF}\left(
h,Q^{\text{EX}}\right)  )\partial^{2}H(\sigma)/\left(  \partial \sigma \right)
^{2}\nonumber \\
&  \qquad+3\left \{  \textsc{\small IF}\left(  h,Q_{\textsc{ex}}\right)  -1\right \}
/\left \{  \sigma^{4}\textsc{\small IF}\left(  h,Q_{\textsc{ex}}\right)  \right \}
,\nonumber
\end{align}
so that $\partial^{2}\textsc{\small uRCT}_{2}(h;\sigma)/\left(  \partial
\sigma \right)  ^{2}>0$ if $\textsc{\small IF}\left(  h,Q_{\textsc{ex}}\right)  \geq1$.
For the limiting case of Part (ii),
\[
\lim_{\textsc{\small IF}_{\textsc{ex}}\uparrow \infty}\partial \textsc{\small uRCT}%
_{2}(h;\sigma)/\partial \sigma=\{ \partial H(\sigma)/\partial \sigma
\}/2-1/\sigma^{3},
\]
which we can verify numerically is equal to $0$ at $\sigma_{\text{opt}}=1.02$.
For general values of $\textsc{\small IF}_{\textsc{ex}}$,
\begin{multline*}
\partial \left \{  \left.  \partial \textsc{\small uRCT}_{2}(h;\sigma)/\partial
\sigma \right \vert _{\sigma=\sigma_{\text{opt}}}\right \}  /\partial
\textsc{\small IF}\left(  h,Q_{\textsc{ex}}\right) \\
=-1/\left \{  \textsc{\small IF}\left(  h,Q_{\textsc{ex}}\right)  \right \}  ^{2}\left(
\left.  \partial H(\sigma)/\partial \sigma \right \vert _{\sigma=\sigma
_{\text{opt}}}/2+1/\sigma_{opt}^{3}\right)  <0,
\end{multline*}
where $\partial H(\sigma)/\partial \sigma>0$ for $\sigma>0.92$. Hence,
$\sigma_{\text{opt}}$ increases with $\textsc{\small IF}\left(  h,Q_{\textsc{ex}}\right)  $, which verifies Part~(ii). Finally, to obtain Part (iii), it is
straightforward to see that
\[
\textsc{\small uRIF}_{2}(h;\sigma)=\frac{\textsc{\small IF}(h,Q_{\pi};\sigma)+1}{2}%
+\frac{\textsc{\small IF}(h,Q_{\pi};\sigma)-1}{2\textsc{\small IF}\left(  h,Q_{\textsc{ex}}\right)  },
\]
so that $\textsc{\small uRIF}_{2}(h;\sigma)$ and $\textsc{\small uRCT}_{2}(h;\sigma
)=\textsc{\small uRIF}_{2}(h;\sigma)/\sigma^{2}$ are decreasing functions of
$\textsc{\small IF}\left(  h,Q_{\textsc{ex}}\right) $, holding $\sigma$ constant.
\end{proof}

%
%
%

\section{Asymptotic upper bound\label{section:asymptoticupperbound}}

This section illustrates, in the Gaussian noise case, the lower bound on
the inefficiency $\mathrm{\textsc{\protect\small lRIF}_{2}}(\sigma)=1/\{2\Phi(-\sigma/\surd{2})\}$ and the
exact relative inefficiency $\mathrm{\textsc{\protect\small aRIF}}(\sigma,l)=\Phi(-l/2)/\Phi\left \{
-\left(  2\sigma^{2}+l^{2}\right)  ^{1/2}/2\right \}$ obtained in
\cite{Sherlock2013efficiency} and discussed in Section 3.6 of the main paper.
Recall that $\mathrm{\textsc{\protect\small aRIF}}(\sigma,l)\rightarrow \textsc{\small lRIF}_{2}(\sigma)$ as
$l\rightarrow0$ and note that $\mathrm{\textsc{\protect\small aRIF}}(\sigma,l)\rightarrow
\Psi \left(  \sigma \right)  =\exp \left(  \sigma^{2}/4\right)  /\sigma^{2}$ as
$l\rightarrow \infty$. Figure \ref{fig:roberts_ifct} displays the
corresponding relative computing times $\mathrm{\textsc{\protect\small lRCT}_{2}}(\sigma)=\textsc{\small lRIF}_{2}(\sigma)/\sigma^{2}$ and $\mathrm{\textsc{\protect\small aRCT}}(\sigma;l)=\mathrm{\textsc{\protect\small aRIF}}(\sigma,l)/\sigma^{2}$. They are very similar in shape as a function of $\sigma$,
regardless of $l$, and are also minimized at similar values: $\textsc{\small lRCT}_{2}(\sigma)$ is minimized at $\sigma_{1}=1.68$ and $\Psi \left(
\sigma \right)  $ is minimized at $\sigma_{2}=2.00$, and $\textsc{\small lRCT}_{2}(\sigma_{1})=1.51$, $\textsc{\small lRCT}_{2}(\sigma_{2})=1.59$, $\Psi \left(\sigma_{1}\right)  =0.72$, $\Psi \left(\sigma_{2}\right)  =0.68$.
\begin{figure}[bh]
\begin{center}
\includegraphics[height=2.2in, width=\textwidth]{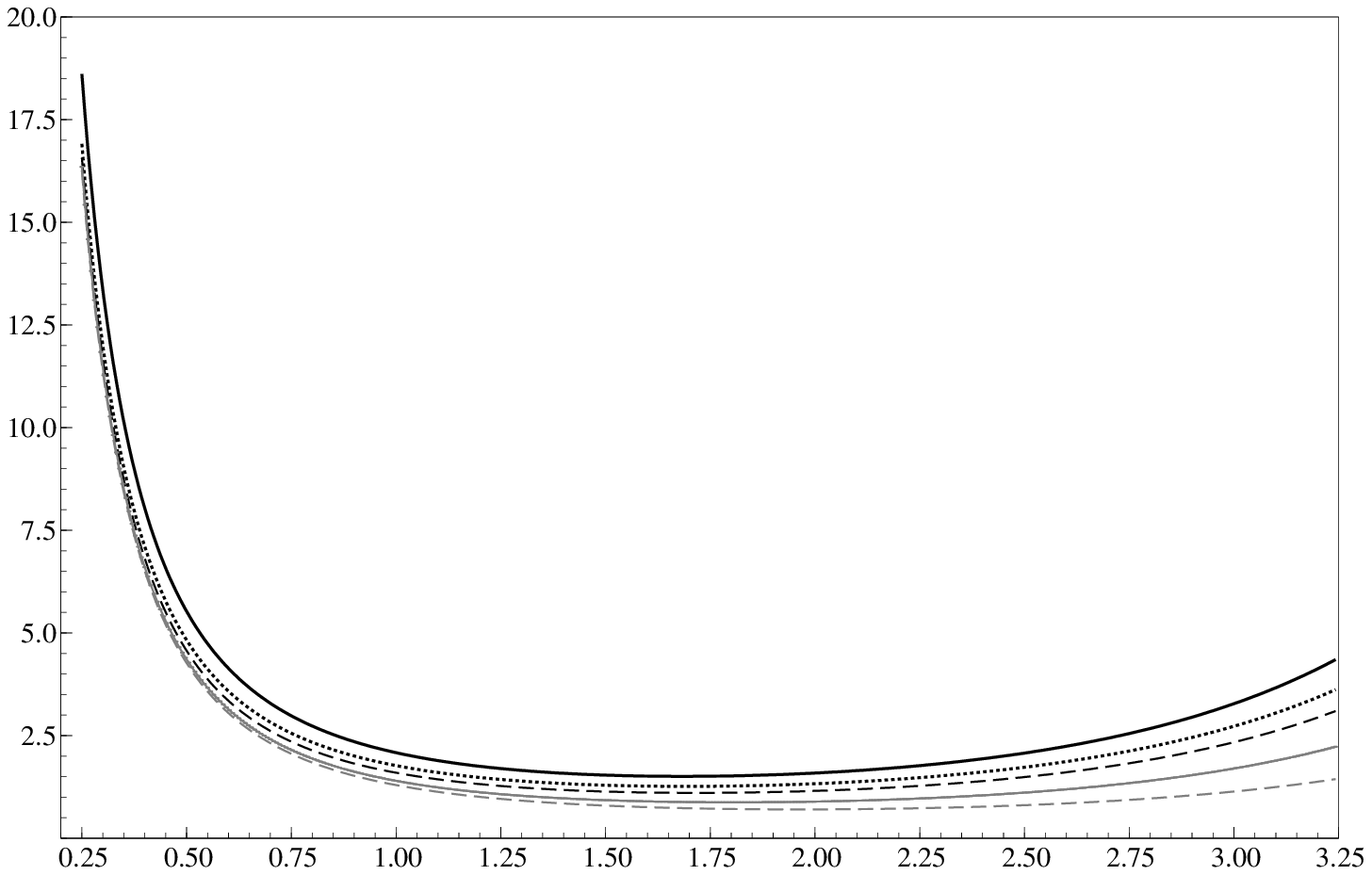}
\end{center}
\caption{Theoretical results for relative computing time.
$\textsc{\protect\small{lRCT}}_{2}(\sigma)$ (solid black) is displayed together with $\textsc{\protect\small aRCT}(\sigma,l)$  against $\sigma$.
$\textsc{\protect\small aRCT}(\sigma,l)$, the relative computing time for the limiting case of a random
walk proposal, is evaluated for $l=0.5$ (dotted black), $1$ (dashed black), $2.5$ (solid grey) and $10$ (dashed grey), where
$l$ is the scaling factor in the proposal.}
\label{fig:roberts_ifct}%
\end{figure}


\section{Simulation results\label{SS: AR plus noise}}

This section applies the pseudo-marginal algorithm to a linear Gaussian state-space model and presents additional simulation results for the stochastic volatility model discussed in the main paper. The linear Gaussian state-space model we consider is a first order autoregression AR$(1)$ observed with noise. In this case, $Y_{t}=X_{t}+\sigma _{\varepsilon }\varepsilon _{t}$, and
the state evolution is $X_{t+1}=\mu _{x}(1-\phi )+\phi X_{t}+\sigma _{\eta
}\eta _{t}$, where $\varepsilon _{t}$ and $\eta _{t}$ are standard normal
and independent. We take $\phi =0.8,$ $\mu _{x}=0.5,$ $\sigma _{\eta
}^{2}=1-\phi ^{2}$, so that the marginal variance $\sigma _{x}^{2}$ of the
state $X_{t}$ is $1$. We consider a series of length $T$, where $\sigma _{\varepsilon }^{2}=0.5$ is assumed known. The parameters of interest are therefore $\theta =(\phi ,\mu _{x},\sigma
_{x})$. The analysis is very similar to that of Section 4 of the  main
paper. However, for this state-space model, the likelihood can
be calculated by using the Kalman filter. This facilitates the
analysis of sections \ref{sec:emp_AR1} and \ref{sec:full_AR1} in two ways.
First, in the calculation of the log-likelihood error $Z=\log \widehat{p}_{N}(y\mid \theta )-\log
p(y\mid \theta )$, the true likelihood term is known rather than estimated.
Second, because the likelihood is known, we can directly examine the exact
chain $Q_{\textsc{ex}}$ and estimate the inefficiency $\textsc{\protect\small IF}(h,Q_{\textsc{ex}})$.

\subsection{Empirical results for the error of the log-likelihood estimator
\label{sec:emp_AR1}}

The analysis in this section mirrors that of Section 4.2 of the main paper.
We investigate empirically Assumptions 1 and 2 by examining the behaviour of
$Z=\log \widehat{p}_{N}(y\mid \theta )-\log p(y\mid \theta )$ for $T=40$, $%
300$ and $2700$. Corresponding values of $N$ are selected in each case to
ensure that the variance of $Z$ evaluated at the posterior mean $\overline{%
\theta }$ is approximately unity. The three plots on the left of Fig.~\ref%
{fig:AR1_hist} display the histograms corresponding to the density $%
g_{N}(z\mid \overline{\theta })$ of $Z$ for $\theta =\overline{\theta }$,
which is obtained by running $S=6000$ particle filters at this value. We
overlay on each histogram a kernel density estimate together with the
corresponding assumed density, $g_{\textsc{z}}^{\sigma }\left( z\right) $
of Assumption 2, where $\sigma ^{2}$ is the sample variance of $Z$ over the $%
S$ particle filters. For $T=40$, there is a slight discrepancy between the
assumed Gaussian densities and the true histograms representing $g_{N}(z\mid
\overline{\theta })$. In particular, whilst $g_{N}(z\mid \overline{\theta })$
is well approximated over most of its support, it is heavier tailed in the
left tail. For $T=300$ and $T=2700$, the assumed Gaussian densities are very
accurate.

\begin{figure}[h]
\captionsetup{font=small}
\begin{center}
\includegraphics[height=2.5in,width=5.5in]{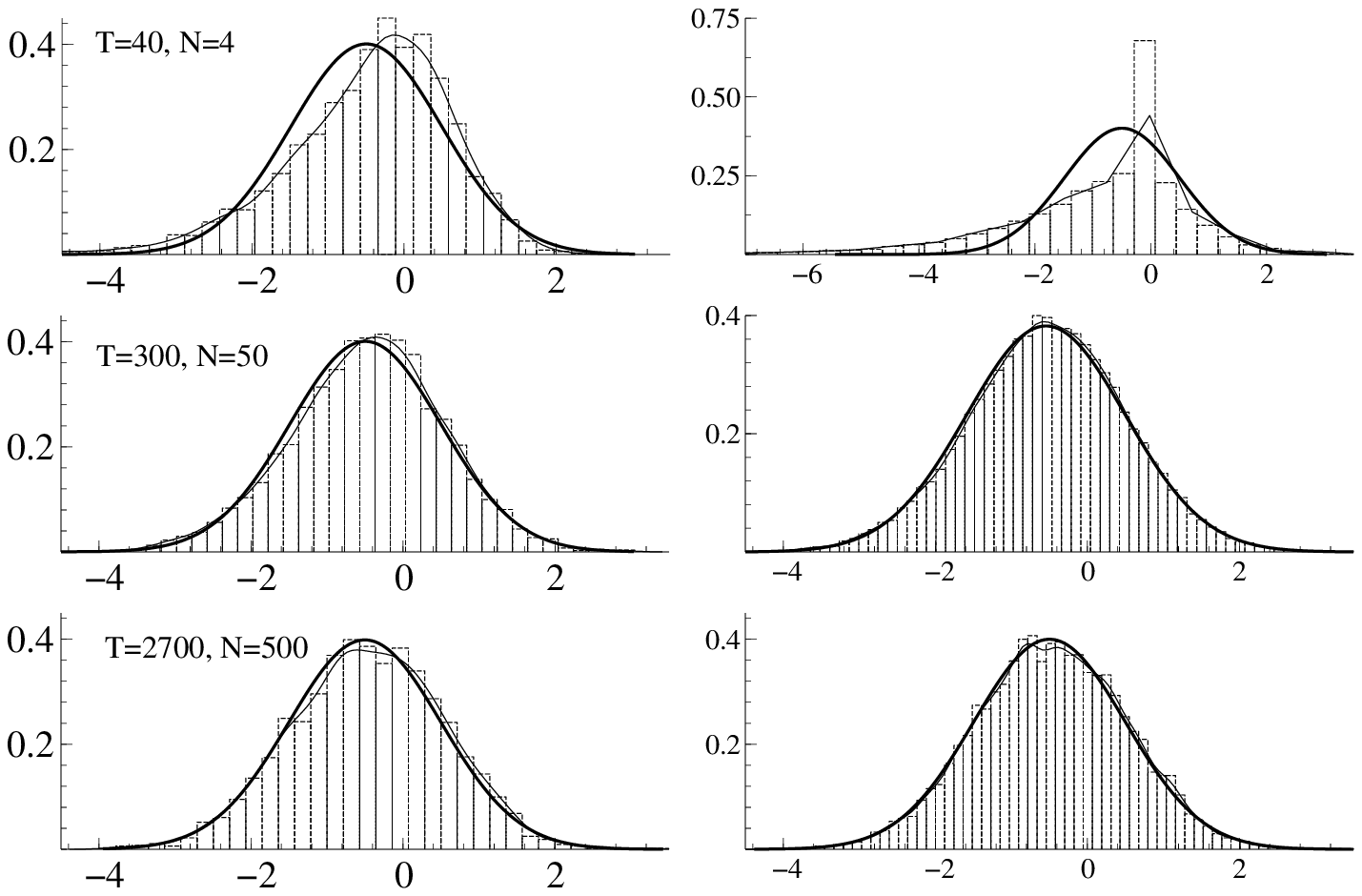}
\end{center}
\caption{AR(1) plus noise model experiment. Top to bottom: ${T=40,}$ $N=4$
(top), $T=300,$ $N=50$ (middle), $T=2700,$ $N=500$ (bottom). Left to right:
histograms and theoretical densities associated with $g_{N}(z\mid \theta)$ evaluated at the posterior mean $\overline{\theta}$
(left), over values from the posterior $\pi(\theta)$ (right). The densities $g_\textsc{z}^\sigma(z)$ are overlaid (solid lines).
}
\label{fig:AR1_hist}
\end{figure}

\begin{figure}[h]
\captionsetup{font=small}
\begin{center}
\includegraphics[height=2.5in, width=5.5in]{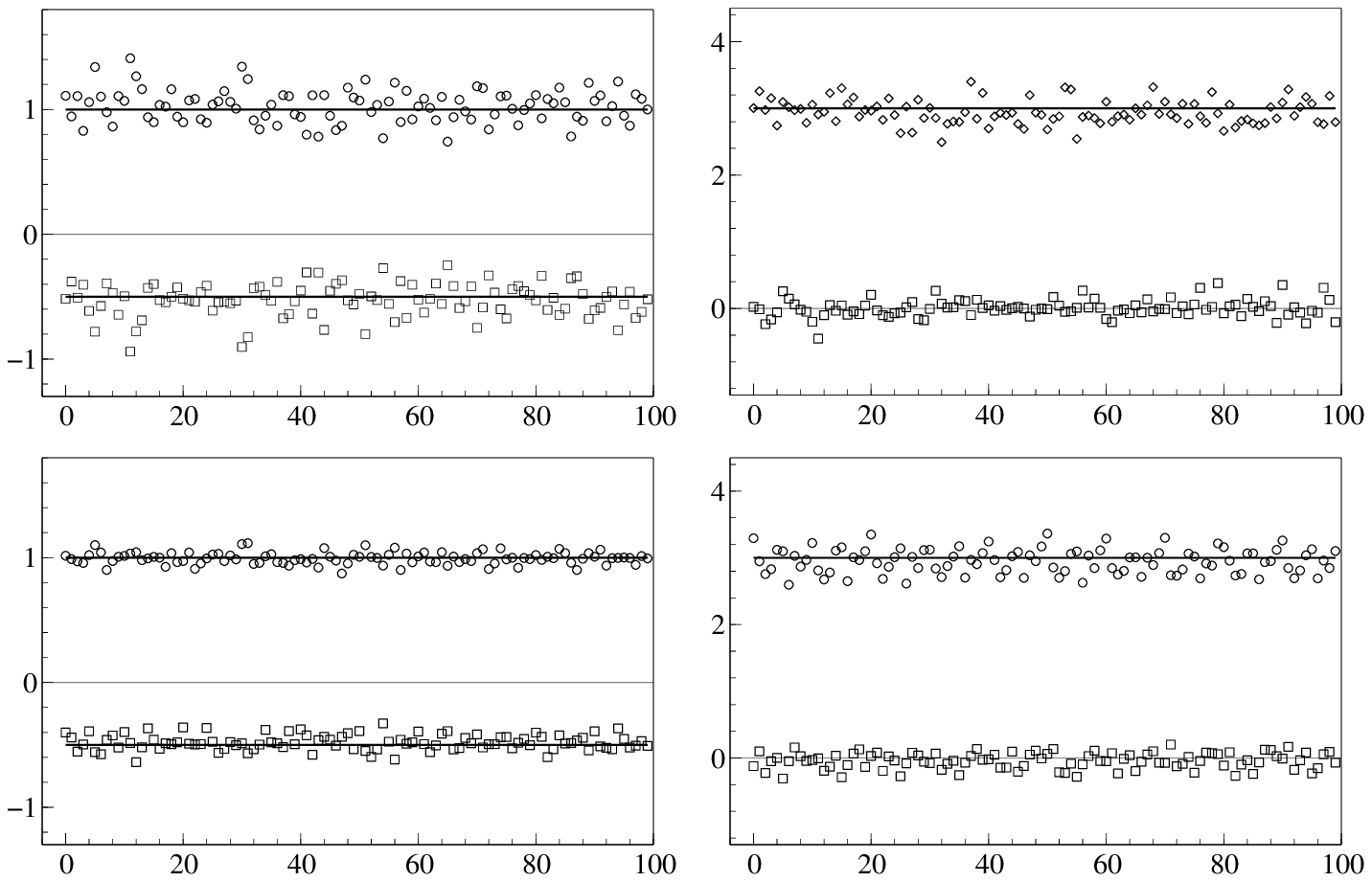}
\end{center}
\caption{AR(1) plus noise model experiment. Top to bottom: $T=300,$ $N=50$ (top), $T=2700,$ $N=500$ (bottom). Left to right:
mean (squares) and variance (circles) associated with $g_{N}(z\mid \theta)$ for $100$ different values of $\theta$ from $\pi(\theta)$ (left). The corresponding estimates of the third (squares) and fourth (circles) moments are displayed.
}
\label{fig:AR_cond_momentsZ}
\end{figure}

We also examine $Z$ when $\theta $ is distributed according to $\pi (\theta )
$. We record $100$ samples from $\pi (\theta )$, for $T=40$, $300$ and $2700$. For each of these samples, we run the particle filter $300$ times in order
to estimate the true likelihood at these values. The resulting histograms,
corresponding to the density ${\textstyle \int }\pi \left( \mathrm{d}\theta
\right) g_{N}(z\mid \theta )$ are displayed on the right panel of Fig.~\ref%
{fig:AR1_hist}. For $T=300$ and $T=2700$, the assumed densities $g_{\textsc{%
z}}^{\sigma }\left( z\right) $ are close to the corresponding histograms and
Assumptions 1 and 2 again appear to capture reasonably well the salient
features of the densities associated with $Z$.

\begin{figure}[h]
\captionsetup{font=small}
\begin{center}
\includegraphics[height=2.5in, width=5.5in]{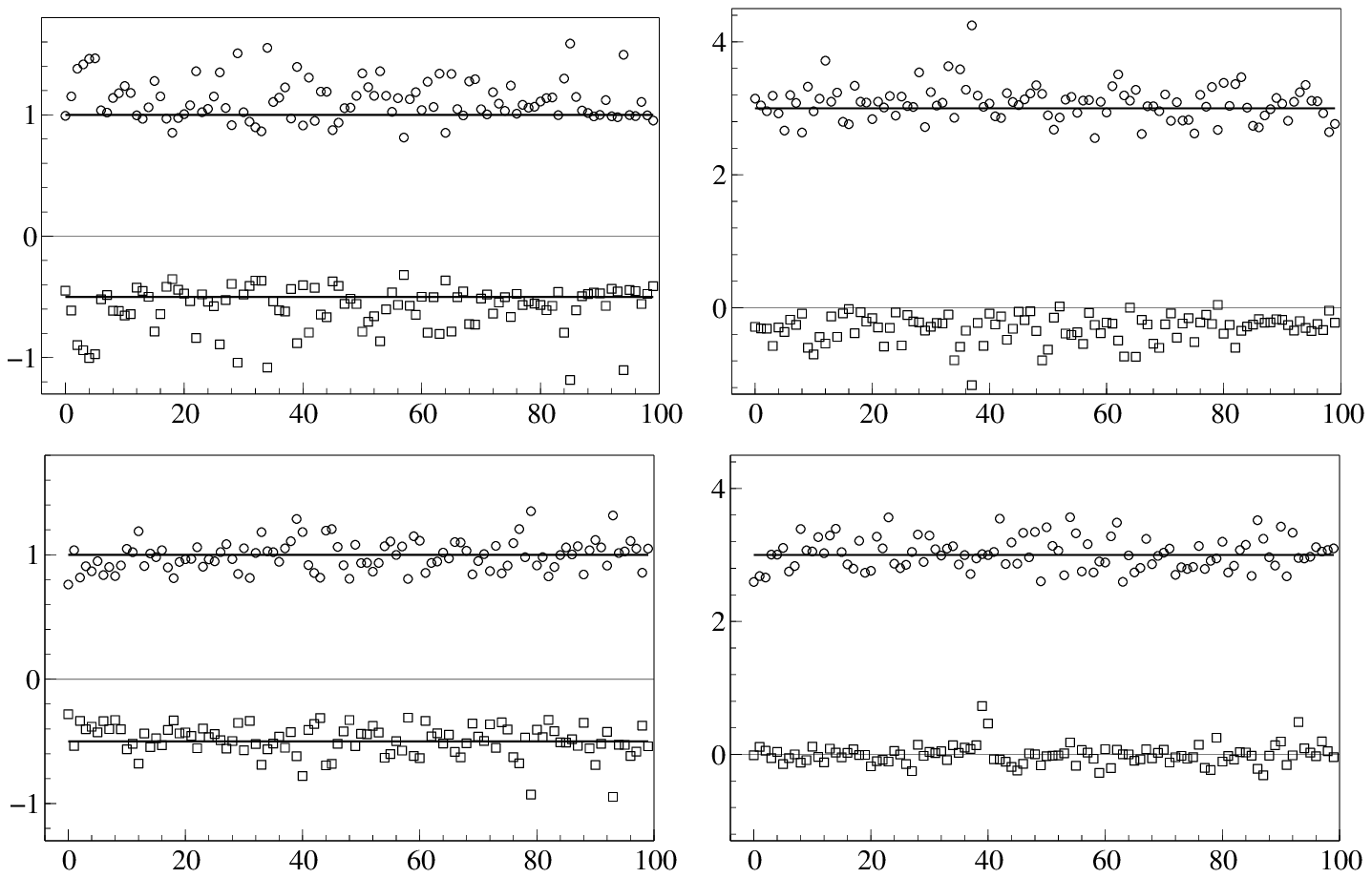}
\end{center}
\caption{Huang and Tauchen two factor model experiment for S\&P 500 data. Top to bottom: $T=300,$ $N=80$ (top), $T=2700,$ $N=700$ (bottom). Left to right:
mean (squares) and variance (circles) associated with $g_{N}(z\mid \theta)$ for $100$ different values of $\theta$ from $\pi(\theta)$ (left). The corresponding estimates of the third (squares) and fourth (circles) moments are displayed.}
\label{fig:HG_cond_momentsZ}
\end{figure}

It is important, in examining departures from Assumption 1, to consider the
heterogeneity of the conditional density $g_{N}(z\mid \theta )$ as $\theta $
varies over $\pi \left( \theta \right) $. In Fig. \ref{fig:AR_cond_momentsZ}%
, the conditional moments associated with the density $g_{N}(z\mid \theta )$
are estimated, based on running the particle filter independently $S=300$
times for each of $100$ values of $\theta $ from $\pi (\theta )$. We record
the estimates of the mean, the variance and the third and fourth central
moments at each value of $\theta $, for $T=300$ and $T=2700$. There is a
small degree of variability for $T=300$ around the values that we would
expect which are $-0.5,$ $1,0$ and $3$ corresponding to $g_{\textsc{z}%
}^{\sigma }\left( z\right) $ where $\sigma =1$. This variability reduces as $%
T$ rises to $2700$. A small degree variability is expected as these are
moments estimated from $S=300$ samples. This lack of heterogeneity explains
why the values of $Z$, marginalized over $\pi \left( \theta \right) $, on
the right hand side of Fig.~\ref{fig:AR1_hist}, are close to $g_{\textsc{z}%
}^{\sigma }\left( z\right) $ for time series of moderate and large length. Figure \ref{fig:HG_cond_momentsZ} records a similar experiment for the stochastic
volatility model and data considered in Section 4 of the main paper. There is rather
more variability as the true value of the likelihood in this case is
unknown and has to be estimated. However, the results are similar and the
variability again reduces as $T$ rises to $2700$.

\subsection{Empirical results for the pseudo-marginal algorithm\label%
{sec:full_AR1}}

The pseudo-marginal algorithm is applied to $T=300$ data. The true likelihood of the data is computed by the Kalman
filter as the model is a linear Gaussian state space model. This allows the
exact Metropolis--Hastings scheme $Q_{\textsc{ex}}$ to
be implemented so that the corresponding inefficiency $\textsc{\protect\small IF}\left(
h,Q_{\textsc{ex}}\right) $ can be easily estimated. We consider varying $N$ so that the standard
deviation $\sigma (\overline{\theta };N)$ of the log-likelihood estimator
varies. The grid of values that we consider for $N$ is $\{11$, $16$, $22$, $%
31$, $43$, $60$, $83$, $116$, $161$, $224$, $312\}$, see Table \ref%
{tab:AR1_CTIF}. The value $N=60$ results in $\sigma (\overline{\theta };N)=0.92
$.

\begin{table}[ptbh]
\captionsetup{font=small, labelsep=period}
\caption{\textsl{{AR(1) plus noise example with proposal parameter } $\rho=0,T=300,\phi=0.8,\mu=0.5,\sigma_{x}^{2}=1$ and $\sigma_{\varepsilon}^{2}=0.5$.
 Inefficiencies $(\textsc{{\protect \small IF}})$ and computing times $(\textsc{{\protect \small CT=$N$ $\times$ IF}})$ \textsl{shown for }$(\phi,\mu,\sigma_{x})$ and
marginal probabilities of acceptance. See Fig. \ref{fig:prA_AR1_allphi} and Fig.
\ref{fig:AR1_RCT}.}}%
\label{tab:AR1_CTIF}%
\begin{center}
\small
\begin{tabular}
[c]{lllllllll}%
\multicolumn{2}{l}{$Q_{\textsc{ex}}$} & IF$(\phi)$ & IF$(\mu)$ & IF$(\sigma
_{x})$ &  &  &  & pr$(Acc)$\\
&  & $2.5845$ & $2.5040$ & $2.4163$ &  &  &  & $0.7678$\\
\multicolumn{2}{l}{$Q$} &  &  &  &  &  &  & \\
$N$ & $\sigma(\overline{\theta};N)$ & IF$(\phi)$ & IF$(\mu)$ & IF$(\sigma_{x})$
& CT$(\phi)$ & CT$(\mu)$ & CT$(\sigma_{x})$ & pr$(Acc)$\\
$11$ & $2.2886$ & $136.32$ & $132.41$ & $128.66$ & $1499.5$ & $1456.5$ &
$1415.3$ & $0.11424$\\
$16$ & $1.8692$ & $61.403$ & $63.756$ & $66.609$ & $982.45$ & $1020.1$ &
$1065.7$ & $0.19036$\\
$22$ & $1.6063$ & $37.256$ & $40.486$ & $37.367$ & $819.63$ & $890.68$ &
$822.07$ & $0.25549$\\
$31$ & $1.3412$ & $15.880$ & $18.099$ & $19.135$ & $492.29$ & $561.08$ &
$593.20$ & $0.32622$\\
$43$ & $1.1096$ & $11.320$ & $9.7400$ & $10.710$ & $486.75$ & $418.82$ &
$460.54$ & $0.39347$\\
$60$ & $0.9197$ & $7.5040$ & $8.0428$ & $7.6168$ & $450.24$ & $482.57$ &
$457.01$ & $0.45933$\\
$83$ & $0.8058$ & $5.7253$ & $5.5841$ & $5.9348$ & $475.20$ & $463.48$ &
$492.59$ & $0.50885$\\
$116$ & $0.6828$ & $4.3756$ & $4.7106$ & $4.1693$ & $507.57$ & $546.43$ &
$483.63$ & $0.56621$\\
$161$ & $0.5828$ & $3.8112$ & $4.2379$ & $3.6388$ & $613.61$ & $682.30$ &
$585.84$ & $0.60160$\\
$224$ & $0.4838$ & $3.2711$ & $3.1605$ & $3.3134$ & $732.73$ & $707.94$ &
$742.19$ & $0.63562$\\
$312$ & $0.4096$ & $3.0774$ & $3.4768$ & $2.8355$ & $960.14$ & $1084.8$ &
$884.67$ & $0.65793$%
\end{tabular}
\end{center}
\end{table}

\begin{table}[ptbh]
\captionsetup{font=small, labelsep=period}
\caption{\textsl{{AR(1) plus noise example with proposal parameter }$\rho=0.9$.
Other settings identical to Table \ref{tab:AR1_CTIF}.}}%
\label{tab:AR1_CTIF_rho_pt9}%
\begin{center}
\small
\begin{tabular}
[c]{lllllllll}%
\multicolumn{2}{l}{$Q_{\textsc{ex}}$} & IF$(\phi)$ & IF$(\mu)$ & IF$(\sigma
_{x})$ &  &  &  & pr$(Acc)$\\
&  & $25.59$ & $22.21$ & $24.44$ &  &  &  & $0.87717$\\
\multicolumn{2}{l}{$Q$} &  &  &  &  &  &  & \\
$N$ & $\sigma(\overline{\theta};N)$ & IF$(\phi)$ & IF$(\mu)$ & IF$(\sigma_{x})$
& CT$(\phi)$ & CT$(\mu)$ & CT$(\sigma_{x})$ & pr$(Acc)$\\
$11$ & $2.2886$ & $594.64$ & $488.30$ & $639.04$ & $6541.1$ & $5371.3$ &
$7029.5$ & $0.12579$\\
$16$ & $1.8692$ & $157.49$ & $183.78$ & $182.07$ & $2519.9$ & $2940.4$ &
$2913.1$ & $0.20410$\\
$22$ & $1.6063$ & $126.87$ & $115.84$ & $125.37$ & $2791.2$ & $2548.6$ &
$2758.2$ & $0.27279$\\
$31$ & $1.3412$ & $69.541$ & $67.421$ & $71.982$ & $2155.9$ & $2089.9$ &
$2231.5$ & $0.35385$\\
$43$ & $1.1096$ & $53.053$ & $62.344$ & $58.002$ & $2281.1$ & $2680.9$ &
$2494.0$ & $0.42577$\\
$60$ & $0.9197$ & $49.351$ & $47.476$ & $45.194$ & $2961.1$ & $2848.6$ &
$2711.6$ & $0.49610$\\
$83$ & $0.8058$ & $37.709$ & $29.550$ & $38.266$ & $3129.8$ & $2452.7$ &
$3176.1$ & $0.55764$\\
$116$ & $0.6828$ & $29.360$ & $36.943$ & $34.892$ & $3405.8$ & $4285.4$ &
$4047.4$ & $0.61174$\\
$161$ & $0.5828$ & $28.277$ & $27.883$ & $29.864$ & $4552.6$ & $4489.2$ &
$4808.1$ & $0.65704$\\
$224$ & $0.4838$ & $27.770$ & $29.471$ & $30.533$ & $6220.5$ & $6601.5$ &
$6839.4$ & $0.69674$\\
$312$ & $0.4096$ & $29.231$ & $25.549$ & $29.967$ & $9120.2$ & $7971.4$ &
$9349.8$ & $0.73057$%
\end{tabular}
\end{center}
\end{table}

We transform each of the parameters to the real line so that $\Psi =k(\theta
)$, where both $\theta $ and $\Psi $ are three dimensional vectors, and
place a Gaussian prior on $\Psi $ centred at zero with a large variance. We
use the autoregressive Metropolis proposal $q(\Psi ,\Psi ^{\ast })$
\begin{equation*}
\Psi ^{\ast }=(1-\rho )\widehat{\Psi }+\rho \Psi +(1-\rho ^{2})^{1/2}\{(\nu
-2)/\nu \}^{1/2}\Sigma ^{1/2}t_{\nu },
\end{equation*}%
for both the pseudo-marginal algorithm and exact
likelihood schemes,
where $\widehat{\Psi }$ is the mode of the log-likelihood obtained from the
Kalman filter and the covariance $\Sigma $ is the negative inverse of the
second derivative of the log-likelihood at the mode.$\ $Here $t_{\nu }$
denotes a standard multivariate t-distributed random variable with $\nu $
degrees of freedom. We set $\nu =5$. We use this autoregressive proposal
with the scalar autoregressive parameter $\rho $ chosen as one of $\{0$, $0.4
$, $0.6$, $0.9\}$. We first apply this proposal, for the four values of $%
\rho $, using the known likelihood in the Metropolis scheme and estimate the
inefficiency for each of the parameters $\theta =(\mu _{x},\phi ,\sigma
_{\eta })$.

\begin{figure}[h]
\captionsetup{font=small}
\begin{center}
\includegraphics[height=2.5in, width=5.5in]{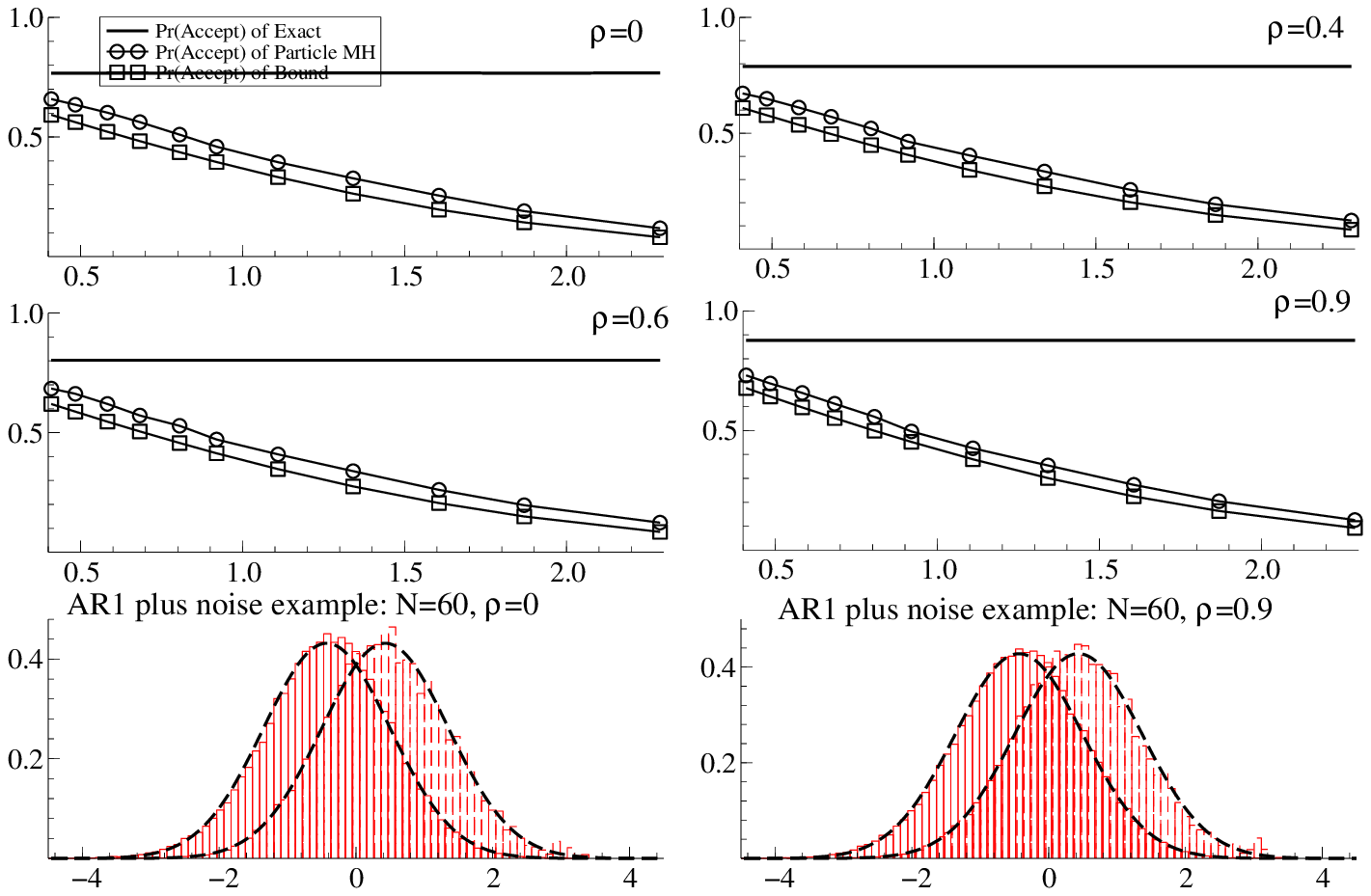}
\end{center}
\caption{
AR(1) plus noise example with $T=300,\phi=0.8,\mu=0.5,\sigma_{x}^{2}=1,\sigma_{\varepsilon}^{2}$ fixed at $0.5$. Marginal acceptance probabilities displayed against $\sigma(\overline{\theta};N)$. The estimated (constant) marginal
probabilities of acceptance for $Q_{\textsc{ex}}$ are
shown (solid line) together with the estimated probabilities (circles) from $Q$. The lower
bound (squares) is given as the probability from the exact scheme times $2\Phi
(-\sigma/\surd{2})$. The proposal autocorrelations are $\rho
=0,0.4,0.6$ and $0.9$. See Tables \ref{tab:AR1_CTIF} and \ref{tab:AR1_CTIF_rho_pt9}. Bottom: Histograms for the accepted
and proposed values of $Z$, the log-likelihood error for $\rho=0$ (left) and for $\rho=0.9$ (right). The
theoretical Gaussian densities for the proposal $\varphi (  -\sigma
^{2}/2,\sigma^{2})$  and
the accepted values $\varphi (  \sigma^{2}/2,\sigma^{2})$ are overlaid where $\sigma=\sigma(\overline{\theta};N)$.
}
\label{fig:prA_AR1_allphi}
\end{figure}

\begin{figure}[h]
\captionsetup{font=small}
\begin{center}
\includegraphics[height=2.5in,width=5.5in]{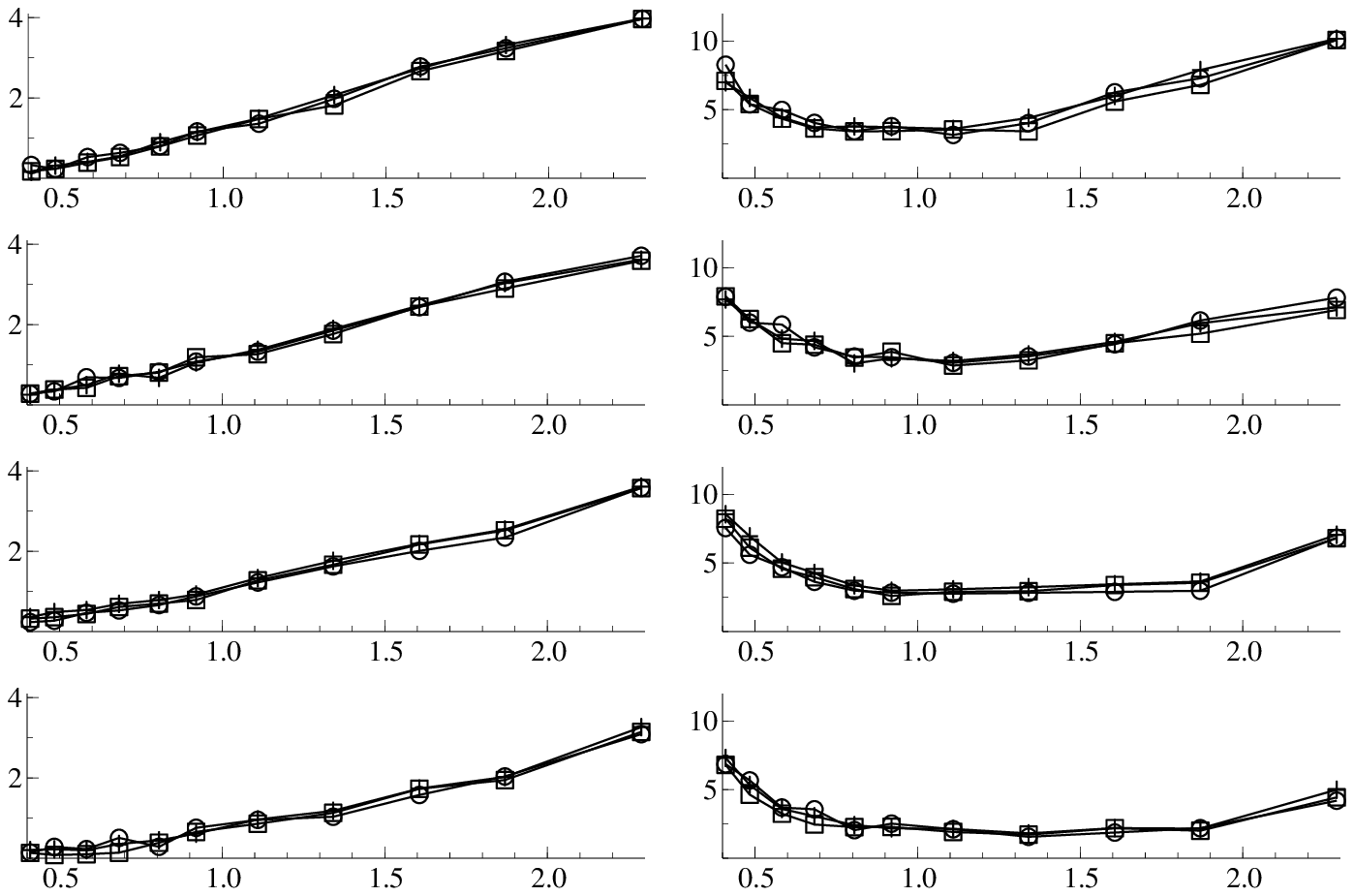}
\end{center}
\caption{AR(1) plus noise example with $T=300,$ $\protect\phi =0.8,\protect\mu %
=0.5,$ $\protect\sigma _{x}^{2}=1$ and $\protect\sigma _{\protect\varepsilon %
}^{2}$ fixed at $0.5$. Left: Logarithm $\textsc{rif}$ against $\protect%
\sigma (\overline{\protect\theta };N)$. Right: $\textsc{rct}$\textsl{\ }$=%
\textsc{rif}/\protect\sigma ^{2}(\overline{\protect\theta };N)$ against $%
\protect\sigma (\overline{\protect\theta };N)$. The three plots on all graphs
are for $\protect\phi $ (square), $\protect\mu $ (circle) and $\protect%
\sigma _{x}$ (cross). From Top to bottom: $\protect\rho =0,$ $0.4,0.6$ and $%
0.9$. Here $\protect\sigma (\overline{\protect\theta };N)$ is the standard
deviation of the log-likelihood estimator evaluated at the posterior mean $%
\overline{\protect\theta }$. See Tables \protect\ref{tab:AR1_CTIF} and
\protect\ref{tab:AR1_CTIF_rho_pt9}.}
\label{fig:AR1_RCT}
\end{figure}
Figure \ref{fig:prA_AR1_allphi} shows the acceptance probability for the
pseudo-marginal algorithm against $\sigma (\overline{\theta };N)$ for the four
values of the proposal parameter $\rho $. The lower bound for the acceptance
probabilities, as discussed at the end of Section 4.3 of the main paper, is
also displayed and there is close correspondence in all cases. The
histograms for the accepted and rejected values of $Z$, for $N=60$ when $%
\sigma (\overline{\theta };N)=0.92$, are also displayed. The approximating
asymptotic Gaussian densities, with $\sigma =0.92$, are superimposed. This
figure shows that the approximating densities correspond very closely to the
two histograms. It should be noted that these are the marginal values for $Z$
over the draws from the posterior $\pi (\theta )$ obtained by running the
pseudo-marginal scheme, rather than being based upon a fixed
value of the parameters.

Tables \ref{tab:AR1_CTIF} and \ref{tab:AR1_CTIF_rho_pt9} show the
pseudo-marginal algorithm results for $\rho =0$ and $\rho =0.9$. For the
independent Metropolis--Hastings proposal, it is clear from Table \ref%
{tab:AR1_CTIF}, that the computing time in minimised around $N=43$ or $60$,
depending on which parameter is examined, with the corresponding values of $%
\sigma (\overline{\theta };N)$ being $1.11$ and $0.92$, supporting the findings
that when an efficient proposal is used the optimal value of $\sigma $ is
close to unity. This is again supported by Fig. \ref{fig:AR1_RCT}, for
which the relative computing time ($\rho =0$ is the top right plot) is shown
against $\sigma (\overline{\theta };N)$. We note that the relative inefficiencies and computing times are straightforward to calculate as the exact chain inefficiencies for the three parameters have been calculated and are given in the top row of Table \ref{tab:AR1_CTIF}. Table \ref{tab:AR1_CTIF_rho_pt9} shows
the results for the more persistent proposal where $\rho =0.9$. In this
case, for all three parameters the optimal value of $N$ is around $31$, at
which $\sigma (\overline{\theta };N)=1.34$, the corresponding graph of the
relative computing time being given by the bottom right of Fig. \ref%
{fig:AR1_RCT}. It is clear that again the findings are consistent with the
discussion of 3.5 in the main paper. In particular, as $\rho $ increases,
then $\textsc{\protect\small IF}\left( h,Q^{\textsc{ex}}\right) $ should increase,
and, from Fig. \ref{fig:AR1_RCT}, it is clear that the optimal value of $%
\sigma (\overline{\theta };N)$ increases, the relative computing time decreases
for any given $\sigma (\overline{\theta };N)$. In addition, the relative computing time becomes more flat as
a function of $\sigma (\overline{\theta };N)$ as $\rho $ increases.

\bigskip

\section{Numerical procedures}\label{Sec:Numericalprocedures}
Under the Gaussian assumption, Corollary 3 specifies the function
$\varrho_{\textsc{z}}^{\sigma}$ and the term $\pi_{\textsc{z}}^{\sigma}\left(
1/\varrho_{\textsc{z}}^{\sigma}\right)  $ can be accurately evaluated using
numerical quadrature. This section explains how we numerically evaluate the
terms $\phi_{\textsc{z}}^{\sigma}$ and $\textsc{\protect\small IF}(1/\varrho_{\textsc{z}%
}^{\sigma},{\widetilde{Q}^{\textsc{z}}})$ which appear in the bounds of
Corollaries 1 and 2. The inefficiency  $\textsc{\protect\small IF}(1/\varrho
_{\textsc{z}}^{\sigma},{\widetilde{Q}^{\textsc{z}}})$ is finite by Lemma
3, because $\pi_{\textsc{z}}^{\sigma}\left(  1/\varrho
_{\textsc{z}}^{\sigma}\right)  $ is finite. The autocorrelations quickly descend to zero as a function of $n$, for all
$\sigma$. Hence, it is straightforward to estimate $\textsc{\protect\small IF}(1/\varrho_{\textsc{z}}^{\sigma
},{\widetilde{Q}^{\textsc{z}}})$ by the
appropriate summation of the autocorrelations, and to tabulate it against $\sigma$
for use in the bounds of Corollaries 1 and 2. The autocorrelation
$\phi_{\textsc{z}}^{\sigma}$, for $n=1$, is similarly tabulated.

From Lemma 3,
\[
\widetilde{Q}^{\textsc{z}}\left(  z,\mathrm{d}w\right)  =\frac{g\left(
w\right)  \min\{1,\exp(w-z)\}\mathrm{d}w}{\varrho_{\textsc{z}}\left(
z\right)  },\text{ \ }\widetilde{\pi}_{\textsc{z}}\left(  \mathrm{d}z\right)
=\frac{\pi_{\textsc{z}}\left(  \mathrm{d}z\right)  \varrho_{\textsc{z}}\left(
z\right)  }{\pi_{\textsc{z}}\left(  \varrho_{\textsc{z}}\right)  },
\]
so the autocorrelation at lag $n$ is%
\[
\phi_{n}\left(  \varrho_{\textsc{z}}^{-1},\widetilde{Q}^{\textsc{z}}\right)
=\frac{\left\langle \varrho_{\textsc{z}}^{-1},\left(  \widetilde{Q}%
^{\textsc{z}}\right)  ^{n}\varrho_{\textsc{z}}^{-1}\right\rangle
_{\widetilde{\pi}_{\textsc{z}}}{-}\left\{  \widetilde{\pi}_{\textsc{z}}\left(
\varrho_{\textsc{z}}^{-1}\right)  \right\}  ^{2}}{\mathbb{V}_{\widetilde{\pi
}_{Z}}\left(  \varrho_{\textsc{z}}^{-1}\right)  }%
\]
with
\begin{align}
\left\langle \varrho_{\textsc{z}}^{-1},\left(  \widetilde{Q}^{\textsc{z}%
}\right)  ^{n}\varrho_{\textsc{z}}^{-1}\right\rangle _{\widetilde{\pi
}_{\textsc{z}}} &  =%
{\textstyle\int}
\varrho_{\textsc{z}}^{-1}(z_{0})\varrho_{\textsc{z}}^{-1}(z_{n})\widetilde{\pi
}_{\textsc{z}}(\mathrm{d}z_{0})\left(  \widetilde{Q}^{\textsc{z}}\right)
^{n}\left(  z_{0},\mathrm{d}z_{n}\right)  \nonumber\\
&  =\pi_{\textsc{z}}(\varrho_{\textsc{z}})^{-1}%
{\textstyle\int}
\varrho_{\textsc{z}}^{-1}(z_{n})\widetilde{\pi}_{\textsc{z}}(\mathrm{d}%
z_{0})\left(  \widetilde{Q}^{\textsc{z}}\right)  ^{n}\left(  z_{0}%
,\mathrm{d}z_{n}\right)  .\label{eq:autocorrelationpho}%
\end{align}
The term $\pi_{\textsc{z}}(\varrho_{\textsc{z}})$ can be computed by
quadrature. The term (\ref{eq:autocorrelationpho}) can be also accurately
calculated by Monte Carlo integration, by simulating a large number $M$ of
i.i.d. samples $Z_{0}^{i}\sim\pi_{\textsc{z}}$ and then propagating each
sample through the transition kernel ${\widetilde{Q}^{\textsc{z}}}$ $n$ times
to obtain $Z_{n}^{i}\sim\pi_{\textsc{z}}(  \widetilde{Q}^{\textsc{z}%
})  ^{n}$, yielding the estimate $\frac{1}{M}\sum_{i=1}^{M}\varrho_{\textsc{z}}^{-1}(Z_{n}^{i}).$

\setlength{\bibsep}{0.3pt}


\end{document}